\documentclass[twoside, times, 11pt, journal, onecolumn]{IEEEtran}

\usepackage{epsfig}
\usepackage{graphicx}
\usepackage{amsmath}
\usepackage{amssymb}
\usepackage{float}
\usepackage{url}
\usepackage{verbatim}

\usepackage[bottom=0.8in]{geometry}

\newtheorem{definition}{Definition}
\newtheorem{lemma}{Lemma}
\newtheorem{proposition}{Proposition}
\newtheorem{corollary}{Corollary}
\newtheorem{theorem}{Theorem}

\newtheorem{remark}{Remark}
\newtheorem{example}{Example}
\newtheorem{discussion}{Discussion}
\newtheorem{question}{Question}

\newcommand{\naturals}{\ensuremath{\mathbb{N}}}
\newcommand{\reals}{\ensuremath{\mathbb{R}}}
\newcommand{\expectation}{\ensuremath{\mathbb{E}}}
\newcommand{\vect}{\mathbf}
\newcommand{\vnode}{{v}}
\newcommand{\cnode}{{c}}

\begin{document}

\title{\huge{On Universal Properties of Capacity-Approaching LDPC Code Ensembles}}

\author{\thanks{This research work was
supported by the Israel Science Foundation (grant no. 1070/07).}
Igal~Sason\thanks{Igal Sason is with the Department of Electrical
Engineering at the Technion -- Israel Institute of Technology,
Haifa 32000, Israel (e-mail:
sason@ee.technion.ac.il).} }

\maketitle

\begin{abstract}
This paper is focused on the derivation of some universal
properties of capacity-approaching low-density parity-check (LDPC)
code ensembles whose transmission takes place over memoryless
binary-input output-symmetric (MBIOS) channels. Properties of the
degree distributions, graphical complexity and the number of
fundamental cycles in the bipartite graphs are considered via the
derivation of information-theoretic bounds. These bounds are
expressed in terms of the target block/ bit error probability and
the gap (in rate) to capacity. Most of the bounds are general for
any decoding algorithm, and some others are proved under belief
propagation (BP) decoding. Proving these bounds under a certain
decoding algorithm, validates them automatically also under any
sub-optimal decoding algorithm. A proper modification of these
bounds makes them universal for the set of all MBIOS channels
which exhibit a given capacity. Bounds on the degree distributions
and graphical complexity apply to finite-length LDPC codes and to
the asymptotic case of an infinite block length. The bounds are
compared with capacity-approaching LDPC code ensembles under BP
decoding, and they are shown to be informative and are easy to
calculate. Finally, some interesting open problems are considered.
\end{abstract}


\section{Introduction}
\label{Section: Introduction} Low-density parity-check (LDPC)
codes form a class of powerful error-correcting codes which are
efficiently encoded and decoded with low-complexity algorithms.
These linear block codes, originally introduced by Gallager in the
early sixties \cite{Gallager_1962}, are characterized by sparse
parity-check matrices which facilitate their low-complexity
decoding with iterative message-passing algorithms. In spite of
the seminal work of Gallager, LDPC codes were ignored for a long
time. Following the breakthrough in coding theory, made by the
introduction of turbo codes \cite{BGT} and the rediscovery of LDPC
codes \cite{MN_LDPC96} in the mid 1990s, it was realized that an
efficient design of these codes enables to closely approach the
channel capacity while maintaining reasonable decoding complexity.
This breakthrough attracted many coding-theorists during the last
decade (see, e.g., \cite{Costello-Forney07}, \cite{RiU_book},
\cite{IT_Feb2001}).

The asymptotic analysis of LDPC code ensembles under iterative
message-passing decoding algorithms relies on the {\em density
evolution} (DE) approach which was developed by Richardson and
Urbanke (see \cite{Richardson1, Richardson2, RiU_book}). This
technique is commonly used for optimizing the degree distributions
of capacity-approaching LDPC code ensembles where the target is to
maximize the achievable rate for a given channel model or to
maximize the threshold for a given code rate subject to some
constraints on the degree distributions \cite{LTHC}.  Some
approximate techniques which optimize the degree distributions of
LDPC code ensembles under further practical constraints are of
interest (e.g., an optimization for obtaining a good tradeoff
between the asymptotic gap to capacity and the decoding complexity
\cite{ArdakaniSFK_Allerton05}). For the binary erasure channel
(BEC), the DE approach is much simplified since it leads to a
one-dimensional analysis. As a result of this significant
simplification, some explicit expressions for capacity-achieving
sequences of LDPC code ensembles have been derived for the BEC
(see, e.g., \cite{LubyMSS_IT01, Oswald-it02, RiU_book} and
\cite{Shokrollahi-IMA2000}). For general memoryless binary-input
output-symmetric (MBIOS) channels, as of yet there are no
closed-form expressions for capacity-achieving LDPC code ensembles
under iterative decoding, and the DE technique serves as a
numerical tool for the design of capacity-approaching LDPC code
ensembles in the limit where their block length tends to infinity.
Although maximum-likelihood (ML) decoding is prohibitively
complex, capacity-achieving sequences of LDPC code ensembles have
been constructed under ML decoding for any MBIOS channel where the
analysis relies on upper bounds on the decoding error probability
which are based on the distance spectra of these ensembles (see
\cite{Hsu_Achilleas1}, \cite{Hsu_Achilleas2}, \cite{Tutorial}, and
\cite[Theorem~2.2]{Sason-it03}).

Consider right-regular LDPC codes (i.e., LDPC codes where the
degree of the parity-check nodes is fixed to a certain value
$a_{\text{R}}$), and assume that their transmission takes place
over a binary symmetric channel (BSC). In his thesis, Gallager
derived an upper bound on the maximal achievable rate of these
codes where it is required to obtain vanishing block error
probability as we let the block length tend to infinity (see
\cite[Theorem~3.3]{Gallager_1962}). This information-theoretic
bound holds under ML decoding or any sub-optimal decoding
algorithm. This bound shows that right-regular LDPC codes cannot
achieve the channel capacity on a BSC, even under ML decoding.
Based on this bound, the inherent gap between the achievable rate
and the channel capacity is well approximated by an expression
which decreases to zero exponentially fast in $a_{\text{R}}$.
Burshtein {\em et al.} have generalized Gallager's bound for
general LDPC code ensembles whose transmission takes place over an
MBIOS channel \cite{Burshtein_IT2002}. An improved upper bound on
the achievable rates of LDPC code ensembles was obtained by
Wiechman and Sason \cite{Wiechman_Sason}, followed by a
generalization of this bound to the case where the transmission
takes place over a set of parallel MBIOS channels
\cite{Sason-it07}. This work partially relies on the analysis in
\cite{Wiechman_Sason} (see Section~\ref{Section: Preliminaries}
for relevant background).

Khandekar and McEliece suggested to measure the encoding and
decoding complexity of codes defined on graphs in terms of the
achievable gap (in rate) to capacity, and they also had some
conjectures regarding the behavior of the complexity as the gap to
capacity vanishes \cite{Khandekar-isit01}. Following their
approach, the tradeoff between the performance and complexity is
analyzed in the literature for LDPC code ensembles and some other
variants of codes defined on graphs (see, e.g.,
\cite{Hsu_Achilleas1}, \cite{Hsu_Achilleas2}, \cite{Pfister1},
\cite{Pfister2}, \cite{Sason-it03}, \cite{Sason-it07},
\cite{iterations}, \cite{Wiechman_Sason} and references therein).

In this paper, we consider some properties of capacity-approaching
LDPC code ensembles whose transmission takes place over MBIOS
channels. One question which is addressed in this paper is the
following: \vspace*{0.1cm}
\begin{question}
How do the degree distributions of capacity-approaching LDPC code
ensembles behave as a function of the achievable gap (in rate) to
capacity ?
\end{question}
\vspace*{0.1cm} The behavior of the degree distributions of
capacity-approaching LDPC code ensembles is addressed in this work
via the derivation of some information-theoretic bounds. Some of
them hold under ML decoding or any sub-optimal decoding algorithm,
and some other bounds are proved under belief propagation (BP)
decoding where we refer to the sum-product decoding algorithm (see
\cite{factor_graphs_IT01} and \cite[Chapter~2]{RiU_book}). For the
characterization of the degree distributions for
capacity-approaching LDPC code ensembles, a special consideration
is given to the fraction of degree--2 variable nodes ($L_2$) and
the fraction of edges connected to these nodes $(\lambda_2)$. This
focus was partially motivated by the influence of $\lambda_2$ on
the satisfiability of the stability condition; this condition is
necessary for achieving vanishing bit error probability under
iterative message-passing decoding when we let the block length
tend to infinity \cite{Richardson1}. Also, some previously
reported information-combining bounds on the performance of LDPC
code ensembles under iterative decoding are sensitive to this
quantity (see, e.g., \cite{Sutskover_IT07}). This motivates a
study of the behavior of $L_2$ and $\lambda_2$ for
capacity-approaching LDPC code ensembles, where the bounds on
these quantities are expressed in terms of the gap between the
channel capacity and the achievable rates of these code ensembles
under BP decoding. We also demonstrate the tightness of these
bounds for the BEC by considering the right-regular sequence of
capacity-achieving LDPC code ensembles proposed by Shokrollahi
\cite{Shokrollahi-IMA2000}.

General upper bounds on the degree distributions of
capacity-approaching LDPC code ensembles are derived in this paper
for the case where the transmission takes place over an MBIOS
channel. The bounds are expressed in terms of the gap (in rate) to
capacity with a target bit (or block) error probability. These
linear programming (LP) upper bounds on the degree distributions
of LDPC code ensembles are general with respect to the decoding
algorithm, and they also hold for ensembles of finite-length codes
or for the asymptotic case of an infinite block length. We note
that two LP problems are formulated in
\cite{Amraoui_LDPCopt_finite_length} for optimizing the degree
distributions of finite-length LDPC code ensembles whose
transmission takes place over a BEC, and also a convex
optimization problem is formulated in
\cite{ArdakaniSFK_Allerton05} for optimizing the degree
distributions of LDPC code ensembles with the goal of obtaining a
good tradeoff between performance and decoding complexity. It is
noted that the LP-based optimizations in
\cite{Amraoui_LDPCopt_finite_length} and
\cite{ArdakaniSFK_Allerton05} hold under BP decoding, whereas the
LP bounds which are derived in this paper are
information-theoretic bounds which hold under ML decoding or any
sub-optimal decoding algorithm. Although the degree distributions
of the parity-check nodes are often set to be regular (or almost
regular), and the irregularity often refers to the degree
distributions of the variable nodes, this is not necessarily the
case for capacity-approaching ensembles. For example,
\cite[Section~VI]{Pfister2} introduces some capacity-achieving
sequences of accumulate-repeat-accumulate code ensembles for the
BEC, which also possess a bounded complexity per information bit
under BP decoding; they are designed in a way where the degree
distributions of the LDPC code ensembles after a proper graph
reduction (as explained in \cite[Section~II]{Pfister2}) are
self-matched and are both irregular. The irregularity of the
parity-check degree distributions in the design of LDPC codes
appears to be useful in various cases under BP decoding, e.g., the
optimization of finite-length LDPC code ensembles whose
transmission takes place over the BEC
\cite{Amraoui_LDPCopt_finite_length}, the heavy-tail Poisson
distribution introduced in \cite{LubyMSS_IT01} and
\cite{Shokrollahi-IMA2000} which gives rise to capacity-achieving
degree distributions for the BEC, the design of bilayer LDPC code
ensembles for a degraded relay AWGN channel
\cite{irregular_check_node_degree_distribution_ISIT08}, and the
design of LDPC code ensembles for unequal error protection
\cite{irregular_check_node_degree_distribution_Turbo06}.

It is well known that linear block codes which are represented by
cycle-free bipartite (Tanner) graphs have poor performance even
under ML decoding \cite{Cycle-free codes}. The bipartite graphs of
capacity-approaching LDPC codes should have cycles. Hence, another
question which is addressed in this paper, as a continuation to a
previous study in \cite{Cycle-free codes} and \cite{Sason-it03}
(see also \cite[Problems~4.52 and 4.53]{RiU_book}), is the
following: \vspace*{0.1cm}
\begin{question}
How does the average cardinality of the fundamental system of
cycles of bipartite graphs behave as a function of the achievable
gap to capacity of the underlying LDPC code ensembles ?
\end{question}
\vspace*{0.1cm}

The fundamental tradeoff between the graphical complexity and
performance of codes defined on graphs is of interest, especially
for codes of finite-length. In this paper, we address the
following question: \vspace*{0.1cm}
\begin{question}
Consider the representation of a finite-length binary linear block
code by an arbitrary bipartite graph. How simple can such a
graphical representation be as a function of the channel model,
target block error probability, and code rate (which is below
capacity) ?
\end{question}
\vspace*{0.1cm}

We note that the graphical complexity referred to in this paper
measures the total number of edges used for the representation of
finite-length codes by bipartite graphs. By referring to the total
number of edges, the graphical complexity is strongly related to the
decoding complexity per iteration. This differs from the graphical
complexity in \cite{ArdakaniSFK_Allerton05}, \cite{Hsu_Achilleas1},
\cite{Pfister1} and \cite{Pfister2} which measures the number of
edges per information bit in the asymptotic case where we let the
block length tend to infinity. Although it may appear at first
glance that the aforementioned distinction is just a matter of
normalization, this is not the case: the reason is that given the
target block error probability and the required gap to capacity for
achieving this target with any finite-length block code, one needs
first to calculate the minimal block length which potentially allows
to fulfill these requirements. It is done in this work via the
calculation of classical and recent sphere-packing bounds (see
\cite{Shannon_1959}, \cite{SGB}, \cite{Valembois_Fossorier} and
\cite{ISP08}).

A universal design of LDPC code ensembles which enables these
codes to operate reliably over a multitude of channels is of great
theoretical and practical interest. We refer the reader to recent
studies on universal LDPC codes (see, e.g.,
\cite{universal_LDPC_Comm06},
\cite{universal_LDPC_Comm_Letters_06}, \cite{univeral_LDPC_ISIT08}
and \cite{universal_LDPC_IT07}). A simple modification of the
bounds derived in this paper makes them universal in the sense
that they hold for the set of MBIOS channels which exhibit a given
channel capacity. The universality of the bounds derived in this
paper stems also from the fact that they do not depend on the full
characterization of the LDPC code ensembles, but only on the gap
between the channel capacity and the design rates of these
ensembles, and they also depend on the target bit/ block error (or
erasure) probability. The bounds derived in this work are
expressed in closed form and are easily calculated.

This paper is structured as follows: Section~\ref{Section:
Preliminaries} provides some preliminary material and notation,
Section~\ref{Section: main results} introduces the new
information-theoretic bounds of this paper, Section~\ref{Section:
Proofs of Main Results} then provides their proofs followed by
some discussions, and Section~\ref{Section: Numerical Results}
formulates some algorithms related to the bounds derived in this
paper, it discusses their implications, and provides numerical
results. Finally, Section~\ref{Section: Outlook} summarizes this
work, and it provides some interesting open problems which are
related to this research.

\section{Preliminaries}
\label{Section: Preliminaries} We introduce in this section some
preliminary material and notation which serve for the analysis in
this paper.

\subsection{LDPC Code Ensembles}
\label{LDPC} LDPC codes are linear block codes which are
characterized by sparse parity-check matrices. A parity-check
matrix is represented by a bipartite graph where the variable and
parity-check nodes are on the left and right sides of this graph,
respectively. An edge connects a variable node with a parity-check
node in this graph if the corresponding parity-check equation
involves the code symbol which is represented by this variable
node (it is illustrated in Fig.~\ref{fig:tannergraph}). The
requirement for a sparse parity-check matrix is equivalent to the
requirement that the number of edges in the corresponding
bipartite graph scales linearly with the block length.

We move to consider ensembles of binary LDPC codes. Following
standard notation, let $\lambda_i$ and $\rho_i$ denote the fraction
of edges attached, respectively, to variable and parity-check nodes
of degree~$i$. Let $\Lambda_i$ and $\Gamma_i$ denote, respectively,
the fraction of variable and parity-check nodes of degree~$i$. The
LDPC code ensemble is characterized by a triple $(n,\lambda,\rho)$,
where $n$ designates the block length of the codes, and $\lambda(x)
\triangleq \sum_i \lambda_i x^{i-1}$ and $\rho(x) \triangleq \sum_i
\rho_i x^{i-1}$ represent, respectively, the left and right degree
distributions from the edge perspective. Equivalently, this ensemble
is also characterized by the triple $(n,\Lambda,\Gamma)$ where
$\Lambda(x) \triangleq \sum_i \Lambda_i x^i$ and $\Gamma(x)
\triangleq \sum_i \Gamma_i x^i$ represent, respectively, the left
and right degree distributions from the node perspective. We denote
by $\text{LDPC}(n,\lambda,\rho)$ (or
$\text{LDPC}(n,\Lambda,\Gamma)$) the ensemble whose bipartite graphs
are constructed according to the corresponding pairs of degree
distributions. The connections between the edges $\mathcal{E}$
emanating from the variable nodes to the parity-check nodes are
constructed by first numbering the connectors on the left and on the
right sides of the graph. The number of connectors is the same on
both sides of the graph, and it is equal to $ |\mathcal{E}| = n
\sum_i i \Lambda_i = m \sum_i i \Gamma_i $ where $n$ and $m$
designate the number of variable nodes and parity-check nodes,
respectively. Finally, the edges which connect the variable nodes
with the parity-check nodes of the bipartite graph are determined by
using a permutation $\pi: \{1, \ldots, |\mathcal{E}|\} \rightarrow
\{1, \ldots, |\mathcal{E}|\}$ which is chosen uniformly at random,
and associates connector number $i$ on the left side of this graph
with the connector whose number is $\pi(i)$ on the right. The degree
distributions with respect to the nodes and edges of a bipartite
graph are related via the following equations:
\begin{eqnarray}
&&\Lambda(x) =
\frac{\int_{0}^{x}\lambda(u)du}{\int_{0}^{1}\lambda(u)du} \; ,
\quad \quad \Gamma(x) =
\frac{\int_{0}^{x}\rho(u)du}{\int_{0}^{1}\rho(u)du}
\label{switching between representations_1} \\[0.2cm]
&&\lambda(x) = \frac{\Lambda'(x)}{\Lambda'(1)} \; ,
\hspace*{1.5cm} \rho(x) = \frac{\Gamma'(x)}{\Gamma'(1)}\;.
\label{switching between representations_2}
\end{eqnarray}
For an LDPC code ensemble, whose codes are represented by
parity-check matrices of dimension $m \times n$, the {\em design
rate} is defined as $R_{\text{d}} \triangleq 1 - \frac{m}{n}$.
This forms a lower bound on the rate of any code from this
ensemble, and the rate is equal to the design rate if the
particular parity-check matrix representing this code is full rank
(i.e., there are no redundant parity-check equations in this
matrix). The design rate is expressed in terms of the degree
distributions in the following two forms:
\begin{equation}
R_{\text{d}} = 1-\frac{\int_0^1\rho(x)dx}{\int_0^1\lambda(x)dx} =
1-\frac{\Lambda'(1)}{\Gamma'(1)}\;. \label{design rate of LDPC
ensemble}
\end{equation}
Note that
\begin{eqnarray}
a_{\mathrm{L}} = \Lambda'(1) = \frac{1}{\int_0^1\lambda(x)dx} \,
\label{eq: average left degree} \\ a_{\text{R}} = \Gamma'(1) =
\frac{1}{\int_0^1\rho(x)dx} \label{eq: average right degree}
\end{eqnarray}
designate the average left and right degrees (i.e., the average
degrees of the variable and parity-check nodes, respectively).

\subsection{Functionals Related to Memoryless Binary-Input
Output-Symmetric Channels} \label{subsection: functionals on MBIOS
Channels}

Consider an MBIOS channel whose channel input and channel output
are designated by $X$ and $Y$, respectively, and let
$p_{Y|X}(\cdot|\cdot)$ be its transition probability. The
associated log-likelihood ratio (LLR) $l(y)$ when the channel
output is $Y=y$ is given by
\begin{equation*}
l(y) = \ln \left(\frac{p_{Y|X}(y|0)}{p_{Y|X}(y|1)} \right).
\end{equation*}
The LLR associated with the random variable $Y$ is defined as
$L=l(Y)$. Let $a$ designate the conditional {\em pdf} of the
random variable $L$ given that the channel input is $X=0$ (to be
referred as the $L$-density function). This density function
satisfies the symmetry property $a(l) = e^l \, a(-l)$ for every $l
\in \reals$ \cite{Richardson2}.

This paper relies on the following two functionals (various other
functionals are presented in \cite[Section~4.1]{RiU_book}).

\begin{lemma}{\bf{[Capacity functional]}} Consider an MBIOS
channel whose symmetric $L$-density function is denoted by $a$.
Then the capacity of this channel in units of bits per channel
use, $C = C(a)$, is given by
\begin{equation}
C = \int_{-\infty}^{\infty} a(l) \bigl(1-\log_2(1+e^{-l}) \bigr)
\, \text{d}l.  \label{eq: channel capacity of an MBIOS channel}
\end{equation}
An equivalent form of the capacity is given by
\begin{equation}
C = \int_0^{\infty} a(l) (1+e^{-l}) \left( 1 -
h_2\Bigl(\frac{1}{1+e^l} \Bigr) \right) \, \text{d}l. \label{eq:
channel capacity 2 of an MBIOS channel}
\end{equation}
\end{lemma}
\vspace*{0.1cm} This lemma is proved in \cite[page~193]{RiU_book}.

\begin{definition}{\bf{[The Bhattacharyya functional]}} The
Bhattacharyya constant which is associated with the symmetric
$L$-density function $a$ is given by
\begin{equation}
\mathcal{B}(a) \triangleq \int_{-\infty}^{\infty} a(l)
e^{-\frac{l}{2}} \, \text{d}l. \label{eq: definition of
Bhattacharyya constant}
\end{equation}
\end{definition}

The analysis in this paper relies partially on the {\em stability
condition}. This condition applies to the asymptotic case where we
let the block length tend to infinity, and it forms a necessary
condition for successful decoding in the sense that it requires
that the fixed point of zero error rate be stable. Consider an
LDPC code ensemble with a given pair of degree distributions
$(\lambda, \rho)$ whose transmission takes place over an MBIOS
channel, characterized by its $L$-density function $a$. Then, the
stability condition under BP decoding gets the form (see
\cite[Theorem~4.125]{RiU_book})
\begin{equation}
\mathcal{B}(a) \lambda_2 \rho'(1) < 1. \label{eq: stability
condition}
\end{equation}
The reader is referred to \cite[Section~4.9]{RiU_book} for a
proof.

\subsection{Lower Bound on the Conditional Entropy for Binary
Linear Block Codes Transmitted over MBIOS Channels}
\label{subsection: Lower Bound on the Conditional Entropy for
Binary Linear Block Codes} We start this section by outlining in
Section~\ref{subsubsection: The analysis for a full-rank
parity-check matrix} the derivation of a lower bound on the
conditional entropy of the transmitted codeword given the received
sequence at the output of an MBIOS channel.
Section~\ref{subsubsection: The analysis for a full-rank
parity-check matrix} relies on \cite[Section~IV]{Wiechman_Sason}
and its appendices where it is assumed that the code is
represented by a full-rank parity-check matrix (the same
assumption is also made in \cite[Section~4.11]{RiU_book}).
Section~\ref{subsubsection: An adaptation of the analysis to LDPC
codes} revisits the derivation in Section~\ref{subsubsection: The
analysis for a full-rank parity-check matrix} in order to extend
the bound for the case where the binary linear block code is
represented by a parity-check matrix which is not necessarily
full-rank; this extension was hinted briefly in
\cite[Section~V]{Wiechman_Sason} (along the lines of the section
on numerical results), and we take this occasion to give a
rigorous proof which serves as a crucial preparatory step towards
the analysis in the continuation to this paper.

\subsubsection{The analysis for a full-rank parity-check matrix}
\label{subsubsection: The analysis for a full-rank parity-check
matrix} We assume in the following that the transmission of a
binary linear block code takes place over an MBIOS channel. Let
$\mathcal{C}$ be a binary linear block code of length $n$ and rate
$R$, and let $\mathbf{X}$ and $\mathbf{Y}$ be the transmitted
codeword and received sequence, respectively. Assume that the
codewords of $\mathcal{C}$ have no bits which are set a-priori to
zero. We assume that the code $\mathcal{C}$ is represented by a
parity-check matrix $H$ which is full rank. In the following, $C$
designates the capacity of the communication channel in units of
bits per channel use.
\begin{itemize}
\item{Define an equivalent channel whose output is the LLR of the original
channel}.
\item{The LLR is represented by a pair which includes its sign and absolute value.}
\item{For the characterization of the equivalent channel, let the function $a$ designate
the $L$-density function.}
\item{We randomly generate an i.i.d. sequence $\{L_i\}_{i=1}^n$ with respect to the $L$-density function $a$,
and define}
\begin{eqnarray*}
\Omega_i \triangleq |L_i|, \quad \Theta_i \triangleq \left\{
\begin{array}{ll}
               0  &\mbox{if $L_i > 0$} \\[0.1cm]
               1  &\mbox{if $L_i < 0$} \\[0.1cm]
               0 \; \text{or} \; 1 \; \text{equally likely} &\mbox{if $L_i = 0$}
               \end{array}
               \right. .
\end{eqnarray*}
Note that $\{\Theta_i\}$ is a sequence which represents the signs
of the LLR (conditioned on ${\bf{X}} = {\bf{0}}$).
\item{The output of the equivalent channel is
$\widetilde{\bf{Y}} = (\widetilde{Y}_1, \ldots, \widetilde{Y}_n$)
where
\begin{equation*}
\widetilde{Y}_i = (\Phi_i, \Omega_i), \quad i=1, \ldots, n
\end{equation*}
and $\Phi_i = \Theta_i + X_i$ (modulo-2 addition). This channel is
memoryless.}
\item{The output of this channel at time $i$ is
$\widetilde{Y}_i \in \{0,1\} \times \reals_{+}$. Note that $\Phi_i$
is a binary random variable which is affected by the channel input
$X_i$, and $\Omega_i$ is a non-negative random variable which is not
affected by $X_i$.}
\item{Due to the symmetry of the communication channel}, the {\em
pdf} of the absolute value of the LLR satisfies
\begin{equation*} \label{definition of f_Omega}
f_{\Omega}(\omega) = \left\{
\begin{array}{lr}
\hspace{-1mm}a(\omega) + a(-\omega) = (1+e^{-\omega}) \, a(\omega)  &\mbox{if $\omega > 0$,} \\[0.1cm]
\hspace{-1mm}a(0)  &\mbox{if $\omega=0$}. \\[0.1cm]
\end{array}
\right.
\end{equation*}
\end{itemize}
The conditional entropy of the transmitted codeword given the
received sequence at the output of the MBIOS channel satisfies
\begin{eqnarray}
&& \hspace*{-1.8cm} H(\mathbf{X}|\mathbf{Y}) = H(\mathbf{X}| \widetilde{\mathbf{Y}}) \nonumber \\
&& \hspace*{-0.4cm} = H(\mathbf{X}) +
H(\widetilde{\mathbf{Y}}|\mathbf{X}) - H(\widetilde{\mathbf{Y}}) \nonumber \\
&& \hspace*{-0.4cm} = nR + nH(\widetilde{Y}_1|X_1) - H(\widetilde{\mathbf{Y}}) \nonumber \\
&& \hspace*{-0.4cm} = nR + n[H(\widetilde{Y}_1) -
I(X_1;\widetilde{Y}_1)] - H(\widetilde{\mathbf{Y}}) \label{eq:
chain of equalities for the conditional entropy}
\end{eqnarray}
and
\begin{eqnarray}
&& I(X_1 ; \widetilde{Y}_1) = I(X_1; Y_1) \leq C \label{eq:
trivial upper bound on the mutual information} \\[0.2cm]
&& H(\widetilde{Y}_1) = H(\Phi_1, \Omega_1) \nonumber \\
&& \hspace*{1cm} = H(\Omega_1) + H(\Phi_1 | \Omega_1) \nonumber \\
&& \hspace*{1cm} = H(\Omega_1) + 1. \label{entropy of the scalar
random variable tilde Y}
\end{eqnarray}
The last transition in \eqref{entropy of the scalar random
variable tilde Y} is due to the fact that given the absolute value
of the LLR, its sign is equally likely to be positive or negative.
The entropy $H(\Omega_1)$ is not expressed explicitly as it will
cancel out.

The entropy of the vector $\widetilde{\bf{Y}}$ satisfies
\begin{eqnarray}
&& \hspace*{-1cm} H(\widetilde{\bf{Y}})= H\bigl( \Phi_1,
\Omega_1, \ldots, \Phi_n, \Omega_n \bigr) \nonumber \\
&&= H(\Omega_1, \ldots, \Omega_n) + H\bigl( \Phi_1, \ldots,
\Phi_n \; | \; \Omega_1, \ldots, \Omega_n \bigr) \nonumber \\
&&= n H(\Omega_1) + H\bigl( \Phi_1, \ldots, \Phi_n \; | \;
\Omega_1, \ldots, \Omega_n \bigr). \label{first step}
\end{eqnarray}
\begin{itemize}
\item Define the syndrome vector
$ {\bf{S}} \triangleq (\Phi_1, \ldots, \Phi_n) H^T $. Since $H$ is
assumed to be a full-rank parity-check matrix of $\mathcal{C}$
then ${\bf{S}} \in \{0, 1\}^{n(1-R)}$, i.e., the syndrome
${\bf{S}}$ is composed of $n(1-R)$ binary components.
\item Let $M$ be the index of
the vector $(\Phi_1, \ldots, \Phi_n)$ in the coset which
corresponds to the syndrome $ {\bf{S}} $.
\item $H(M) = nR$ since all the codewords are transmitted with equal probability, and we get
{\setlength\arraycolsep{0pt}
\begin{align}
&H\bigl( \Phi_1, \ldots, \Phi_n \, | \,
\Omega_1, \ldots, \Omega_n \bigr) \nonumber\\[0.1cm]
&= H({\bf{S}}, M \, | \, \Omega_1, \ldots,
\Omega_n \bigr) \nonumber \\[0.1cm]
&\leq H(M) + H\bigl({\bf{S}} \, | \, \Omega_1, \ldots, \Omega_n
\bigr) \nonumber \\[0.1cm]
&\leq nR + \sum_{j=1}^{n(1-R)} H\bigl(S_j \, | \, \Omega_1,
\ldots, \Omega_n \bigr)\,. \label{second step}
\end{align}}
\item Since ${\bf{X}}H^T={\bf{0}}$ for every codeword ${\bf{X}} \in \mathcal{C}$,
and also $\Phi_i = X_i + \Theta_i$ for all $i$, then $ {\bf{S}} =
(\Theta_1, \ldots, \Theta_n) H^T $ is independent of the
transmitted codeword.
\end{itemize}
Combining \eqref{eq: chain of equalities for the conditional
entropy}--\eqref{second step} gives
\begin{equation}
H(\mathbf{X}|\mathbf{Y}) \geq n(1-C) - \sum_{j=1}^{n(1-R)}H(S_j
\big{|} \Omega_1, \ldots, \Omega_n) \label{eq: new inequality for
the conditional entropy}
\end{equation}
where
\begin{itemize}
\item $S_j = 1 $ if and only if $\Theta_i =1$ for an odd number of
indices $i$ in the $j$-th parity-check equation.
\item Due to the symmetry of the channel
\begin{eqnarray*}
&& P(\alpha_i) \triangleq \text{Prob}(\Theta_i = 1 \big{|}
\Omega_i = \alpha_i) \\
&& \hspace*{1cm} = \frac{a(-\alpha_i)}{a(\alpha_i) + a(-\alpha_i)}
= \frac{1}{1+e^{\alpha_i}}.
\end{eqnarray*}
\end{itemize}
In order to calculate the conditional entropy of a single
component of the syndrome, the following lemma is used:
\begin{lemma}
If the $j$-th component of the syndrome $\mathbf{S}$ involves $k$
variables whose indices are $\{i_1,\ldots,i_k\}$ then
\begin{eqnarray*}
&& \text{Prob}(S_j= 1 \big{|} \Omega_{i_1} = \alpha_1, \ldots, \Omega_{i_k} = \alpha_k) \\
&&
=\frac{1}{2}\Bigr[1-\prod_{m=1}^{k}\bigl(1-2P(\alpha_m)\bigr)\Bigr]
\end{eqnarray*}
where
\begin{equation*}
1-2P(\alpha)=\tanh\left(\frac{\alpha}{2}\right).
\end{equation*}
\end{lemma}
The proof of this lemma follows from
\cite[Lemma~4.1]{Gallager_1962}.
\begin{itemize}
\item For a parity-check node of degree $k$, the conditional entropy
$H(S_j \big{|} \Omega_1, \ldots, \Omega_n)$ is equal to the
$k$-dimensional integral
\begin{eqnarray*}
&& \int_0^{\infty} \ldots \int_0^{\infty} h_2
\Biggl(\frac{1}{2}\biggr[1-\prod_{m=1}^{k} \tanh
\bigl(\frac{\alpha_m}{2}\bigr) \biggr] \Biggr) \; \prod_{m=1}^k
f_{\Omega}(\alpha_m) \; d\alpha_1 \ldots d\alpha_k
\end{eqnarray*}
where $f_{\Omega}$ is the {\em pdf} of the absolute value of the
LLR, and $h_2$ is the binary entropy function to the base~2.
\item Using the following Taylor series expansion of $h_2$:
\begin{equation}
h_2(x) = 1 - \frac{1}{2 \ln 2} \sum_{p=1}^{\infty}
\frac{(1-2x)^{2p}}{p(2p-1)}, \quad 0 \leq x \leq 1 \label{eq:
power series for h_2}
\end{equation}
then, for a parity-check node of degree~$k$, the above
$k$-dimensional integral is transformed to the following infinite
sum of one-dimensional integrals (see
\cite[Appendix~II]{Wiechman_Sason}):
\begin{eqnarray}
&& \hspace*{-1cm} H(S_j \big{|} \Omega_1, \ldots, \Omega_n) \nonumber \\
&& \hspace*{-1cm} = 1 - \frac{1}{2 \ln 2} \sum_{p=1}^{\infty}
\biggl\{ \frac{1}{p(2p-1)} \cdot \left( \int_0^{\infty} a(l)
(1+e^{-l}) \tanh^{2p}\Bigl( \frac{l}{2} \Bigr) \, \text{d}\l
\right)^k \biggr\} \, . \label{eq: lower bound on conditional
entropy for a parity-check node of degree k}
\end{eqnarray}
\end{itemize}
\vspace*{1mm} For an arbitrary full-rank parity-check matrix of a
binary linear block code $\mathcal{C}$, let $\Gamma_k $ designate
the fraction of the parity-checks involving $k$ variables, and let
$\Gamma(x) \triangleq \sum_{k}\Gamma_k x^k$. The combination of
\eqref{eq: new inequality for the conditional entropy} and
\eqref{eq: lower bound on conditional entropy for a parity-check
node of degree k} leads to the following lower bound on the
conditional entropy of the transmitted codeword given the received
sequence at the channel output:
\begin{equation}
\frac{H(\vect{X}|\vect{Y})}{n} \geq R - C + \frac{1-R}{2\ln 2}\;
\sum_{p=1}^{\infty} \frac{\Gamma(g_p)}{p(2p-1)} \label{eq: Lower
bound on conditional entropy}
\end{equation}
where
\begin{equation}
g_p \triangleq \int_{0}^{\infty} a(l) (1+e^{-l})
\tanh^{2p}\left(\frac{l}{2}\right) dl, \quad p \in \naturals.
\label{eq: definition of g_k}
\end{equation}
The above lower bound on the conditional entropy holds for any
representation of the code by a full-rank parity-check matrix. The
symmetry condition for MBIOS channels states that $a(l) = e^{\, l}
a(-l)$ for all $l \in \reals$, and therefore \eqref{eq: definition
of g_k} gives that
\begin{equation}
g_p =
\expectation\left[\tanh^{2p}\left(\frac{L}{2}\right)\right],\quad
p\in\naturals \label{eq: alternative definition of g_k}
\end{equation}
where $\expectation$ designates the statistical expectation with
respect to the $L$-density function $a$, and $L$ is a random
variable which stands for the LLR at the output of the channel
given that the input bit is zero. Eq.~\eqref{eq: alternative
definition of g_k} implies that the non-negative sequence
$\{g_p\}_{p \geq 1}$ is monotonically non-increasing and it only
depends on the communication channel (but not on the code). Note
also that, from \eqref{eq: alternative definition of g_k}, $0 \leq
g_p < 1$ for all $p \in \naturals$ (unless the channel is perfect,
which then implies that $g_p=1$ for all values of $p$).

We note that the conditional entropy on the LHS of \eqref{eq:
Lower bound on conditional entropy} depends only on the code and
the communication channel, but its lower bound on the RHS of
\eqref{eq: Lower bound on conditional entropy} depends also on the
specific representation of the code by a bipartite graph.

The lower bound in \eqref{eq: Lower bound on conditional entropy}
improves the bound in \cite[Eq.~(15)]{Burshtein_IT2002}, except
for the binary symmetric channel (BSC) where they both coincide.
The reason is that the derivation of \eqref{eq: Lower bound on
conditional entropy} relies on the un-quantized soft output of the
channel whereas the derivation of the bound in
\cite[Eq.~(15)]{Burshtein_IT2002} relies on a two-level
quantization of this output (which therefore does not loosen the
bound for a BSC).

\vspace*{0.1cm}
\subsubsection{An adaptation of the analysis to LDPC codes which
are not necessarily represented by full-rank parity-check matrices}
\label{subsubsection: An adaptation of the analysis to LDPC codes}
The derivation of the lower bound in \eqref{eq: Lower bound on
conditional entropy} relies on the assumption that the parity-check
matrix is full rank. Though it seems like a feasible requirement for
specific binary linear block codes, this poses a problem when
considering ensembles of LDPC codes. In the latter case, a
parity-check matrix which corresponds to a randomly chosen bipartite
graph with a given pair of degree distributions may not be full
rank.\footnote{A concentration of the code rate to the design rate
of LDPC code ensembles is proved asymptotically (for an infinite
block length) under some conditions (see \cite{Miller} and
\cite[Lemma~3.22]{RiU_book}). However, we are interested in a lower
bound on the conditional entropy which also holds for finite-length
binary linear block codes regardless of this asymptotic
concentration property.} To this end, we present the following
lemma:
\begin{lemma}
For (regular and irregular) ensembles of binary LDPC codes, the
inequality in \eqref{eq: Lower bound on conditional entropy} stays
valid for every code from the ensemble with the following
modifications:
\begin{itemize}
\item The rate $R$ of the code is replaced with the design rate ($R_{\text{d}}$)
of the ensemble.
\item The sequence $\{\Gamma_k\}$ denotes the degree distribution of the parity-check
nodes of the ensemble (where the representation of a code by a
parity-check matrix, with the given degree distribution, possibly
includes some linearly dependent rows).
\end{itemize}
\label{lemma: extension of lower bound on the conditional entropy}
\end{lemma}
\begin{proof}
See Appendix~\ref{Appendix: extension of lower bound on the
conditional entropy}.
\end{proof}

\subsection{Sphere-Packing Bounds}
Sphere-packing bounds are commonly used for the study of the
performance limitations of finite-length error-correcting codes
over memoryless symmetric channels. For a tutorial on classical
sphere-packing bounds, the reader is referred to
\cite[Chapter~5]{Tutorial}. This paper relies on the following
sphere-packing bounds (see Section~\ref{Numerical Results for
Finite-Length Analysis}):
\begin{itemize}
\item The {\em SP59 bound}: The 1959 sphere-packing (SP59) bound of Shannon
\cite{Shannon_1959} serves for the evaluation of the performance
limits of block codes whose transmission takes place over an AWGN
channel. This lower bound on the decoding error probability is
expressed in terms of the block length and the rate of the code;
however, it does not take into account the modulation used, but
only assumes that the modulated signals have equal energy. It is
often used as a reference for quantifying the sub-optimality of
error-correcting codes under some practical decoding algorithms
(see \cite[Chapter~5]{Tutorial} and references therein). An
efficient algorithm for the calculation of the SP59 bound is
introduced in \cite[Section~IV.C]{ISP08}.
\item The {\em ISP bound}: This sphere-packing bound was recently
derived in \cite[Section~III]{ISP08}. The ISP bound applies to all
memoryless symmetric channels. For codes of finite block length,
it improves the classical sphere-packing bound of Shannon,
Gallager and Berlekamp \cite{SGB} and the sphere-packing bound of
Valembois and Fossorier \cite{Valembois_Fossorier} where this
improvement is especially pronounced for short to moderate block
lengths. We note that the ISP bound in \cite{ISP08} is not
uniformly tighter than the SP59 bound for equi-energy signals
transmitted over an AWGN channel.
\end{itemize}
Comparisons between the sphere-packing bounds in
\cite{Shannon_1959}, \cite{Valembois_Fossorier} and
\cite[Section~III]{ISP08} are shown in \cite[Section~V]{ISP08}.
\label{sub-section: sphere-packing bounds}

\subsection{Cycles in Graphs}
\label{subsection: elements from graph theory}

We consider in this paper the cycles in bipartite graphs which
represent capacity-approaching LDPC code ensembles. To this end,
we define and exemplify some notions which are relevant to the
analysis in this paper.

\begin{definition}{\bf[Cycle and cycle length]} A {\em cycle} in
an un-directed graph is a closed path. The length of a cycle is
the number of edges on this closed path. The {\em girth} of an
un-directed graph is defined as the shortest length of its cycles.
\end{definition}

\begin{definition}{\bf{[Tree]}} A {\em tree} is a connected graph that has no
cycles. \label{definition: tree}
\end{definition}
From Definition~\ref{definition: tree}, a removal of any edge from
a tree makes the graph disconnected. An important property of
trees is that any two vertices are connected by a single path.

Every graph $\mathcal{G}$ has subgraphs that are trees. This
motivates the following definition:
\begin{definition}{\bf{[Spanning tree]}}
A {\em spanning tree} of a connected graph $\mathcal{G}$ is a tree
which spans all the vertices of $\mathcal{G}$. Note that by
repeatedly removing edges which originally create cycles in the
graph, it follows that every connected graph has a spanning tree.
\label{definition: spanning tree}
\end{definition}

\begin{definition}{\bf{[Number of components of a graph]}}
Let $\mathcal{G}$ be a possibly disconnected graph. The {\em number
of components} of $\mathcal{G}$ is the minimal number of its
connected subgraphs whose union forms the graph $\mathcal{G}$
(clearly, a connected graph has a single component).
\label{definition: number of components of a graph}
\end{definition}

\begin{definition}{\bf{[Cycle rank]}} Let $\mathcal{G}$ be an
un-directed graph with $|V_{\mathcal{G}}|$ vertices,
$|E_{\mathcal{G}}|$ edges and $C(\mathcal{G})$ components. The
{\em cycle rank} of $\mathcal{G}$, denoted by
$\beta(\mathcal{G})$, is defined as the maximal number of edges
which can be removed from the graph without increasing its number
of components (note that each component becomes a spanning tree
after the removal of these edges). \label{definition: cycle rank}
\end{definition}
From Definition~\ref{definition: cycle rank}, the cycle rank of a
graph is a measure of the edge redundancy with respect to the
connectedness of this graph. The cycle rank satisfies the
following equality (see
\cite[p.~154]{Graph_theory_and_applications}):
\begin{equation}
\beta(\mathcal{G}) = |E_{\mathcal{G}}| - |V_{\mathcal{G}}| +
C(\mathcal{G}). \label{eq: cycle rank}
\end{equation}

\begin{definition}{\bf{[Full spanning forest]}}
Let $\mathcal{G}$ be an un-directed graph. A {\em full spanning
forest} $\mathcal{F}$ of the graph $\mathcal{G}$ is the subgraph of
$\mathcal{G}$ that results from removing the $\beta(\mathcal{G})$
edges from Definition~\ref{definition: cycle rank}. Clearly, the
number of components of $\mathcal{F}$ and $\mathcal{G}$ is the same.
Note that a graph may have a multiplicity of full spanning forests.
\label{definition: full spanning forest}
\end{definition}

\begin{definition}{\bf{[Fundamental cycle]}}
Let $\mathcal{F}$ be a full spanning forest of an un-directed
graph $\mathcal{G}$, and let $e$ be an edge in the relative
complement of $\mathcal{F}$. The cycle of the subgraph
$\mathcal{F}\cup\{e\}$ (whose existence and uniqueness is
guaranteed by
\cite[Theorem~3.1.11]{Graph_theory_and_applications}) is called a
{\em fundamental cycle} of $\mathcal{G}$ which is associated with
$\mathcal{F}$. \label{definition: fundamental cycle}
\end{definition}

\begin{remark}
Each of the edges in the relative complement of a full spanning
forest $\mathcal{F}$ gives rise to a {\em different} fundamental
cycle of the graph $\mathcal{G}$. \label{remark: on the different
fundamental cycles}
\end{remark}

\begin{definition}{\bf{[Fundamental system of cycles]}}
The {\em fundamental system of cycles} of a graph $\mathcal{G}$
which is associated with a full spanning forest $\mathcal{F}$ is
the set of all fundamental cycles of $\mathcal{G}$ associated with
$\mathcal{F}$.
\end{definition} \label{definition: fundamental system of cycles}

\begin{remark}
From Remark~\ref{remark: on the different fundamental cycles}, the
cardinality of the fundamental system of cycles of $\mathcal{G}$
associated with a full spanning forest of this graph is equal to
the cycle rank $\beta(\mathcal{G})$. \label{remark: equality
between the fundamental system of cycles and cycle rank}
\end{remark}

\begin{example}{\bf{[Fundamental system of cycles in a bipartite graph]}}
This example refers to the bipartite graph in
Fig.~\ref{fig:tannergraph}. This graph is connected, but it is
clearly not a tree. As an example, consider the cycle $ \langle
v_9, c_4, v_{10}, c_5, v_9 \rangle$ whose length is~4. Since the
number of vertices in this graph is~15 and the number of its edges
is~30, then from \eqref{eq: cycle rank}, the cycle rank of this
connected graph is $30-15+1=16$.

\begin{figure}[htp]
\begin{center}
\parbox[c]{5.7cm}{
\[
\arraycolsep=2pt H:=
\begin{array}{crcccccccccccccl}
 &  \text{\scriptsize 1} &  \text{\scriptsize 2} &  \text{\scriptsize 3} &  \text{\scriptsize 4} &
 \text{\scriptsize 5} &  \text{\scriptsize 6} &  \text{\scriptsize 7} &  \text{\scriptsize \bf 8}
 &  \text{\scriptsize 9} &  \text{\scriptsize 10}\\
\begin{array}{@{}c@{}}
\text{\scriptsize 1} \\
\text{\scriptsize \bf 2} \\
\text{\scriptsize 3} \\
\text{\scriptsize 4} \\
\text{\scriptsize 5} \\
\end{array}
& \left( \vphantom{\begin{array}{@{}c@{}} 0 \\ 0 \\ 0 \\ 0 \\ 0
\end{array}} \right.
\begin{array}{@{}c@{}}
1 \\ {\bf 0} \\ 0 \\ 1 \\ 1
\end{array}
&
\begin{array}{@{}c@{}}
1 \\ {\bf 0} \\ 1 \\ 0 \\ 1
\end{array}
&
\begin{array}{@{}c@{}}
1 \\ {\bf 1} \\ 0 \\ 1 \\ 0
\end{array}
&
\begin{array}{@{}c@{}}
1 \\ {\bf 1} \\ 1 \\ 0 \\ 0
\end{array}
&
\begin{array}{@{}c@{}}
0 \\ {\bf 1} \\ 0 \\ 1 \\ 1
\end{array}
&
\begin{array}{@{}c@{}}
1 \\ {\bf 1} \\ 1 \\ 0 \\ 0
\end{array}
&
\begin{array}{@{}c@{}}
1 \\ {\bf 1} \\ 0 \\ 0 \\ 1
\end{array}
&
\begin{array}{@{}c@{}}
{\bf 0} \\ {\bf 1} \\ {\bf 1} \\ {\bf 1} \\ {\bf 0}
\end{array}
&
\begin{array}{@{}c@{}}
0 \\ {\bf 0} \\ 1 \\ 1 \\ 1
\end{array}
&
\begin{array}{@{}c@{}}
0 \\ {\bf 0} \\ 1 \\ 1 \\ 1
\end{array}
& \left. \vphantom{\begin{array}{@{}c@{}} 0 \\ 0 \\ 0 \\ 0 \\ 0
\end{array}} \right)
\end{array}
\]
}
\parbox[c]{1cm}{\phantom{xxxx}}
\parbox[c]{4.5cm}{\input{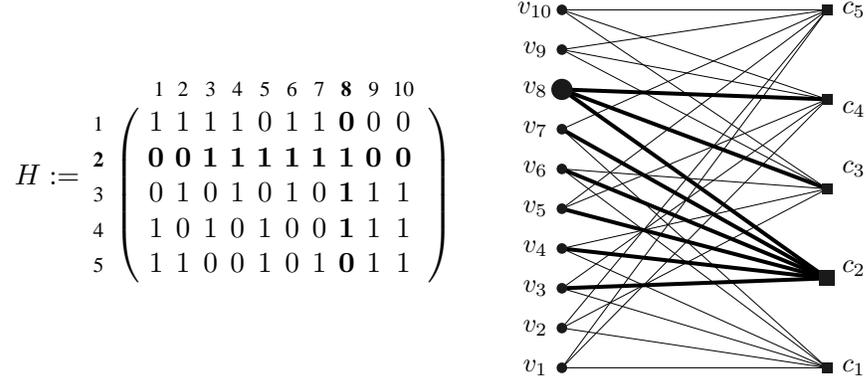}}\\[0.8cm]
\caption{\label{fig:tannergraph} A parity-check matrix $H$ and the
corresponding bipartite graph. For illustrating this relationship,
column~8 and row~2 of $H$ are bolded; the corresponding variable
and parity-check nodes, and the attached edges are also bolded
(this figure appears in \cite{comm03}).}
\end{center}
\end{figure}

In order to get a spanning tree of the graph in
Fig.~\ref{fig:tannergraph}, we remove repeatedly 16~edges which
create cycles while preserving the connectivity of the graph.

\begin{figure}[here!]
\begin{center}
\parbox[c]{5.7cm}{
\[
\arraycolsep=2pt \widetilde{H}=
\begin{array}{crcccccccccccccl}
 &  \text{\scriptsize 1} &  \text{\scriptsize 2} &  \text{\scriptsize 3} &  \text{\scriptsize 4} &
 \text{\scriptsize 5} &  \text{\scriptsize 6} &  \text{\scriptsize 7} &  \text{\scriptsize 8}
 &  \text{\scriptsize 9} &  \text{\scriptsize 10}\\
\begin{array}{@{}c@{}}
\text{\scriptsize 1} \\
\text{\scriptsize 2} \\
\text{\scriptsize 3} \\
\text{\scriptsize 4} \\
\text{\scriptsize 5} \\
\end{array}
& \left( \vphantom{\begin{array}{@{}c@{}} 0 \\ 0 \\ 0 \\ 0 \\ 0
\end{array}} \right.
\begin{array}{@{}c@{}}
1 \\ 0 \\ 0 \\ 1 \\ 1
\end{array}
&
\begin{array}{@{}c@{}}
1 \\ 0 \\ {\bf 0} \\ 0 \\ {\bf 0}
\end{array}
&
\begin{array}{@{}c@{}}
1 \\ 1 \\ 0 \\ {\bf 0} \\ 0
\end{array}
&
\begin{array}{@{}c@{}}
{\bf 0} \\ 1 \\ {\bf 0} \\ 0 \\ 0
\end{array}
&
\begin{array}{@{}c@{}}
0 \\ {\bf 0} \\ 0 \\ 1 \\ {\bf 0}
\end{array}
&
\begin{array}{@{}c@{}}
{\bf 0} \\ 1 \\ {\bf 0} \\ 0 \\ 0
\end{array}
&
\begin{array}{@{}c@{}}
{\bf 0} \\ 1 \\ 0 \\ 0 \\ {\bf 0}
\end{array}
&
\begin{array}{@{}c@{}}
0 \\ 1 \\ {\bf 0} \\ {\bf 0} \\ 0
\end{array}
&
\begin{array}{@{}c@{}}
0 \\ 0 \\ 1 \\ {\bf 0} \\ 1
\end{array}
&
\begin{array}{@{}c@{}}
0 \\ 0 \\ {\bf 0} \\ {\bf 0} \\ 1
\end{array}
& \left. \vphantom{\begin{array}{@{}c@{}} 0 \\ 0 \\ 0 \\ 0 \\ 0
\end{array}} \right)
\end{array}
\]
}
\parbox[c]{0.7cm}{\phantom{xxxx}}
\caption{\label{fig:new parity-check matrix} A parity-check matrix
which corresponds to a spanning tree of the bipartite graph in
Fig.~\ref{fig:tannergraph}. As compared to the parity-check matrix
$H$ in Fig.~\ref{fig:tannergraph}, the new parity-check matrix
$\widetilde{H}$ is obtained by changing the values of the bolded
16~entries from~$1$ to~$0$.}
\end{center}
\end{figure}
The parity-check matrix $\widetilde{H} = [\widetilde{h}_{i,j}]$ in
Fig.~\ref{fig:new parity-check matrix}, with 16~bolded zero
entries which correspond to the removed edges from the original
graph in Fig.~\ref{fig:tannergraph}, represents a spanning tree of
this graph. To exemplify its connectivity, note that the variable
nodes $v_5$ and $v_6$ are connected by the path $\langle v_6, c_2,
v_3, c_1, v_1, c_4, v_5 \rangle$ which is of length~6. This path
can be observed directly from the parity-check matrix
$\widetilde{H} = [\widetilde{h}_{i,j}]$ by alternate horizontal
and vertical moves through the ones of $\widetilde{H}$;
explicitly, this path is determined by a horizontal move from
$\widetilde{h}_{2,6}$ to $\widetilde{h}_{2,3}$, a vertical move to
$\widetilde{h}_{1,3}$, a horizontal move to $\widetilde{h}_{1,1}$,
a vertical move to $\widetilde{h}_{4,1}$ and finally a horizontal
move to $\widetilde{h}_{4,5}$. In a similar way, it can be
verified that every two vertices in the bipartite graph of
$\widetilde{H}$ are connected, and it spans all the 15~vertices of
the graph in Fig.~\ref{fig:tannergraph} (since there is no row or
column in $\widetilde{H}$ which is a zero vector). Hence, this
graph is indeed a spanning tree of the bipartite graph in
Fig.~\ref{fig:tannergraph}. This spanning tree enables to obtain a
set of~16 fundamental cycles by returning back a single bolded
zero in Fig.~\ref{fig:new parity-check matrix} (among its 16
bolded zeros) to~1. For example, by setting $\widetilde{h}_{1,6} =
1$ (which is equivalent to returning the edge which connects $v_6$
with $c_1$), we get the fundamental cycle $\langle v_3, c_2, v_6,
c_1, v_3 \rangle$. \label{Example for fundamental cycles}
\end{example}

\subsection{Notation}
We consider in this paper sequences of capacity-approaching LDPC
code ensembles, and refer to the case where the fractional gap (in
rate) to capacity $(\varepsilon)$ vanishes. Accordingly, based on
standard notation \cite{Big-O notation}, we define
\begin{itemize}
\item $f(\varepsilon) = O\bigl(g(\varepsilon)\bigr)$ means that there are
positive constants $c$ and $\delta$, such that $0 \leq
f(\varepsilon) \leq c \; g(\varepsilon)$ for all $0 \leq
\varepsilon \leq \delta$.
\item $f(\varepsilon) = \Omega\bigl(g(\varepsilon)\bigr)$ means that there are positive constants
$c$ and $\delta$, such that $0 \leq c \; g(\varepsilon) \leq
f(\varepsilon)$ for all $0 \leq \varepsilon \leq \delta$.
\end{itemize}
Note that the values of $c$ and $\delta$ must be fixed, and should
not depend on~$\varepsilon$.

Throughout the paper
\begin{equation*}
h_2(x) \triangleq -x \log_2(x) - (1-x) \log_2(1-x), \quad 0 \leq x
\leq 1
\end{equation*}
denotes the binary entropy function to the base~2, and
$h_2^{-1}:[0,1] \rightarrow \left[0, \frac{1}{2}\right]$ is the
inverse of the restriction of $h_2$ to $\left[0,
\frac{1}{2}\right]$. We also denote the block error probability
and the bit error probability of a code by $P_{\text{B}}$ and
$P_{\text{b}}$, respectively (for the BEC, the error probability
is replaced with an erasure probability). Note that $P_{\text{b}}$
refers to the bit error probability of the information bits.

This paper is focused on the analysis for MBIOS channels. For
basic definitions and examples of MBIOS channels, the reader is
referred to \cite[Section~4.1]{RiU_book} (which uses a slightly
different abbreviation: BMS channels).

For further notation used throughout this paper,
Section~\ref{LDPC} provides the setting and notation for the
degree distributions and the design rate of LDPC code ensembles,
Section~\ref{subsection: functionals on MBIOS Channels} provides
the notation for the capacity and Bhattacharyya functionals,
Section~\ref{subsection: Lower Bound on the Conditional Entropy
for Binary Linear Block Codes} presents the notation for the lower
bound on the conditional entropy (see \eqref{eq: Lower bound on
conditional entropy}--\eqref{eq: alternative definition of g_k}),
and Section~\ref{subsection: elements from graph theory} provides
the terminology and notation used here in the context of cycles in
bipartite graphs.

\section{New Information-Theoretic Bounds}
\label{Section: main results} This section introduces
information-theoretic bounds which are related to the degree
distributions, graphical complexity, and the number of fundamental
systems of cycles in the bipartite graphs of LDPC code ensembles.
\begin{theorem}{\bf{[On the average degree of the parity-check nodes]}}  Let
$\mathcal{C}$ be a binary linear block code of block length $n$
whose transmission takes place over an MBIOS channel. Let
$\mathcal{G}$ be a bipartite graph which corresponds to a
full-rank parity-check matrix of $\mathcal{C}$. Let $C$ designate
the capacity of the channel, in bits per channel use, and $a$ be
the $L$-density function of this channel. Assume that the code
rate is (at least) a fraction $1-\varepsilon$ of the channel
capacity (where $0 < \varepsilon < 1$), and the code achieves a
block error probability $P_{\text{B}}$ or a bit error probability
$P_{\text{b}}$ under some decoding algorithm. Then, the average
right degree of the bipartite graph (i.e., the average degree of
the parity-check nodes in $\mathcal{G}$) satisfies
\begin{equation}
a_{\text{R}} \geq \frac{2\ln\left(\frac{1}{1-2
h_2^{-1}\Bigl(\frac{1-C-\delta}
{1-(1-\varepsilon)C}\Bigr)}\right)}{\ln\Bigl(\frac{1}{g_1}\Bigr)}
\label{eq: lower bound on a_R with finite P_b}
\end{equation}
where $g_1$ is given in \eqref{eq: definition of g_k} (and it
depends only on the channel), and
\begin{eqnarray}
\delta \triangleq \left\{ \begin{array}{ll}
                  P_{\text{B}} + \frac{h_2(P_{\text{B}})}{n} &
                  \mbox{for a block error probability
                  $P_{\text{B}}$} \\
                  h_2(P_{\text{b}}) & \mbox{for a bit error probability
                  $P_{\text{b}}$}
                  \end{array}
                  \right. .  \label{delta}
\end{eqnarray}
Furthermore, among all the MBIOS channels which exhibit a given
capacity $C$ and for which a target block error probability
$(P_{\text{B}})$ or a bit error probability $(P_{\text{b}})$ is
obtained under some decoding algorithm, a universal lower bound on
$a_{\text{R}}$ holds by replacing $g_1$ on the RHS of \eqref{eq:
lower bound on a_R with finite P_b} with $C$.

For the BEC, the following tightened version of \eqref{eq: lower
bound on a_R with finite P_b} holds:
\begin{equation}
a_{\text{R}} \geq \frac{\ln\left(1+\frac{p-
P_{\text{b}}}{(1-p)\varepsilon+P_{\text{b}}}
\right)}{\ln\left(\frac{1}{1-p}\right)} \label{eq: lower bound on
a_R for the BEC with finite P_b}
\end{equation}
where $p$ is the erasure probability of the channel, and
$P_{\text{b}}$ is the bit erasure probability at the decoder.
\label{Theorem: Lower bound on right degree}
\end{theorem}

\vspace*{0.2cm}
\begin{remark}{\bf{[The relation of Theorem~\ref{Theorem: Lower bound on right
degree} to the bound in \cite{Wiechman_Sason}]}} In the particular
case where $P_{\text{b}}$ vanishes, the bound in \eqref{eq: lower
bound on a_R with finite P_b} forms a tightened version of the
bound given in \cite[Eq.~(77)]{Wiechman_Sason}. This point is
clarified in Discussion~\ref{Discussion: on the proof of lower
bound on average right degree} which succeeds the proof of
Theorem~\ref{Theorem: Lower bound on right degree} (see
page~\pageref{Discussion: on the proof of lower bound on average
right degree}). In the limit where the gap (in rate) to capacity
vanishes (and with vanishing $P_{\text{b}}$), the lower bounds on
the average right degree in \eqref{eq: lower bound on a_R with
finite P_b} and \cite[Eq.~(77)]{Wiechman_Sason} both grow like the
logarithm of the inverse of this gap, and they therefore possess
the same asymptotic behavior where
\begin{equation}
a_{\text{R}} \triangleq a_{\text{R}}(\varepsilon) = \Omega \left(
\ln \frac{1}{\varepsilon} \right). \label{eq: asymptotic behavior
of the average right degree}
\end{equation}
However, in spite of the similarity in the asymptotic behavior of
the two lower bounds as $\varepsilon \rightarrow 0$, they may
differ significantly even for rather small values of $\varepsilon$
(see Example~\ref{example: comparison of numerical results for the
new and old lower bounds} on p.~\pageref{example: comparison of
numerical results for the new and old lower bounds}).

Theorem~\ref{Theorem: Lower bound on right degree} also provides a
universal lower bound on the average right degree for the set of
all MBIOS channels with a given capacity $C$. This theorem states
the conditions where the bound in \eqref{eq: lower bound on a_R
with finite P_b} gets its extreme values among all MBIOS channels
which exhibit a given capacity. \label{Remark: connection with
previous lower bound on average right degree}
\end{remark}

\begin{remark}{\bf{[Adaptation of Theorem~\ref{Theorem: Lower bound on right
degree} to LDPC code ensembles]}} As is clarified in
Discussion~\ref{discussion: extension of Theorem 1 to LDPC
ensembles} (see page~\pageref{discussion: extension of Theorem 1
to LDPC ensembles}), Theorem~\ref{Theorem: Lower bound on right
degree} can be adapted to hold for an arbitrary ensemble of
$(n,\lambda,\rho)$ LDPC codes. In this case, the requirement of a
full-rank parity-check matrix of a particular code $\mathcal{C}$
from this ensemble is relaxed by requiring that the design rate of
the LDPC code ensemble is equal to a fraction $1-\varepsilon$ of
the channel capacity. In this case, $P_{\text{b}}$ and
$P_{\text{B}}$ stand for the average bit and block error (or
erasure) probabilities of the ensemble under some decoding
algorithm. \label{remark: extension of Theorem 1 to LDPC
ensembles}
\end{remark}

\begin{remark}{\bf[The graphical complexity of finite-length LDPC
codes]} In Section~\ref{Numerical Results for Finite-Length
Analysis}, we apply Theorem~\ref{Theorem: Lower bound on right
degree} and sphere-packing bounds on the decoding error
probability (see \cite{Shannon_1959}, \cite{SGB},
\cite{Valembois_Fossorier}, \cite{ISP08}) to obtain
information-theoretic lower bounds on the graphical complexity of
finite-length LDPC codes. These bounds are expressed as a function
of the target block error probability and the gap between the
design rate of the code and the channel capacity. We note that in
this context, the graphical complexity measures the number of
edges used for the representation of finite-length codes by
bipartite graphs. By referring to the total number of edges, the
graphical complexity is strongly related to the decoding
complexity per iteration. The bounds are compared with
capacity-approaching LDPC code ensembles under BP decoding, and
they are shown to be informative (see Section~\ref{Numerical
Results for Finite-Length Analysis}).
\end{remark}

\vspace*{0.1cm} Based on Remark~\ref{remark: extension of Theorem
1 to LDPC ensembles} and the background which is provided in
Section~\ref{subsection: elements from graph theory}, the
following result is derived:

\begin{corollary}{\bf{[On the asymptotic average cardinality of the fundamental system
of cycles of LDPC code ensembles]}} Let $\bigl\{\bigl(n, \lambda,
\rho\bigr)\bigr\}$ be a sequence of LDPC code ensembles whose
transmission takes place over an MBIOS channel. Let the design
rate of these ensembles be a fraction $1-\varepsilon$ of the
channel capacity $C$, and assume that the average bit error/
erasure probability of this sequence vanishes under some decoding
algorithm as we let the block length $(n)$ tend to infinity.
Consider the average cardinality of the fundamental system of
cycles in bipartite graphs from the LDPC code ensemble $(n,
\lambda, \rho)$ where the graphs are chosen uniformly at random
(from Remark~\ref{remark: equality between the fundamental system
of cycles and cycle rank}, the cardinality of the fundamental
system of cycles in a graph $\mathcal{G}$ is equal to its cycle
rank $\beta(\mathcal{G})$). Then, the following asymptotic lower
bound holds:
\begin{eqnarray}
&& \hspace*{-1cm} \liminf_{n \rightarrow \infty}
\frac{\expectation_{\text{LDPC}(n,\lambda,\rho)}\bigl[\beta(\mathcal{G})\bigr]}{n}
\nonumber \\ && \hspace*{-1cm} \geq \frac{(1-C) \; \ln\biggl(g_1
\; \left[1-2 h_2^{-1}\Bigl(\frac{1-C}
{1-(1-\varepsilon)C}\Bigr)\right]^{-2}\biggr)}{\ln\Bigl(\frac{1}{g_1}\Bigr)}-1
\label{eq: lower bound on the cardinality of the fundamental
system of cycles of LDPC ensembles for general MBIOS channels}
\end{eqnarray}
where $g_1$ is introduced in \eqref{eq: definition of g_k}. For a
BEC whose erasure probability is $p$, a tightened bound gets the
form:
\begin{equation}
\liminf_{n \rightarrow \infty}
\frac{\expectation_{\text{LDPC}(n,\lambda,\rho)}\bigl[\beta(\mathcal{G})\bigr]}{n}
\geq \frac{p \; \ln\left(1-p+\frac{p}{\varepsilon}
\right)}{\ln\left(\frac{1}{1-p}\right)}-1. \label{eq: tightened
lower bound on the cardinality of the fundamental system of cycles
of LDPC ensembles for the BEC}
\end{equation}
\label{corollary: lower bound on the cardinality of the
fundamental system of cycles of LDPC ensembles}
\end{corollary}

\begin{remark}
Corollary~\ref{corollary: lower bound on the cardinality of the
fundamental system of cycles of  LDPC ensembles} provides two
results which are of the type $\Omega \left( \ln
\frac{1}{\varepsilon} \right)$.
\end{remark}

\begin{theorem}{\bf{[On the degree distributions of
capacity-approaching LDPC code ensembles]}} Let $\bigl(n,\lambda,
\rho\bigr)$ (or $(n, \Lambda, \Gamma)$) be an ensemble of LDPC
codes whose transmission takes place over an MBIOS channel. Assume
that the design rate of the ensemble is equal to a fraction
$1-\varepsilon$ of the channel capacity $C$, and let
$P_{\text{b}}$ designate the average bit error (or erasure)
probability of the ensemble under ML decoding or any sub-optimal
decoding algorithm. Then, the following properties hold for an
arbitrary finite (and fixed) degree~$i$
\begin{eqnarray}
&& \Lambda_i(\varepsilon) = O(1) \label{eq: behavior of Lambda_i}\\[0.1cm]
&& \Gamma_i(\varepsilon) = O\bigl(\varepsilon
C+h_2(P_{\text{b}})\bigr) \label{eq: behavior of Gamma_i}\\[0.1cm]
&& \lambda_i(\varepsilon) = O \left(\frac{1}{\ln
\frac{1}{\varepsilon C +h_2(P_{\text{b}})}}\right) \label{eq:
behavior
of lambda_i}\\[0.1cm] && \rho_i(\varepsilon) = O\left(\frac{\varepsilon C+h_2(P_{\text{b}})}{\ln
\frac{1}{\varepsilon C+h_2(P_{\text{b}})}}\right)\,.\label{eq:
behavior of rho_i}
\end{eqnarray}
For the case where the transmission takes place over the BEC, the
bounds above are tightened by replacing $h_2(P_{\text{b}})$ with
$P_{\text{b}}$. \label{Theorem: Behavior of left and right
degrees}
\end{theorem}

\begin{remark}{\bf{[On the connection between Theorems~\ref{Theorem: Lower bound on right degree}
and~\ref{Theorem: Behavior of left and right degrees}]}}
Theorem~\ref{Theorem: Behavior of left and right degrees} implies
that for every capacity-approaching LDPC code ensemble whose bit
error probability vanishes and also for an arbitrary finite
degree~$i$ in their bipartite graphs, the fraction of edges
attached to variable nodes or parity-check nodes of degree~$i$
tends to zero as the gap to capacity $(\varepsilon)$ vanishes.
This conclusion is consistent with Theorem~\ref{Theorem: Lower
bound on right degree} which states that the average left and
right degrees of the bipartite graphs scale at least like $\ln
\frac{1}{\varepsilon}$; hence, these average degrees necessarily
become unbounded as the gap to capacity vanishes.
\end{remark}

\begin{corollary} Under the assumptions of Theorem~\ref{Theorem: Behavior of left
and right degrees}, if the asymptotic bit error/ erasure
probability vanishes then the following properties hold for an
arbitrary finite degree~$i$
\begin{eqnarray*}
&& \Lambda_i = O(1) \, , \quad \quad \quad \; \Gamma_i = O(\varepsilon)\, , \\
&& \lambda_i = O \left(\frac{1}{\ln \frac{1}{\varepsilon}}\right)
\, , \quad \rho_i = O \left(\frac{\varepsilon}{\ln
\frac{1}{\varepsilon}}\right).
\end{eqnarray*}
\label{Corollary: Behavior of the left and right degrees when
P_b=0}
\end{corollary}

\begin{remark}{\bf{[Linear programming upper bounds on the degree
distributions of LDPC code ensembles]}} Theorem~\ref{Theorem:
Behavior of left and right degrees} and Corollary~\ref{Corollary:
Behavior of the left and right degrees when P_b=0} provide
asymptotic results for the degree distributions of LDPC code
ensembles in the limit where the gap to capacity vanishes (i.e.,
$\varepsilon \rightarrow 0$). Section~\ref{LP bounds on degree
distributions of LDPC code ensembles} provides linear programming
(LP) upper bounds on the degree distributions which are expressed
in terms of the target average bit error probability, and the
(possibly non-zero) gap between the channel capacity and the
design rate of the ensemble for achieving this target. Similarly
to Theorem~\ref{Theorem: Behavior of left and right degrees} and
Corollary~\ref{Corollary: Behavior of the left and right degrees
when P_b=0}, the LP bounds in Section~\ref{LP bounds on degree
distributions of LDPC code ensembles} hold under ML decoding, and
are therefore general in terms of the decoding algorithm. We note
that these LP bounds apply to finite-length LDPC code ensembles
and to the asymptotic case of an infinite block length. Analytical
solutions for these LP bounds are provided in Section~\ref{LP
bounds on degree distributions of LDPC code ensembles}, and these
bounds are also compared with some capacity-achieving sequences of
LDPC code ensembles for the BEC under BP decoding. Additional LP
bounds are derived to hold for the set of all the MBIOS channels
which exhibit a given capacity, and that also achieve a target bit
error probability. These universal LP bounds are compared with the
LP bounds which refer to specific MBIOS channels (see
Section~\ref{LP bounds on degree distributions of LDPC code
ensembles}).
\end{remark}

We turn now our attention to sequences of LDPC code ensembles
which asymptotically achieve vanishing bit error probability under
BP decoding. The following theorem gives upper bounds on the
fraction of degree-2 variable nodes $(\Lambda_2)$ and the fraction
of edges attached to these nodes $(\lambda_2)$ for an arbitrary
sequence of LDPC code ensembles whose transmission takes place
over an MBIOS channel. It relies on information-theoretic
arguments and the stability condition. We note that $\lambda_2$ is
involved in the stability condition (see \eqref{eq: stability
condition}). Moreover, some previously reported
information-combining bounds on the performance of LDPC code
ensembles under BP decoding are sensitive to the value of
$\lambda_2$ (see, e.g., \cite{Sutskover_IT07}).
\begin{theorem}{\bf{[On the fraction of degree-2 variable nodes
and the fraction of edges attached to these nodes for LDPC code
ensembles]}} Let $\bigl\{\bigl(n_m, \lambda(x),
\rho(x)\bigr)\bigr\}_{m \geq 1}$ be a sequence of LDPC code
ensembles whose transmission takes place over an MBIOS channel.
Assume that this sequence asymptotically achieves a fraction
$1-\varepsilon$ of the channel capacity under BP decoding with
vanishing bit error probability. Then, the fraction of degree-2
variable nodes satisfies
\begin{eqnarray} && \Lambda_2 < \frac{1-C}{2 \, \mathcal{B}(a)} \left(1+\frac{\varepsilon
C}{1-C}\right) \nonumber \\ && \hspace*{0.8cm} \cdot \left[ 1 +
\frac{\ln\left(\frac{1}{g_1}\right)}{\ln\Biggl(\frac{g_1}{\bigl[1-2h_2^{-1}
\bigl(\frac{1-C}{1-(1-\varepsilon)C}\bigr)\bigr]^2}\Biggr)}
\right] \label{eq: upper bound on fraction of degree-2 variable
nodes}
\end{eqnarray}
and the fraction of edges attached to these nodes satisfies
\begin{equation}
\lambda_2 < \frac{\ln\left(\frac{1}{g_1}\right)}{\mathcal{B}(a) \,
\ln\biggl(\frac{g_1}{\bigl[1-2h_2^{-1}
\bigl(\frac{1-C}{1-(1-\varepsilon)C}\bigr)\bigr]^2}\biggr)}
\label{eq: bound on lambda_2 for MBIOS channels}
\end{equation}
where the Bhattacharyya constant $\mathcal{B}(a)$ and the
parameter $g_1$ are introduced in \eqref{eq: definition of
Bhattacharyya constant} and \eqref{eq: definition of g_k},
respectively. Consider the set of all the MBIOS channels with a
given capacity $C$ and a Bhattacharyya constant $\mathcal{B}(a)$,
for which the bit error probability vanishes under BP decoding.
Then, universal upper bounds on $\Lambda_2$ and $\lambda_2$ hold
for this set of channels by replacing $g_1$ on the RHS of
\eqref{eq: upper bound on fraction of degree-2 variable nodes} and
\eqref{eq: bound on lambda_2 for MBIOS channels}, respectively,
with~$C$.

For a BEC with an erasure probability $p$, the following tightened
bounds hold:
\begin{equation}
\Lambda_2 < \frac{1}{2} \left(1+\frac{\varepsilon (1-p)}{p}\right)
\left[1 + \frac{\ln\left(\frac{1}{1-p}\right)}{\ln\left(1 - p +
\frac{p}{\varepsilon}\right)} \right] \label{eq: upper bound on
fraction of degree-2 variable nodes for the BEC}
\end{equation}
and
\begin{equation}
\lambda_2 < \frac{\ln\left(\frac{1}{1-p}\right)}{p \,
\ln\left(1-p+\frac{p}{\varepsilon}\right)}. \label{eq: bound on
lambda_2 for the BEC}
\end{equation}
\label{Theorem: degree-2 variable nodes}
\end{theorem}

\begin{corollary}
Under the assumptions of Theorem~\ref{Theorem: degree-2 variable
nodes}, in the limit where the gap to capacity vanishes under BP
decoding (i.e., $\varepsilon \rightarrow 0$), the fraction of
degree-2 variable nodes satisfies
\begin{equation}
\Lambda_2 \leq \frac{1-C}{2 \, \mathcal{B}(a)}\ \label{eq: upper
bound on degree-2 variable nodes for capacity-achieving LDPC
ensembles}
\end{equation}
where this upper bound is necessarily not larger than~$\frac{1}{2}$.
Note that this forms a universal upper bound on the fraction of
degree-2 variable nodes for all MBIOS channels with a given capacity
$C$ and a Bhattacharyya constant $\mathcal{B}(a)$ for which the bit
error probability vanishes under BP decoding, and for which the gap
to capacity vanishes. \label{corollary: The fraction of degree-2
variable nodes as the gap to capacity vanishes}
\end{corollary}

In the continuation to this paper, sufficient conditions for the
tightness of \eqref{eq: upper bound on degree-2 variable nodes for
capacity-achieving LDPC ensembles} are considered (see
Lemma~\ref{lemma: L_2 for capacity approaching LDPC ensembles} on
page~\pageref{lemma: L_2 for capacity approaching LDPC
ensembles}).

\begin{remark} Note that for capacity-achieving sequences of
LDPC code ensembles whose transmission takes place over the BEC,
the bound in \eqref{eq: upper bound on degree-2 variable nodes for
capacity-achieving LDPC ensembles} is particularized to
$\frac{1}{2}$ {\em regardless of the erasure probability of this
channel}. This is indeed the case for some sequences of LDPC code
ensembles which achieve the capacity of the BEC under BP decoding
(see, e.g., \cite{LubyMSS_IT01, Oswald-it02,
Shokrollahi-IMA2000}). \label{remark: on the fraction of degree-2
variable nodes for c.a. sequences over the BEC}
\end{remark}

\vspace*{0.2cm}
\begin{corollary}{\bf{[A looser and simpler version of the upper
bound on $\lambda_2$]}} The bound \eqref{eq: bound on lambda_2 for
MBIOS channels} implies that
\begin{equation}
\lambda_2 < \frac{1}{\Bigl[c_1 + c_2 \ln
\left(\frac{1}{\varepsilon}\right)\Bigr]^{+}} \label{eq: loosened
version of the upper bound on lambda_2}
\end{equation}
for some constants $c_1$ and $c_2$ which only depend on the MBIOS
channel, and where $[x]^{+} \triangleq \max(x,1)$; the coefficient
$c_2$ of the logarithm in \eqref{eq: loosened version of the upper
bound on lambda_2} is given by
\begin{equation}
c_2 = \frac{\mathcal{B}(a)}{\ln \left(\frac{1}{g_1}\right)}
\label{eq: c_2}
\end{equation}
and it is strictly positive. \label{Corollary: The fraction of
edges connected to degree-2 variable nodes}
\end{corollary}

In the following proposition, it is shown that for the BEC, the
bounds in \eqref{eq: bound on lambda_2 for the BEC} and \eqref{eq:
loosened version of the upper bound on lambda_2} are tight under
BP decoding.

\begin{proposition}{\bf{[On the tightness of the upper bound on $\lambda_2$
for capacity-achieving sequences of LDPC code ensembles over the
BEC]}} The bounds in \eqref{eq: bound on lambda_2 for the BEC} and
\eqref{eq: loosened version of the upper bound on lambda_2} are
tight for the capacity-achieving sequence of right-regular LDPC
code ensembles over the BEC in \cite{Shokrollahi-IMA2000}. For
this sequence, $\lambda_2 \triangleq \lambda_2(\varepsilon)$
vanishes as $\varepsilon \rightarrow 0$ similarly to the upper
bound in \eqref{eq: loosened version of the upper bound on
lambda_2} with the same coefficient $c_2$ in \eqref{eq: c_2}.
\label{Proposition: the tightness of the upper bound on lambda_2}
\end{proposition}

\section{Proofs and Discussions}
\label{Section: Proofs of Main Results}

\subsection{Proof of Theorem~\ref{Theorem:
Lower bound on right degree}} \label{Subsection: Proof of
Theorem:lower bound on right degree}

Let $\vect{X}$ be a random codeword from the binary linear block
code $\mathcal{C}$, and let $\vect{Y}$ designate the output of the
communication channel when $\vect{X}$ is transmitted. Based on the
assumption that the code $\mathcal{C}$ is represented by a
full-rank parity-check matrix and $\mathcal{G}$ is the
corresponding bipartite graph which represents this code, then
inequality \eqref{eq: Lower bound on conditional entropy} holds.
Since $f(t) = x^t$ is convex for any $x \geq 0$ then Jensen's
inequality gives
\begin{equation*}
\Gamma(x) = \sum_i \Gamma_i x^i \geq x^{\sum_i i\, \Gamma_i} =
x^{a_{\text{R}}}\,,\quad x\geq 0\,.
\end{equation*}
Substituting the inequality above in \eqref{eq: Lower bound on
conditional entropy} implies that
\begin{equation}
\frac{H(\vect{X}|\vect{Y})}{n} \geq  R - C + \frac{1-R}{2\ln 2}\;
\sum_{k=1}^{\infty} \frac{g_k^{a_{\text{R}}}}{k(2k-1)}\,.
\label{eq: Lower bound on conditional entropy in terms of a_R}
\end{equation}
\begin{lemma}
\begin{equation}
g_k \geq \left(g_1\right)^k\, ,\quad \forall k\in\naturals.
\label{eq: g_k and g_1}
\end{equation}
\label{lemma: g_k and g_1}
\end{lemma}
\vspace*{-0.5cm}
\begin{proof}
For $k \geq 1$, Jensen's inequality and \eqref{eq: alternative
definition of g_k} give
\begin{eqnarray*}
& & g_k = \expectation\left[\tanh^{2k}\Bigl(\frac{L}{2}\Bigr)\right]\\
& & \hspace*{0.5cm} \geq \left(\expectation\left[\tanh^{2}
\Bigl(\frac{L}{2}\Bigr)\right]\right)^k\\
& & \hspace*{0.5cm} = \left(g_1\right)^k.
\end{eqnarray*}
\end{proof}
The substitution of \eqref{eq: g_k and g_1} in \eqref{eq: Lower
bound on conditional entropy in terms of a_R} gives
\begin{equation}
\frac{H(\vect{X}|\vect{Y})}{n} \geq  R - C + \frac{1-R}{2\ln 2}\;
\sum_{k=1}^{\infty}
\frac{\left(g_1^{a_{\text{R}}}\right)^k}{k(2k-1)}\,. \label{eq:
Lower bound on conditional entropy in terms of g_1 and a_R}
\end{equation}
The substitution $x = \frac{1-\sqrt{u}}{2}$ in \eqref{eq: power
series for h_2} gives
\begin{equation}
\frac{1}{2\ln2}\sum_{k=1}^{\infty} \frac{u^k}{k(2k-1)} =
1-h_2\left(\frac{1-\sqrt{u}}{2}\right), \; \forall \, u \in [0,1].
\label{eq: power series for 1-h_2}
\end{equation}
Since $0 \leq \tanh^2(x) < 1$ for all $x \in \reals$, we get from
\eqref{eq: alternative definition of g_k} that $0\leq g_1 \leq 1$
(this property holds for the entire sequence
$\{g_k\}_{k=1}^{\infty}$). Substituting \eqref{eq: power series for
1-h_2} into \eqref{eq: Lower bound on conditional entropy in terms
of g_1 and a_R} gives the following lower bound on the conditional
entropy:
\begin{equation}
\frac{H(\vect{X}|\vect{Y})}{n} \geq  1 - C - \left(1-R\right)
h_2\left(\frac{1-g_1^{a_{\text{R}}/2}}{2}\right). \label{eq: Lower
bound on conditional entropy in terms of h_2, g_1 and a_R}
\end{equation}
On the other hand, Fano's inequality provides the upper bound
\begin{equation}
\frac{H(\vect{X}|\vect{Y})}{n} \leq \left\{
\begin{array}{ll}
R \, P_{\text{B}} + \frac{h_2(P_{\text{B}})}{n}  \\
R \, h_2(P_{\text{b}})
\end{array}
\right. \label{eq: upper bound on conditional entropy}
\end{equation}
where, for the bound which is expressed in terms of the bit error
probability $P_{\text{b}}$, one can assume without any loss of
generality that the first $nR$ bits of the code are its
information bits, and their knowledge is sufficient for
determining the codeword.

In order to make the statement also valid for code ensembles (to be
clarified in Discussion~\ref{discussion: extension of Theorem 1 to
LDPC ensembles}), we rely on the inequality $R \leq 1$, and loosen
the bound in \eqref{eq: upper bound on conditional entropy} to get
\begin{equation}
\frac{H(\vect{X}|\vect{Y})}{n} \leq \delta \label{eq: loosened
upper bound on conditional entropy}
\end{equation}
where $\delta$ is introduced in \eqref{delta}. Combining
\eqref{eq: Lower bound on conditional entropy in terms of h_2, g_1
and a_R} and \eqref{eq: loosened upper bound on conditional
entropy} gives
\begin{equation}
\delta \geq  1 - C - (1-R)\,
h_2\left(\frac{1-g_1^{a_{\text{R}}/2}}{2}\right). \label{eq:
intermediate inequality for the generalization to LDPC ensembles}
\end{equation}
Since the RHS of \eqref{eq: intermediate inequality for the
generalization to LDPC ensembles} is monotonically increasing in
$R$, then following our assumption that $R \geq (1-\varepsilon)C$,
the bound is loosened by replacing $R$ with $(1-\varepsilon)C$.
This gives the inequality
\begin{equation*}
h_2\left(\frac{1-g_1^{a_{\text{R}}/2}}{2}\right) \geq
\frac{1-C-\delta}{1-(1-\varepsilon)C}\,.
\end{equation*}
Since the binary entropy function $h_2$ is monotonically
increasing on $[0, \frac{1}{2}]$ then
\begin{equation*}
g_1^{\frac{a_{\text{R}}}{2}} \leq
1-2h_2^{-1}\left(\frac{1-C-\delta}{1-(1-\varepsilon)C}\right)
\end{equation*}
which gives the lower bound on $a_{\text{R}}$ in \eqref{eq: lower
bound on a_R with finite P_b}.

Let us now consider the particular case where the transmission is
over the BEC. Note that for a BEC with erasure probability $p$,
$g_k = 1-p$ for all $k\in\naturals$ (in this case we have $L \in
\{0, +\infty\}$ with probabilities $p$ and $1-p$, respectively,
and the equality $\tanh(+\infty)=1$ is exploited in \eqref{eq:
alternative definition of g_k}). Therefore \eqref{eq: Lower bound
on conditional entropy in terms of a_R} is particularized to
\begin{equation*}
\frac{H(\vect{X}|\vect{Y})}{n} \geq  R - C +
\frac{(1-R)(1-p)^{a_{\text{R}}}}{2\ln 2}\;\sum_{k=1}^{\infty}
\frac{1}{k(2k-1)}\,.
\end{equation*}
Substituting $u = 1$ in \eqref{eq: power series for 1-h_2} gives
the equality
\begin{equation}
\frac{1}{2\ln 2}\sum_{k=1}^{\infty} \frac{1}{k(2k-1)} = 1
\label{eq: infinite series}
\end{equation}
and
\begin{equation}
\frac{H(\vect{X}|\vect{Y})}{n} \geq  R - C +
(1-R)(1-p)^{a_{\text{R}}}\,. \label{eq: almost final lower bound
on conditional entropy for the BEC}
\end{equation}
Note that the RHS of \eqref{eq: almost final lower bound on
conditional entropy for the BEC} is monotonic increasing as a
function of the rate $R$. Following the assumption that $R \geq
(1-\varepsilon)C$ where $C = 1-p$ is the capacity of the BEC, we
get
\begin{equation}
\frac{H(\vect{X}|\vect{Y})}{n} \geq  -\varepsilon (1-p) +
\bigl(1-(1-\varepsilon)(1-p)\bigr)(1-p)^{a_{\text{R}}}\,.
\label{eq: final lower bound on conditional entropy for the BEC}
\end{equation}
Similarly to \eqref{eq: upper bound on conditional entropy} and
\eqref{eq: loosened upper bound on conditional entropy}, we get
for the BEC
\begin{equation}
  \frac{H(\vect{X}|\vect{Y})}{n} \leq P_{\text{b}}
  \label{eq: upper bound on conditional entropy for the BEC}
\end{equation}
where the decoder finds $X_i$ with probability $1-P_{\text{b}}$;
otherwise, the bit $X_i$ is not determined by the decoder, and its
conditional entropy (given the sequence $\vect{Y}$) is upper
bounded by~1 bit. Combining \eqref{eq: final lower bound on
conditional entropy for the BEC} with \eqref{eq: upper bound on
conditional entropy for the BEC} gives
\begin{equation}
P_{\text{b}} \geq  -\varepsilon(1-p) +
\bigl(1-(1-\varepsilon)(1-p)\bigr)(1-p)^{a_{\text{R}}}\,.
\label{eq: expression with a_R and p}
\end{equation}
Finally, the lower bound on the average right degree in \eqref{eq:
lower bound on a_R for the BEC with finite P_b} follows from
\eqref{eq: expression with a_R and p} by simple algebra. Note that
in the case where $P_{\text{b}}=0$, the resulting lower bound
coincides with the result obtained in \cite[p.~1619]{Sason-it03}
(though it was derived there in a different way), and it gets the
form
\begin{equation}
a_{\text{R}} \geq \frac{\ln\left(1+\frac{p}{(1-p)\varepsilon}
\right)}{\ln \left(\frac{1}{1-p}\right)}. \label{eq: lower bound
on a_R for the BEC}
\end{equation}

We wish now to show that among all the MBIOS channels which
exhibit a given capacity $C$, the lower bound on the average
degree of the parity-check nodes as given in \eqref{eq: lower
bound on a_R with finite P_b} attains its maximal and minimal
values for a BSC and BEC, respectively.

\begin{lemma}{\bf{[Extreme values of $g_1$ among all MBIOS channels
with a given capacity]}} Among all the MBIOS channels with a given
capacity $C$, the value of $g_1$ satisfies
\begin{equation}
C \leq g_1 \leq \bigl(1-2h_2^{-1}(1-C)\bigr)^2 \label{eq: extreme
values of g_1}
\end{equation}
and these upper and lower bounds on $g_1$ are attained for a BSC
and BEC, respectively.  \label{lemma: Extreme values of g_1}
\end{lemma}

\begin{proof}
See Appendix~\ref{Appendix: Proof of the lemma on the extreme
values of g_1}.
\end{proof}

\begin{remark}
This lemma is in fact equivalent to the statement in
\cite[Theorem~1]{Jiang_IT08} with the extreme values derived in
its proof (note that \eqref{eq: alternative definition of g_k}
implies that the sequence $\{g_k\}$ is equal to the sequence
$\{m_{2k}\}$ in \cite{Jiang_IT08}, from which the equivalence
between Lemma~\ref{lemma: Extreme values of g_1} and
\cite[Theorem~1]{Jiang_IT08} follows directly). In
Appendix~\ref{Appendix: Proof of the lemma on the extreme values
of g_1}, we present an alternative proof which is more
elementary.\footnote{The author was un-aware of \cite{Jiang_IT08}
until its publication as a journal paper. The alternative proof on
Lemma 5 was found independently of this work.}
\end{remark}

\begin{remark}
The ratio between the upper and lower bounds on $g_1$ (see
Lemma~\ref{lemma: Extreme values of g_1}) is equal to $\eta(C) =
\frac{\bigl(1-2h_2^{-1}(1-C)\bigr)^2}{C}$. Based on \eqref{eq:
power series for 1-h_2}, one can verify that $\eta$ is a monotonic
decreasing function of the capacity where it tends to $2 \ln 2
\approx 1.386$ when $C \rightarrow 0$, and it is~1 (i.e., the
upper and lower bounds coincide) for $C=1$.
\end{remark}

\vspace*{0.2cm} Consider the set of all MBIOS channels with a
given capacity $C$ for which a target block error probability
$(P_{\text{B}}))$ or bit error probability $(P_{\text{b}})$ is
obtained under some decoding algorithm. To complete the proof of
the last statement in Theorem~\ref{Theorem: Lower bound on right
degree}, note that among this set of channels, the lower bound in
\eqref{eq: lower bound on a_R with finite P_b} is maximized or
minimized by maximizing or minimizing the value of $g_1$,
respectively. It therefore follows from Lemma~\ref{lemma: Extreme
values of g_1} that a universal bound on $a_{\text{R}}$ for the
above set of channels holds by replacing $g_1$ on the RHS of
\eqref{eq: lower bound on a_R with finite P_b} with $C$. The gives
the following universal lower bound:
\begin{equation}
a_{\text{R}} \geq \frac{2\ln\left(\frac{1}{1-2
h_2^{-1}\Bigl(\frac{1-C-\delta}
{1-(1-\varepsilon)C}\Bigr)}\right)}{\ln\Bigl(\frac{1}{C}\Bigr)}.
\label{eq: universal lower bound on a_R with finite P_b}
\end{equation}

\subsection*{Discussions on Theorem~\ref{Theorem: Lower bound on right degree} via its Proof}
In the following we discuss Theorem~\ref{Theorem: Lower bound on
right degree} via its proof, and consider some of the
generalizations of this theorem.

\begin{discussion}{\bf{[A discussion on the bounds in Theorem~\ref{Theorem: Lower
bound on right degree} and \cite[Eq.~(77)]{Wiechman_Sason}]}} If the
bit error probability vanishes, the lower bound in \eqref{eq: lower
bound on a_R with finite P_b} forms a tightened version of
\cite[Eq.~(77)]{Wiechman_Sason}. We note that both bounds are based
on \eqref{eq: Lower bound on conditional entropy} but the difference
in their derivation follows since \cite{Wiechman_Sason} relies on
the fact that the RHS of \eqref{eq: Lower bound on conditional
entropy} is an infinite sum of non-negative terms, and a simple
lower bound is obtained in \cite{Wiechman_Sason} by truncating this
sum after its first term. In the proof of Theorem~\ref{Theorem:
Lower bound on right degree}, on the other hand, a tightened lower
bound on the average right degree $(a_{\text{R}})$ is derived by
applying Jensen's inequality to the RHS of \eqref{eq: Lower bound on
conditional entropy} (see \eqref{eq: Lower bound on conditional
entropy in terms of g_1 and a_R}), and calculating exactly the
resulting bound via \eqref{eq: power series for 1-h_2}. In this
context, see Remark~\ref{Remark: connection with previous lower
bound on average right degree} on page~\pageref{Remark: connection
with previous lower bound on average right degree}. The additional
dependence of the bound in \eqref{eq: lower bound on a_R with finite
P_b} on $P_{\text{b}}$ makes Theorem~\ref{Theorem: Lower bound on
right degree} valid for codes of finite block length, whereas the
bound in \cite[Eq.~(77)]{Wiechman_Sason} can be only applied to the
asymptotic case of vanishing bit error (or erasure) probability by
letting the block length tend to infinity. \label{Discussion: on the
proof of lower bound on average right degree}
\end{discussion}

\begin{discussion}{\bf{[An adaptation of Theorem~\ref{Theorem: Lower
bound on right degree} for LDPC code ensembles]}} The statement in
Theorem~\ref{Theorem: Lower bound on right degree} can be adapted
for finite-length LDPC code ensembles whose transmission takes
place over an MBIOS channel. First, from
Section~\ref{subsubsection: An adaptation of the analysis to LDPC
codes}, the lower bound on the conditional entropy~\eqref{eq:
Lower bound on conditional entropy} holds for every code from this
ensemble if we relax the requirement of a full-rank parity-check
matrix, and instead replace the rate $R$ of the code by the design
rate $R_{\text{d}}$ of the ensemble. Similarly to the derivation
of \eqref{eq: Lower bound on conditional entropy in terms of h_2,
g_1 and a_R}, we get
\begin{equation*}
\frac{H(\vect{X}|\vect{Y})}{n} \geq 1-C - (1-R_{\text{d}}) \,
h_2\left(\frac{1-g_1^{a_{\text{R}}/2}}{2}\right).
\end{equation*}
Assume that $R_{\text{d}} \geq (1-\varepsilon)C$. Since the RHS of
the above inequality is monotonic increasing with $R_{\text{d}}$,
then for every code in this ensemble
\begin{equation}
\frac{H(\vect{X}|\vect{Y})}{n} \geq 1-C - \bigl(
1-(1-\varepsilon)C \bigr) \,
h_2\left(\frac{1-g_1^{a_{\text{R}}/2}}{2}\right). \label{eq: lower
bound on conditional entropy for a code whose parity-check matrix
is not necessarily full-rank}
\end{equation}
Note that this lower bound on the conditional entropy is global in
the sense that it does not depend on the code from the $(n,
\lambda, \rho)$ LDPC code ensemble; all these codes are
represented by bipartite graphs whose common value of
$a_{\text{R}}$ is equal to $\bigl(\int_0^1 \rho(x) \, dx
\bigr)^{-1}$. Note also that the parameter $g_1$ does not depend
on the code. Taking the expectation over the LDPC code ensemble
gives
\begin{equation}
\hspace*{-0.1cm} \expectation \left[\frac{H(\vect{X}|\vect{Y})}{n}
\right] \geq 1-C - \bigl( 1-(1-\varepsilon)C \bigr) \,
h_2\biggl(\frac{1-g_1^{a_{\text{R}}/2}}{2}\biggr). \label{eq:
inequality for the LDPC code ensemble}
\end{equation}
Note that $0 \leq g_1 < 1$ (unless $g_1=1$ when the capacity of
the binary-input channel is~1 bit per channel use which implies
that the channel is noiseless).

The loosening of the bound in the transition from \eqref{eq: upper
bound on conditional entropy} to \eqref{eq: loosened upper bound
on conditional entropy} is due to the fact that an upper bound on
the rate $R$ of a code from this ensemble is required; since
binary codes are considered, a trivial upper bound on the rate
is~1 bit per channel use (note that the rate of an arbitrarily
chosen code from this ensemble may exceed the channel capacity).
Due to the concavity of the binary entropy function, Jensen's
inequality gives
\begin{eqnarray}
\expectation \left[\frac{H(\vect{X}|\vect{Y})}{n}\right] \leq
\left\{\begin{array}{ll}
\overline{P_{\text{B}}} + \frac{h_2(\overline{P_{\text{B}}})}{n} \\
h_2(\overline{P_{\text{b}}})
\end{array}
\right.  \label{eq: first Jensen}
\end{eqnarray}
where $\overline{P_{\text{B}}} \triangleq
\expectation\bigl[P_{\text{B}}\bigr]$ and $\overline{P_{\text{b}}}
\triangleq \expectation\bigl[P_{\text{b}}\bigr]$ designate the
average block and bit error probabilities, respectively, of the
ensemble. Combining \eqref{eq: inequality for the LDPC code
ensemble} and \eqref{eq: first Jensen} leads to an adaptation of
Theorem~\ref{Theorem: Lower bound on right degree} for LDPC code
ensembles with the following modifications:
\begin{itemize}
\item The parity-check matrices of the codes are not required to be
full-rank (which otherwise would be problematic for LDPC code
ensembles).
\item The requirement on the rate a code is replaced by the same
requirement on the design rate of the LDPC code ensemble where we
refer to the average block and bit error probabilities of this
ensemble.
\end{itemize}
Note that the adaptation of the statement in Theorem~\ref{Theorem:
Lower bound on right degree} for LDPC code ensembles whose
transmission takes place over the BEC is more direct. For a BEC,
since $h_2(P_{\text{b}})$ on the LHS of \eqref{eq: intermediate
inequality for the generalization to LDPC ensembles} is replaced
by $P_{\text{b}}$ on the LHS of \eqref{eq: expression with a_R and
p}, then there is no need for Jensen's inequality as in \eqref{eq:
first Jensen}. \label{discussion: extension of Theorem 1 to LDPC
ensembles}
\end{discussion}

\begin{discussion}{\bf{[Adaptation of Theorem~\ref{Theorem:
Lower bound on right degree} for punctured LDPC code ensembles]}}
In the following, we consider an adaptation of
Theorem~\ref{Theorem: Lower bound on right degree} for LDPC code
ensembles with random or intentional puncturing where the
transmission takes place over an MBIOS channel. To this end, the
reader is referred to \cite[Section~V]{Sason-it07} where lower
bounds are derived on the average right degree and the graphical
complexity of such ensembles. The derivation of these bounds
relies on a lower bound \cite[Eqs.~(2) and (3)]{Sason-it07} which
generalizes \eqref{eq: Lower bound on conditional entropy} to the
case of statistically independent parallel MBIOS channels. This
lower bound was particularized in
\cite[Sections~II--IV]{Sason-it07} for the two settings of
randomly and intentionally punctured LDPC code ensembles which are
communicated over a single MBIOS channel. The concept of the proof
of Theorem~\ref{Theorem: Lower bound on right degree} enables to
tighten the lower bounds on the average right degree and the
graphical complexity, as presented in
\cite[Section~V]{Sason-it07}, for both randomly and intentionally
punctured LDPC code ensembles. More explicitly, by comparing the
proof of \eqref{eq: lower bound on a_R with finite P_b} with the
derivation of \cite[Eq.~(77)]{Wiechman_Sason} under the assumption
of vanishing bit error probability, one notices that the
tightening of the bound in the former case is enabled by combining
Lemma~\ref{lemma: g_k and g_1} with the equality in \eqref{eq:
power series for 1-h_2} (instead of the truncation of a
non-negative infinite series after its first term, as was done for
the derivation of the looser bound in \cite{Wiechman_Sason}). This
difference can be exploited exactly in the same way in connection
with the results from \cite[Section~V]{Sason-it07} for improving
the tightness of the lower bounds on the average right degree and
the graphical complexity for punctured LDPC code ensembles.
\label{discussion: extension of Theorem 1 to punctured LDPC
ensembles}
\end{discussion}

\subsection*{Proof of Corollary~\ref{corollary: lower bound on the
cardinality of the fundamental system of cycles of LDPC
ensembles}} The following lemma relies on the background material
in Section~\ref{subsection: elements from graph theory}, and it
serves for proving Corollary~\ref{corollary: lower bound on the
cardinality of the fundamental system of cycles of LDPC
ensembles}.
\begin{lemma}{\bf{[Cardinality of the fundamental system of
cycles]}} Under the assumptions of Theorem~\ref{Theorem: Lower
bound on right degree}, the cardinality of the fundamental system
of cycles of a bipartite graph $\mathcal{G}$, associated with a
full spanning forest of $\mathcal{G}$, is larger than
\begin{equation}
n \bigl[(1-R) (a_{\text{R}}-1) - 1 \bigr] \label{eq: fundamental
system of cycles of a Tanner graph}
\end{equation}
where $a_{\text{R}}$ can be replaced by the lower bounds in
\eqref{eq: lower bound on a_R with finite P_b} and \eqref{eq:
lower bound on a_R for the BEC with finite P_b} for a general
MBIOS channel and a BEC, respectively. From \eqref{eq: asymptotic
behavior of the average right degree}, the cardinality of the
fundamental system of cycles of the bipartite graph $\mathcal{G}$
which is associated with a full spanning forest of this graph is
$\Omega \left( \ln \frac{1}{\varepsilon} \right)$. \label{lemma:
lower bound on the cardinality of the fundamental system of cycles
of a binary linear code}
\end{lemma}
\begin{proof}
From Remark~\ref{remark: equality between the fundamental system
of cycles and cycle rank} (see Section~\ref{subsection: elements
from graph theory}), the cardinality of the fundamental system of
cycles of a bipartite graph $\mathcal{G}$, which is associated
with a full spanning forest of $\mathcal{G}$, is equal to the
cycle rank $\beta(\mathcal{G})$. From Eq.~\eqref{eq: cycle rank},
$\beta(\mathcal{G})
> |E_{\mathcal{G}}| - |V_{\mathcal{G}}|$ where $|E_{\mathcal{G}}|$
and $|V_{\mathcal{G}}|$ designate the number of edges and
vertices. Specializing this for a bipartite graph $\mathcal{G}$
which represents a full-rank parity-check matrix of a binary
linear block code, the number of vertices satisfies
$|V_{\mathcal{G}}|=n(2-R)$ (since there are $n$ variable nodes and
$n(1-R)$ parity-check nodes in the graph) and the number of edges
satisfies $|E_{\mathcal{G}}|=n(1-R) a_{\text{R}}$. Combining these
equalities gives the lower bound on the cardinality of the
fundamental system of cycles in \eqref{eq: fundamental system of
cycles of a Tanner graph}.
\end{proof}

The proof of \eqref{eq: lower bound on the cardinality of the
fundamental system of cycles of LDPC ensembles for general MBIOS
channels} and \eqref{eq: tightened lower bound on the cardinality
of the fundamental system of cycles of LDPC ensembles for the BEC}
is based on Remark~\ref{remark: extension of Theorem 1 to LDPC
ensembles} and Lemma~\ref{lemma: lower bound on the cardinality of
the fundamental system of cycles of a binary linear code}. By
substituting $P_{\text{b}}=0$ in \eqref{eq: lower bound on a_R
with finite P_b}, one obtains the following lower bound on the
average right degree as the average bit error probability of the
LDPC code ensemble vanishes:
\begin{equation}
a_{\text{R}} \geq \frac{2\ln\left(\frac{1}{1-2
h_2^{-1}\Bigl(\frac{1-C}
{1-(1-\varepsilon)C}\Bigr)}\right)}{\ln\Bigl(\frac{1}{g_1}\Bigr)}.
\label{eq: lower bound on a_R with vanishing P_b}
\end{equation}
Since the average bit error probability of the ensemble is assumed
to vanish as the block length tends to infinity, then
asymptotically with probability~1, the code rate of an arbitrary
code from the considered ensemble does not exceed the channel
capacity. By substituting the lower bound on $a_{\text{R}}$ from
\eqref{eq: lower bound on a_R with vanishing P_b} and an upper
bound on $R$ (i.e., $R \leq C$) into \eqref{eq: fundamental system
of cycles of a Tanner graph}, the asymptotic result in \eqref{eq:
lower bound on the cardinality of the fundamental system of cycles
of LDPC ensembles for general MBIOS channels} follows readily. A
similar proof of the tightened bound for the BEC in \eqref{eq:
tightened lower bound on the cardinality of the fundamental system
of cycles of LDPC ensembles for the BEC} follows by substituting
$P_{\text{b}}=0$ in \eqref{eq: lower bound on a_R for the BEC with
finite P_b}. This concludes the proof of Corollary~\ref{corollary:
lower bound on the cardinality of the fundamental system of cycles
of LDPC ensembles}.

\subsection{Proof of Theorem~\ref{Theorem: Behavior of left and right degrees}}
\label{Proof of Theorem: Behavior of left and right degrees}
Eq.~\eqref{eq: behavior of Lambda_i} is trivial (though it is
demonstrated in the continuation that, for degree-2 variable
nodes, this result is asymptotically tight as the gap to capacity
vanishes).

We turn now to consider the degrees of the parity-check nodes.
Similarly to Discussion~\ref{discussion: extension of Theorem 1 to
LDPC ensembles} (which succeeds the proof of Theorem~\ref{Theorem:
Lower bound on right degree}), we denote by $\vect{X}$ a random
codeword from the LDPC code ensemble $\bigl(n,\lambda,\rho\bigr)$
where the randomness is over the selected code from the ensemble
and the codeword which is selected from the code. Let $\vect{Y}$
designate the output of the communication channel when $\vect{X}$
is transmitted. From \eqref{eq: Lower bound on conditional
entropy} and its adaptation to LDPC code ensembles (see
Section~\ref{subsubsection: An adaptation of the analysis to LDPC
codes})
\begin{eqnarray}
&& \hspace*{-0.7cm} \frac{H(\vect{X}|\vect{Y})}{n} \nonumber \\
&& \hspace*{-0.7cm} \geq R_{\text{d}} - C +
\frac{1-R_{\text{d}}}{2\ln
2}\; \sum_{k=1}^{\infty} \frac{\Gamma(g_k)}{k(2k-1)}\nonumber\\
&& \hspace*{-0.7cm} = -\varepsilon\,C +
\frac{1-(1-\varepsilon)C}{2\ln 2}\;
\sum_{i=1}^{\infty}\left\{\Gamma_i \,
\sum_{k=1}^{\infty}\frac{g_k^i}{k(2k-1)}\right\} \label{eq: lower
bound on conditional entropy as function of Gamma and g_k}
\end{eqnarray}
where the last equality follows from the equality $\Gamma(x) =
\sum_i \Gamma_i x^i$ (see Section~\ref{LDPC}) and also since, by
assumption, the design rate of the LDPC code ensemble forms a
fraction $1-\varepsilon$ of the channel capacity. Applying
Lemma~\ref{lemma: g_k and g_1} to the RHS of \eqref{eq: lower
bound on conditional entropy as function of Gamma and g_k}, we get
\begin{eqnarray*}
&& \hspace*{-0.6cm} \frac{H(\vect{X}|\vect{Y})}{n} \nonumber \\ &&
\hspace*{-0.6cm} \geq -\varepsilon\,C +
\frac{1-(1-\varepsilon)C}{2\ln 2}\;
\sum_{i=1}^{\infty}\left\{\Gamma_i \, \sum_{k=1}^{\infty}\frac{\left(g_1^i\right)^k}{k(2k-1)}\right\}\nonumber\\
&& \hspace*{-0.6cm} = -\varepsilon\,C +
\big(1-(1-\varepsilon)C\big)
\sum_{i=1}^{\infty}\left\{\left[1-h_2\left(\frac{1-g_1^{i/2}}{2}\right)\right]
\, \Gamma_i \right\}
\end{eqnarray*}
where the last equality follows from \eqref{eq: power series for
1-h_2}. Combining \eqref{eq: loosened upper bound on conditional
entropy} with the last result gives
\begin{equation*}
h_2(P_{\text{b}}) \geq -\varepsilon\,C +
\big(1-(1-\varepsilon)C\big)
\sum_{i=1}^{\infty}\left[1-h_2\left(\frac{1-g_1^{i/2}}{2}\right)\right]
\, \Gamma_i
\end{equation*}
and therefore
\begin{equation}
\sum_{i=1}^{\infty}
\left\{\left[1-h_2\left(\frac{1-g_1^{i/2}}{2}\right)\right] \;
\Gamma_i \right\} \leq \frac{\varepsilon\,C +
h_2(P_{\text{b}})}{1-(1-\varepsilon)C} \label{eq: relationship
between Gamma epsilon and P_b}
\end{equation}
where $P_{\text{b}}$ designates the average bit error probability
of the ensemble under the considered decoding algorithm. Since all
the terms in the sum on the LHS of \eqref{eq: relationship between
Gamma epsilon and P_b} are non-negative, this sum is lower bounded
by its $i$-th term, for any degree $i$. This provides the
following upper bound on the fraction of parity-check nodes of any
finite degree $i$:
\begin{eqnarray}
\hspace*{-0.3cm} \Gamma_i &\leq& \frac{\varepsilon C +
h_2(P_{\text{b}})}{1-(1-\varepsilon)C}\;\frac{1}{1-h_2\left(\frac{1-g_1^{i/2}}{2}\right)}
\nonumber \\
\hspace*{-0.3cm} &\leq& \big(\varepsilon C+
h_2(P_{\text{b}})\big)\left[\frac{1}{1-C}\;
\frac{1}{1-h_2\left(\frac{1-g_1^{i/2}}{2}\right)}\right].
\label{eq: upper bound on Gamma_i}
\end{eqnarray}
This completes the proof of \eqref{eq: behavior of Gamma_i} for a
general MBIOS channel. Let us now consider the particular case where
the transmission is over a BEC with an erasure probability $p$. In
this case, $g_k=1-p$ for all $k\in\naturals$ (this equality follows
directly from \eqref{eq: definition of g_k}), and the channel
capacity is equal to $1-p$ bits per channel use. Therefore,
\eqref{eq: lower bound on conditional entropy as function of Gamma
and g_k} is particularized to
\begin{eqnarray}
&& \hspace*{-1cm} \frac{H(\vect{X}|\vect{Y})}{n} \nonumber \\ &&
\hspace*{-1cm} \geq -\varepsilon(1-p) +
\frac{1-(1-\varepsilon)(1-p)}{2\ln 2} \nonumber \\ &&
\hspace*{1.5cm} \cdot \sum_{i=1}^{\infty}
\left[\Gamma_i\,(1-p)^i \; \sum_{k=1}^{\infty}\frac{1}{k(2k-1)}\right]\nonumber\\
&& \hspace*{-1cm} = -\varepsilon(1-p) +
\big(1-(1-\varepsilon)(1-p)\big)\sum_{i=1}^{\infty}\Gamma_i\,(1-p)^i
\label{eq: lower bound on conditional entropy as function of Gamma
for the BEC}
\end{eqnarray}
where the above equality holds since
$\sum_{k=1}^{\infty}\frac{1}{k(2k-1)} = 2\ln 2$. Applying the
upper bound on the conditional entropy \eqref{eq: upper bound on
conditional entropy for the BEC} to the LHS of \eqref{eq: lower
bound on conditional entropy as function of Gamma for the BEC}, we
get
\begin{equation*}
P_{\text{b}} \geq -\varepsilon(1-p) + \bigl(p+\varepsilon \,
(1-p)\bigr) \, \sum_{i=1}^{\infty}\Gamma_i\,(1-p)^i
\end{equation*}
where $P_{\text{b}}$ denotes the average bit erasure probability
of the ensemble, and therefore
\begin{equation}
\sum_{i=1}^{\infty} \Bigl\{ \Gamma_i (1-p)^i \Bigr\} \leq
\frac{\varepsilon\,(1-p)+P_{\text{b}}}{p + \varepsilon \, (1-p)}.
\label{eq: relationship between Gamma epsilon and P_b for the BEC}
\end{equation}
Since the sum on the LHS of \eqref{eq: relationship between Gamma
epsilon and P_b for the BEC} is of non-negative terms, then we get
\begin{equation}
  \Gamma_i \leq
  \bigl(\varepsilon \, (1-p)
  +P_{\text{b}}\bigr)\;\left(\frac{1}{p\,(1-p)^i}\right)
  \label{eq: upper bound on Gamma_i for the BEC}
\end{equation}
so $h_2(P_{\text{b}})$ in \eqref{eq: behavior of Gamma_i} is
replaced for the BEC with $P_{\text{b}}$.

We turn now to consider the pair of degree distributions from the
edge perspective. The average left degree $(a_{\text{L}})$ of the
LDPC code ensemble satisfies
\begin{equation}
\frac{1}{a_{\text{L}}} = \sum_{i=2}^{\infty} \frac{\lambda_i}{i}
\label{eq: simple relationship between a_L and lambda}
\end{equation}
which implies that for any degree~$i$ of the variable nodes
\begin{equation}
\lambda_i \leq \frac{i}{a_{\text{L}}}\,. \label{eq: first bound on
lambda_i}
\end{equation}
Since the design rate of the LDPC code ensemble is assumed to be a
fraction $1-\varepsilon$ of the channel capacity, then the average
right and left degrees satisfy
\begin{eqnarray}
&& a_{\text{L}} = \bigl(1-(1-\varepsilon)C\bigr) a_{\text{R}}
\nonumber\\
&& \hspace*{0.5cm} \geq (1-C) a_{\text{R}}. \label{eq: relation
between average right and left degrees}
\end{eqnarray}
Substituting \eqref{eq: relation between average right and left
degrees} on the RHS of \eqref{eq: first bound on lambda_i} and
applying the lower bound on $a_{\text{R}}$ in \eqref{eq: lower
bound on a_R with finite P_b} gives
\begin{equation}
\lambda_i \leq \frac{i \ln\bigl(\frac{1}{g_1}\bigr)}{2 (1-C) \;
\ln\left(\frac{1}{1-2 h_2^{-1}\Bigl(\frac{1-C-
h_2(P_{\text{b}})}{1-(1-\varepsilon)C}\Bigr)} \right)}\,.
\label{eq: upper bound on lambda_i}
\end{equation}
Using the power series for the binary entropy function in
\eqref{eq: power series for h_2} and truncating the sum on the RHS
after the first term gives
\begin{equation*}
1-h_2(x) \geq \frac{(1-2x)^2}{2\,\ln 2}
\end{equation*}
and substituting $u=h_2(x)$ yields
\begin{equation}
\big(1-2h_2^{-1}(u)\big)^2 \leq 2\,\ln 2\cdot(1-u), \quad \forall
\; 0 \leq u \leq 1. \label{eq: upper bound on (1-2h_2^-1)^2}
\end{equation}
Combining \eqref{eq: upper bound on lambda_i} and \eqref{eq: upper
bound on (1-2h_2^-1)^2} gives
\begin{eqnarray*}
\lambda_i &\leq& \frac{i \ln\bigl(\frac{1}{g_1}\bigr)}{(1-C) \;
\ln\left(\frac{1}{2\ln 2} \;
\frac{1}{1-\frac{1-C-h_2(P_{\text{b}})}{1-(1-\varepsilon)C}}\right)}\\
& = & \frac{i \ln\bigl(\frac{1}{g_1}\bigr)}{(1-C) \;
\ln\left(\frac{1}{2\ln 2} \; \frac{1-(1-\varepsilon)C}{\varepsilon
C +h_2(P_{\text{b}})}\right)}\\
& \leq & \frac{i \ln\bigl(\frac{1}{g_1}\bigr)}{(1-C) \;
\left[\ln\left(\frac{1}{\varepsilon C+ h_2(P_{\text{b}})}\right) +
\ln\left(\frac{{1-C}}{2\ln 2}\right)\right]}
\end{eqnarray*}
which completes the proof of \eqref{eq: behavior of lambda_i} for
general MBIOS channels. For the BEC, we substitute \eqref{eq:
relation between average right and left degrees} and the lower
bound on the average right degree in \eqref{eq: lower bound on a_R
for the BEC with finite P_b} into the RHS of \eqref{eq: first
bound on lambda_i} to get
\begin{eqnarray}
\lambda_i &\leq& \frac{i \ln\bigl(\frac{1}{1-p}\bigr)}{p
\;\ln\left(1+\frac{p- P_{\text{b}}}{\varepsilon \,
(1-p)+P_{\text{b}}}
\right)}\nonumber\\
&=& \frac{i \ln\bigl(\frac{1}{1-p}\bigr)}{p
\;\ln\left(\frac{\varepsilon \, (1-p) + p}{\varepsilon \,
(1-p)+P_{\text{b}}}
\right)}\nonumber\\
&\leq&\frac{i \ln\bigl(\frac{1}{1-p}\bigr)}{p
\;\left[\ln\left(\frac{1}{\varepsilon \, (1-p)+ P_{\text{b}}}
\right) + \ln(p)\right]}\,. \label{eq: upper bound on lambda_i for
the BEC}
\end{eqnarray}
Hence, $h_2(P_{\text{b}})$ in \eqref{eq: behavior of lambda_i} is
replaced by $P_{\text{b}}$ when the communication channel is a
BEC. Considering the right degree distribution of the ensemble, we
have
\begin{equation*}
\frac{1}{a_{\text{R}}} = \sum_{i=1}^{\infty} \frac{\rho_i}{i}\,.
\end{equation*}
By following the same steps as in \eqref{eq: simple relationship
between a_L and lambda}--\eqref{eq: upper bound on lambda_i for
the BEC}, one obtains an upper bound on $\rho_i$ for any
degree~$i$ of the parity-check nodes. The asymptotic behavior of
the resulting upper bound on $\rho_i$ is similar to the upper
bound on $\lambda_i$ as given in \eqref{eq: upper bound on
lambda_i for the BEC}. However, as we show in the following, a
tighter upper bound on the fraction of edges connected to
parity-check nodes of degree $i$ is derived from the equality
\begin{equation}
  \rho_i = \frac{i\,\Gamma_i}{a_{\text{R}}}\,.
  \label{eq: basic relationship between rho_i and Gamma_i}
\end{equation}
Substituting \eqref{eq: lower bound on a_R with finite P_b} and
\eqref{eq: upper bound on Gamma_i} in the above equality, we get
\begin{eqnarray}
&& \hspace*{-1cm} \rho_i \leq \frac{\varepsilon\,C +
    h_2(P_{\text{b}})}{1-C}\;
    \frac{\ln\Bigl(\frac{1}{g_1}\Bigr)}{2\ln\left(\frac{1}{1-2h_2^{-1}
    \Bigl(\frac{1-C-h_2(P_{\text{b}})}{1-(1-\varepsilon)C}\Bigr)}\right)}
    \cdot \frac{i}{1-h_2\left(\frac{1-g_1^{i/2}}{2}\right)}\,.
    \label{eq: upper bound on rho_i}
\end{eqnarray}
Applying \eqref{eq: upper bound on (1-2h_2^-1)^2} to the
denominator of the second term on the RHS of \eqref{eq: upper
bound on rho_i} gives
\begin{eqnarray*}
&& \hspace*{-0.7cm} \rho_i \leq \frac{\varepsilon\,C +
    h_2(P_{\text{b}})}{1-C} \;
    \frac{\ln\Bigl(\frac{1}{g_1}\Bigr)}{\ln\left(\frac{1}{2\ln 2} \; \frac{1}
    {1-\frac{1-C-h_2(P_{\text{b}})}{1-(1-\varepsilon)C}}\right)}
    \; \frac{i}{1-h_2\left(\frac{1-g_1^{i/2}}{2}\right)}\\[.1cm]
&& \hspace*{-0.3cm} = \frac{\varepsilon\,C+h_2(P_{\text{b}})}{1-C}
    \; \frac{\ln\Bigl(\frac{1}{g_1}\Bigr)}{\ln\left(\frac{1}{2\ln 2} \;
    \frac{1-(1-\varepsilon)C}{\varepsilon C +
    h_2(P_{\text{b}})}\right)} \;
    \frac{i}{1-h_2\left(\frac{1-g_1^{i/2}}{2}\right)}\\[.1cm]
&& \hspace*{-0.3cm} \leq
    \frac{\ln\Bigl(\frac{1}{g_1}\Bigr)}{1-C}\; \frac{\varepsilon \,
    C+ h_2(P_{\text{b}})}{\ln\left(\frac{1}{\varepsilon \, C+
    h_2(P_{\text{b}})}\right) + \ln\left(\frac{1-C}{2\ln 2}\right)} \;
    \frac{i}{1-h_2\left(\frac{1-g_1^{i/2}}{2}\right)}
    \,.
\end{eqnarray*}
This proves \eqref{eq: behavior of rho_i} regarding the fraction
of edges connected to parity-check nodes of an arbitrary finite
degree $i$. For a BEC, a substitution of \eqref{eq: lower bound on
a_R for the BEC with finite P_b} and \eqref{eq: upper bound on
Gamma_i for the BEC} in \eqref{eq: basic relationship between
rho_i and Gamma_i} gives
\begin{equation*}
\rho_i \leq \frac{i \bigl[\varepsilon\,(1-p) +
    P_{\text{b}}\bigr]}{p(1-p)^i}\;\;
    \frac{\ln\left(\frac{1}{1-p}\right)}{\ln\left(1+\frac{p-
P_{\text{b}}}{\varepsilon \, (1-p)+P_{\text{b}}} \right)}.
\end{equation*}
Followed by some straightforward algebra, this proves \eqref{eq:
behavior of rho_i} for the BEC when $h_2(P_{\text{b}})$ is
replaced with $P_{\text{b}}$.

\begin{remark}{\bf{[\bf{Note on Theorem~\ref{Theorem: Behavior of left and right degrees}
and Corollary~\ref{Corollary: Behavior of the left and right degrees
when P_b=0}]}}} Consider the capacity-achieving sequence of
right-regular LDPC code ensemble as introduced in
\cite{Shokrollahi-IMA2000}. The gap to capacity $(\varepsilon)$ can
be made arbitrarily small for this sequence (even under BP
decoding), although $\rho_i=1$ for some integer $i$. At first
glance, it looks contradictory to Corollary~\ref{Corollary: Behavior
of the left and right degrees when P_b=0} (see
p.~\pageref{Corollary: Behavior of the left and right degrees when
P_b=0}) which states that $\rho_i$ is upper bounded by an expression
which scales like $\frac{\varepsilon}{\ln \frac{1}{\varepsilon}}$
for any finite degree $i$, and it therefore should tend to zero as
the gap to capacity vanishes. However, the right degree of this
sequence scales like $\ln \frac{1}{\varepsilon}$ (see
\cite{Shokrollahi-IMA2000} and \cite[Theorem~2.3]{Sason-it03}),
hence the index~$i$ for which $\rho_i=1$  becomes unbounded as
$\varepsilon \rightarrow 0$. Note that Corollary~\ref{Corollary:
Behavior of the left and right degrees when P_b=0} applies on the
other hand to finite and bounded degrees~$i$ in the limit where the
gap to capacity vanishes. Moreover, as we let $\varepsilon
\rightarrow 0$ for this capacity-achieving and right-regular
sequence, then $\rho_i$ is identically zero for all finite and
bounded degrees~$i$. \label{remark: comment on the asymptotic
behavior of the degree distributions}
\end{remark}

\begin{remark}{\bf{[On the degree
distribution of the parity-check nodes for the set of MBIOS
channels with a given capacity]}} Consider the set of all MBIOS
channels of a given capacity $C$, and consider a required bit
error probability $p_{\text{b}}$. By combining the inequality
constraint \eqref{eq: relationship between Gamma epsilon and P_b}
with the extreme values of $g_1$ in Lemma~\ref{lemma: Extreme
values of g_1} (see \eqref{eq: extreme values of g_1}), we obtain
the following universal inequality constraint which should hold
for this set of channels:
\begin{equation}
\sum_{i=1}^{\infty}
\left\{\left[1-h_2\left(\frac{1-C^{\frac{i}{2}}}{2}\right)\right]
\; \Gamma_i \right\} \leq \frac{\varepsilon\,C +
h_2(P_{\text{b}})}{1-(1-\varepsilon)C} \, . \label{eq: universal
relationship between Gamma epsilon and P_b}
\end{equation}
We refer later to this inequality when we consider linear
programming bounds for the degree distributions of
capacity-approaching LDPC code ensembles (see
Section~\ref{Section: Numerical Results}).
\end{remark}

\subsection{Proof of Theorem~\ref{Theorem: degree-2 variable
nodes}} \label{Proof of Theorem: degree-2 variable nodes} Consider
bipartite graphs which correspond to an LDPC code ensemble with
pair of degree distributions $(\lambda, \rho)$. The average
degrees of the variable nodes and the parity-check nodes of these
graphs are given in \eqref{eq: average left degree} and \eqref{eq:
average right degree}, respectively. Hence, the fraction of
degree-2 variable nodes is given by
\begin{equation}
\Lambda_2 = \frac{\lambda_2 \; a_{\mathrm{L}}}{2} =
\frac{\lambda_2}{2 \int_0^1 \lambda(x) \mathrm{d}x} \label{eq:
simple expression for Lambda_2}
\end{equation}
and the design rate of this ensemble is given by \eqref{design
rate of LDPC ensemble}. Using \eqref{design rate of LDPC
ensemble}, we rewrite $\int_0^1 \lambda(x) \mathrm{d}x$ at the
denominator of \eqref{eq: simple expression for Lambda_2} as
\begin{equation}
\int_0^1 \lambda(x) \mathrm{d}x =
\frac{1}{1-R_{\text{d}}}\,\int_0^1 \rho(x) \mathrm{d}x\,.
\label{eq: expression for integral of lambda}
\end{equation}
By assumption, the considered sequence of ensembles achieves
vanishing bit error probability under BP decoding, and hence the
stability condition in \eqref{eq: stability condition} is
satisfied. Combining \eqref{eq: stability condition}, \eqref{eq:
simple expression for Lambda_2} and \eqref{eq: expression for
integral of lambda} leads to the following upper bound on
$\Lambda_2$:
\begin{equation}
\Lambda_2 < \frac{1-R_{\text{d}}}{2 \, \mathcal{B}(a) \; \rho'(1)
\int_0^1 \rho(x) \mathrm{d}x} \,. \label{eq: upper bound on
Lambda_2 in terms of rho}
\end{equation}
From the convexity of $f(t) = x^t$ for $x>0$, Jensen's inequality
gives
\begin{eqnarray*}
&& \int_0^1 \rho(x) \mathrm{d}x \nonumber \\ && = \int_0^1
\sum_{i}\rho_i x^{i-1} \mathrm{d}x \nonumber \\ && \geq \int_0^1
x^{\sum_{i}\rho_i\,(i-1)} \mathrm{d}x \nonumber \\ && = \int_0^1
x^{\,\rho'(1)} \mathrm{d}x \nonumber \\ && = \frac{1}{\rho'(1)+1}
\end{eqnarray*}
which implies that
\begin{equation}
\rho'(1) \geq \frac{1}{\int_0^1 \rho(x) \mathrm{d}x} - 1 =
a_{\text{R}} - 1\,. \label{eq: rho'(1) > a_R - 1}
\end{equation}
Substituting \eqref{eq: rho'(1) > a_R - 1} in \eqref{eq: upper
bound on Lambda_2 in terms of rho} and since $R_{\text{d}} =
(1-\varepsilon)C$ then
\begin{eqnarray}
\Lambda_2 &<& \frac{1-R_{\text{d}}}{2 \, \mathcal{B}(a)} \left(1+ \frac{1}{\rho'(1)}\right)\nonumber\\
&\leq& \frac{1-R_{\text{d}}}{2 \, \mathcal{B}(a)} \left(1+ \frac{1}{a_{\text{R}} - 1}\right)\nonumber\\
&=&\frac{1-C}{2 \, \mathcal{B}(a)}\left(1+\frac{\varepsilon
C}{1-C}\right) \left(1+ \frac{1}{a_{\text{R}} -1}\right).
\label{eq: upper bound on Lambda_2 in terms of a_R}
\end{eqnarray}
Since the RHS of \eqref{eq: upper bound on Lambda_2 in terms of
a_R} is monotonically decreasing with the average right degree
($a_{\text{R}}$), this bound still holds when $a_{\text{R}}$ is
replaced by a lower bound. For all $m\in\naturals$, let
$P_{\text{b},m}$ designate the average bit error probability of
the LDPC code ensemble $\bigl(n_m, \lambda(x), \rho(x)\bigr)$
under BP decoding. Applying Theorem~\ref{Theorem: Lower bound on
right degree} where $P_{\text{b},m}$ vanishes as $m \rightarrow
\infty$ gives
\begin{equation}
a_{\text{R}} \geq \frac{2\ln\left(\frac{1}{1-2
h_2^{-1}\bigl(\frac{1-C}{1-(1-\varepsilon)C}\bigr)}\right)}{\ln\Bigl(\frac{1}{g_1}\Bigr)}\,.
\label{eq: lower bound on a_R}
\end{equation}
The upper bound in \eqref{eq: upper bound on fraction of degree-2
variable nodes} follows by substituting \eqref{eq: lower bound on
a_R} in \eqref{eq: upper bound on Lambda_2 in terms of a_R}.

We now turn to derive the upper bound on the fraction of edges
which are connected to degree-2 variable nodes. Since the
considered sequence of LDPC code ensembles achieves vanishing bit
error probability under BP decoding, then the stability condition
\eqref{eq: stability condition} implies that
\begin{equation*}
\lambda_2 = \lambda'(0) < \frac{1}{\rho'(1) \, \mathcal{B}(a)}
\end{equation*}
where $\mathcal{B}(a)$ is given in \eqref{eq: definition of
Bhattacharyya constant}. Combining  this with \eqref{eq: rho'(1) >
a_R - 1} gives
\begin{equation}
\lambda_2 <\frac{1}{(a_{\text{R}}-1) \, \mathcal{B}(a)} \label{eq:
upper bound on lambda_2 in terms of a_R}
\end{equation}
where $a_{\text{R}}$ designates the common average right degree of
the sequence of ensembles. The upper bounds on $\lambda_2$ in
\eqref{eq: bound on lambda_2 for MBIOS channels} and \eqref{eq:
bound on lambda_2 for the BEC} are obtained by substituting
\eqref{eq: lower bound on a_R} and \eqref{eq: lower bound on a_R
for the BEC} (these are the lower bounds on $a_{\text{R}}$ derived
in Theorem~\ref{Theorem: Lower bound on right degree} for
vanishing bit error/ erasure probability), respectively, in
\eqref{eq: upper bound on lambda_2 in terms of a_R}.

Consider the set of all MBIOS channels with a given capacity $C$
and a Bhattacharyya constant $\mathcal{B}(a)$, for which the bit
error probability of the BP decoder vanishes for the considered
sequence of LDPC code ensembles. Universal upper bound on
$\Lambda_2$ and $\lambda_2$ follow directly by combining the
bounds in \eqref{eq: upper bound on fraction of degree-2 variable
nodes} and \eqref{eq: bound on lambda_2 for MBIOS channels},
respectively, with Lemma~\ref{lemma: Extreme values of g_1} (note
that the upper bound on the RHS of \eqref{eq: upper bound on
fraction of degree-2 variable nodes} is a monotonic decreasing
function of $g_1$; this bound therefore attains its maximal value
at the minimal value of $g_1$, i.e., when $g_1 = C$). Therefore,
the universal upper bounds on $\Lambda_2$ and $\lambda_2$ hold for
all the channels from the above set by substituting $g_1=C$ on the
RHS of \eqref{eq: upper bound on fraction of degree-2 variable
nodes} and \eqref{eq: bound on lambda_2 for MBIOS channels},
respectively.

For a transmission over the BEC, the improved upper bound on the
degree-2 variable nodes follows by substituting the lower bound in
\eqref{eq: lower bound on a_R for the BEC with finite P_b} (where
the bit erasure probability $P_{\text{b}}$ vanishes) into
\eqref{eq: upper bound on Lambda_2 in terms of a_R}. Note that for
a BEC with erasure probability $p$, $1-C=\mathcal{B}(a)=p$ and
$\frac{1-C}{2 \mathcal{B}(a)} = \frac{1}{2}$. Similarly, the upper
bound on the fraction of edges which are attached to degree-2
variable nodes follows by substituting \eqref{eq: lower bound on
a_R for the BEC with finite P_b} and $\mathcal{B}(a)=p$ into
\eqref{eq: upper bound on lambda_2 in terms of a_R}.

\vspace*{0.2cm}
\begin{discussion}{\bf{[On the tightness of the upper bound \eqref{eq: upper bound
on degree-2 variable nodes for capacity-achieving LDPC ensembles}
on the fraction of degree-2 variable nodes for capacity-achieving
LDPC code ensembles over MBIOS channels]}} In the following, the
tightness of the bound in \eqref{eq: upper bound on degree-2
variable nodes for capacity-achieving LDPC ensembles} is
considered:
\begin{lemma}{\bf{[On the asymptotic fraction of degree~2 variable
nodes for capacity-achieving sequences of LDPC code ensembles]}}
Let $(n_m, \lambda_m, \rho_m)$ be a sequence of LDPC code
ensembles whose transmission takes place over an MBIOS channel of
capacity $C$ (in bits per channel use). Assume that this sequence
is capacity-achieving under BP decoding, and also that the
flatness condition is asymptotically satisfied for this sequence
(i.e., the stability condition in \eqref{eq: stability condition}
is satisfied asymptotically with equality). Let us also assume
that the limit of the ratio between the standard deviation and the
expectation of the right degree distribution in the LDPC code
ensemble $(n_m, \lambda_m, \rho_m)$ is finite as $m \rightarrow
\infty$, and denote this limit by $K$. Then, the asymptotic
fraction of degree-2 variable nodes in this sequence is equal to
\begin{equation}
\lim_{m\rightarrow\infty} \Lambda^{(m)}_2 = \frac{1-C}{2(1+K^2) \,
\mathcal{B}(a)} \label{eq: L_2 for capacity approaching LDPC
ensembles}
\end{equation}
where $\mathcal{B}(a)$ is introduced in \eqref{eq: definition of
Bhattacharyya constant}.  \label{lemma: L_2 for capacity
approaching LDPC ensembles}
\end{lemma}
\begin{proof}
See Appendix~\ref{Appendix: Proof of lemma on L_2 for c.a. LDPC
ensembles}.
\end{proof}

As a particular case of Lemma~\ref{lemma: L_2 for capacity
approaching LDPC ensembles}, if $K=0$ (this happens, e.g., when the
right degree is fixed), then the asymptotic fraction of degree-2
variable nodes in \eqref{eq: L_2 for capacity approaching LDPC
ensembles} coincides with the upper bound in \eqref{eq: upper bound
on degree-2 variable nodes for capacity-achieving LDPC ensembles}.

\begin{remark}
We note that the property proved in Lemma~\ref{lemma: L_2 for
capacity approaching LDPC ensembles} for the non-vanishing
asymptotic fraction of degree-2 {\em variable nodes} of
capacity-achieving sequences of LDPC code ensembles is reminiscent
of another information-theoretic property which was proved by
Shokrollahi with respect to the non-vanishing fraction of degree-2
{\em output nodes} for capacity-achieving sequences of Raptor
codes whose transmission takes place over an MBIOS channel (see
\cite[Theorem~11 and Proposition~12]{Raptor_IT06}). \label{remark:
raptor codes}
\end{remark}
\end{discussion}

\vspace*{0.2cm} {\em Proof of Corollary~\ref{corollary: The
fraction of degree-2 variable nodes as the gap to capacity
vanishes}}: The upper bound \eqref{eq: upper bound on degree-2
variable nodes for capacity-achieving LDPC ensembles} on the
fraction of degree-2 variable nodes for capacity-achieving LDPC
code ensembles follows directly by letting the gap to capacity
$\varepsilon$ tend to zero in \eqref{eq: upper bound on fraction
of degree-2 variable nodes}. We wish to show that the upper bound
in \eqref{eq: upper bound on degree-2 variable nodes for
capacity-achieving LDPC ensembles} is necessarily not larger
than~$\frac{1}{2}$ for all MBIOS channels, and it is equal to
$\frac{1}{2}$ for a BEC regardless of the erasure probability of
this channel. To this end, we prove the following lemma:
\begin{lemma}
For every MBIOS channel, the sum of its capacity and its
Bhattacharyya constant is at least~1. The minimal value of this sum
is attained for a BEC, irrespectively of the erasure probability of
this channel, and is equal to~1. \label{lemma: relation between the
capacity and Bhattacharyya constant of an MBIOS channel}
\end{lemma}
\begin{proof}
See Appendix~\ref{Appendix: inequality related to the Bhattacharyya
constant and channel capacity}.
\end{proof}

Combining Lemma~\ref{lemma: relation between the capacity and
Bhattacharyya constant of an MBIOS channel} and the RHS of
\eqref{eq: upper bound on degree-2 variable nodes for
capacity-achieving LDPC ensembles} implies that the fraction of
degree-2 variable nodes for an arbitrary capacity-achieving
sequence of LDPC code ensembles under BP decoding is upper bounded
by $\frac{1}{2}$. Note that this maximal value is attained for a
BEC (see also Remark~\ref{remark: on the fraction of degree-2
variable nodes for c.a. sequences over the BEC} on
page~\pageref{remark: on the fraction of degree-2 variable nodes
for c.a. sequences over the BEC}). This completes the proof of
Corollary~\ref{corollary: The fraction of degree-2 variable nodes
as the gap to capacity vanishes}.

In the following, we compare two upper bounds on the fraction of
edges connected to degree-2 variable nodes. One of these bounds is
given in Theorem~\ref{Theorem: degree-2 variable nodes}, and the
other bound follows along the lines of the proof of
Theorem~\ref{Theorem: Behavior of left and right degrees}.
\begin{discussion}{\bf{[Comparison between two upper bounds on $\lambda_2$: ML versus
iterative decoding]}} In the proof of Theorem~\ref{Theorem:
Behavior of left and right degrees}, we derive an upper bound on
the fraction of edges connected to variable nodes of degree $i$
for ensembles of LDPC codes which achieve a bit error (or erasure)
probability $P_{\text{b}}$ under an arbitrary decoding algorithm
(see \eqref{eq: upper bound on lambda_i} and the tightened version
\eqref{eq: upper bound on lambda_i for the BEC} of this bound for
the BEC). Referring to degree-2 variable nodes and letting
$P_{\text{b}}$ vanish, \eqref{eq: upper bound on lambda_i} gives
\begin{equation}
\lambda_2 \leq \frac{\ln\bigl(\frac{1}{g_1}\bigr)}{(1-C) \;
\ln\biggl(\frac{1}{1-2
h_2^{-1}\bigl(\frac{1-C}{1-R}\bigr)}\biggr)} \label{eq: first
bound on lambda_2 for MBIOS channels}
\end{equation}
where $R =(1-\varepsilon)C$. It is interesting to see that there
is some similarity between the two upper bounds on $\lambda_2$ as
given in \eqref{eq: bound on lambda_2 for MBIOS channels} and
\eqref{eq: first bound on lambda_2 for MBIOS channels}. In the
following, we compare between the two bounds on $\lambda_2$ by
calculating the ratio between the bound in \eqref{eq: bound on
lambda_2 for MBIOS channels} which relies on the stability
condition, and the bound in \eqref{eq: first bound on lambda_2 for
MBIOS channels} which follows along the lines of the proof of
Theorem~\ref{Theorem: Behavior of left and right degrees}. This
gives
\begin{eqnarray}
&& \frac{\ln\left(\frac{1}{g_1}\right)}{\mathcal{B}(a) \,
\ln\biggl(\frac{g_1}{\bigl[1-2h_2^{-1}
\bigl(\frac{1-C}{1-R}\bigr)\bigr]^2}\biggr)}
\nonumber\\[0.1cm]
&& \cdot \frac{\bigl(1-C\bigr) \; \ln\biggl(\frac{1}{1-2
h_2^{-1}\bigl(\frac{1-C}{1-R}\bigr)}\biggr)}{\ln\bigl(\frac{1}{g_1}\bigr)}
\nonumber \\[0.1cm]
&& = \frac{1-C}{\mathcal{B}(a)} \; \frac{\ln\biggl(\frac{1}{1-2
h_2^{-1}\bigl(\frac{1-C}{1-R}\bigr)}\biggr)}{
\ln\biggl(\frac{g_1}{\bigl[1-2h_2^{-1}
\bigl(\frac{1-C}{1-R}\bigr)\bigr]^2}\biggr)}
\nonumber\\[0.1cm]
&& = \frac{1-C}{\mathcal{B}(a)} \; \frac{\ln\biggl(\frac{1}{1-2
h_2^{-1}\bigl(\frac{1-C}{1-R}\bigr)}\biggr)}{ \ln(g_1) + 2
\ln\biggl(\frac{1}{1-2h_2^{-1}
\bigl(\frac{1-C}{1-R}\bigr)}\biggr)} \, . \label{eq: ratio between
two bounds on lambda_2}
\end{eqnarray}
Hence, as the gap to capacity vanishes (i.e., $\varepsilon
\rightarrow 0$), the expression in \eqref{eq: ratio between two
bounds on lambda_2} for the ratio between the two bounds on
$\lambda_2$ tends to $\frac{1-C}{2 \mathcal{B}(a)}$. By
Lemma~\ref{lemma: relation between the capacity and Bhattacharyya
constant of an MBIOS channel}, $\mathcal{B}(a) + C - 1 \geq 0$,
which implies that $\frac{1-C}{2 \mathcal{B}(a)} \leq
\frac{1}{2}$. Hence, the upper bound on $\lambda_2$ in \eqref{eq:
bound on lambda_2 for MBIOS channels} improves the bound in
\eqref{eq: first bound on lambda_2 for MBIOS channels} by at least
a factor of~2 (where the former bound is given in
Theorem~\ref{Theorem: degree-2 variable nodes}, and the latter
bound follows along the lines of the proof of
Theorem~\ref{Theorem: Behavior of left and right degrees}). We
note that the basis of the comparison between these two upper
bounds on $\lambda_2$ is the assumption of vanishing bit error
probability under BP decoding, though the bound in \eqref{eq:
first bound on lambda_2 for MBIOS channels} also holds with the
weaker requirement of vanishing bit error probability under ML
decoding.
\end{discussion}
\vspace*{0.2cm}
\subsubsection*{Proof of Corollary~\ref{Corollary: The fraction of edges connected to
degree-2 variable nodes}} See Appendix~\ref{Appendix: Proof of the
corollary on lambda_2}. \vspace*{0.2cm}
\subsubsection*{Proof of Proposition~\ref{Proposition: the
tightness of the upper bound on lambda_2}} See
Appendix~\ref{Appendix: Proof of the proposition on the tightness
of the upper bound on lambda_2}.

\section{Implications of the Information-Theoretic Bounds and Numerical Results}
\label{Section: Numerical Results} We provide here some
implications of the information-theoretic bounds and numerical
results which refer to the following issues:
\begin{itemize}
\item Examination of the tightness of the bounds provided in
Section~\ref{Section: main results} by comparing these bounds to
the asymptotic performance of some LDPC code ensembles under BP
decoding (referring here to the sum-product decoding). In order to
make this comparison more conclusive, we compare the new bounds
with previously reported bounds (see Section~\ref{Numerical
Results for Asymptotic Analysis}) in order to exemplify their
practicality.
\item Information-theoretic lower bound on the cardinality of the fundamental
system of cycles of LDPC code ensembles, expressed in terms of the
achievable gap to capacity (see Section~\ref{subsection:
fundamental system of cycles for LDPC ensembles}).
\item Linear programming (LP) bounds on the degree distributions
of capacity-approaching LDPC code ensembles. The bounds refer to
the case where the communication takes place over an MBIOS
channel, as well as universal bounds which are valid for the set
of all MBIOS channels which exhibit a given capacity $C$. These
bounds are valid under ML decoding (and hence, they are also valid
under any sub-optimal decoding algorithm). These LP bounds are
solved analytically, and are also compared with the degree
distributions of capacity-approaching LDPC code ensembles under BP
decoding (see Section~\ref{LP bounds on degree distributions of
LDPC code ensembles}).
\item Lower bounds on the graphical complexity of binary linear
block codes which are represented by an arbitrary bipartite graph
and whose transmission takes place over an MBIOS channel. The
graphical complexity is measured by the total number of edges in
the graph, and the bound provides a quantitative measure of the
minimal number of edges required for this graphical representation
as a function of the target block error probability and the gap
(in rate) to capacity. This bound refers to codes of
finite-length, and is valid under ML decoding (or any sub-optimal
decoding). It can be also applied to LDPC code ensembles, and then
it provides a lower bound on the decoding complexity per iteration
of a BP decoder. Comparison of the information-theoretic lower
bound on the graphical complexity in terms of the achievable gap
to capacity with a target block error probability with some
efficient finite-length LDPC codes which are provided in the
literature enables to evaluate the maximal potential gain that can
be attained by future design of such finite-length codes in terms
of the tradeoff between performance and graphical complexity (see
Section~\ref{Numerical Results for Finite-Length Analysis}).
\end{itemize}

\subsection{Numerical Results for the Asymptotic Analysis under BP Decoding}
\label{Numerical Results for Asymptotic Analysis}

The following sub-section relies on the theoretic results provided
in Section~\ref{Section: main results}, and it exemplifies the use
of these results in the context of capacity-approaching sequences
of LDPC code ensembles whose transmission takes place over an
MBIOS channel, and whose bit error probability vanishes under BP
decoding. As representatives of MBIOS channels, the considered
communication channels are the binary erasure channel (BEC),
binary symmetric channel (BSC) and the binary-input AWGN channel
(BIAWGNC) (as presented in \cite[Example~4.1]{RiU_book}).

\begin{example}{\bf{[BEC]}} Consider a sequence of LDPC code ensembles $(n,
\lambda, \rho)$ where the block length $(n)$ tends to infinity and
the pair of degree distributions is given by
\begin{eqnarray*}
&& \hspace*{-0.7cm} \lambda(x) = 0.409x + 0.202x^2 + 0.0768x^3 +
0.1971x^6 + 0.1151x^7 \\ && \hspace*{-0.7cm} \rho(x) = x^5.
\end{eqnarray*}
The design rate of this ensemble is $R=0.5004$, and the threshold
under BP decoding is (see \cite[Theorem~3.59]{RiU_book})
\begin{equation*}
p^{\text{BP}} = \inf_{x \in (0,1]}
\frac{x}{\lambda\bigl(1-\rho(1-x)\bigr)} = 0.4810
\end{equation*}
so the minimum capacity of a BEC over which it is possible to
transmit with vanishing $P_{\text{b}}$ under BP decoding is
$C=1-p^{\text{BP}} = 0.5190$ bits per channel use, and the
multiplicative gap to capacity is $\varepsilon = 1-\frac{R}{C} =
0.0358$. The lower bound on the average right degree in \eqref{eq:
lower bound on a_R for the BEC with finite P_b} with vanishing bit
erasure probability (i.e., $P_{\text{b}}=0$) gives that the
average right degree should be at least 5.0189. By imposing a
prior assumption that the LDPC code ensemble has a fixed right
degree (as is the case with the above LDPC code ensemble), then it
follows that this right degree cannot be below~6. Hence, the lower
bound is attained in this case with equality. An upper bound on
the fraction of edges which are connected to degree-2 variable
nodes $(\lambda_2)$ is calculated from \eqref{eq: upper bound on
lambda_2 in terms of a_R} with $ \mathcal{B}(a) =
p^{\text{BP}}=0.4810 $, and the above lower bound on
$a_{\text{R}}$ (for LDPC code ensembles of a fixed right degree)
which is equal to~6; this gives from \eqref{eq: upper bound on
lambda_2 in terms of a_R} that $\lambda_2 \leq 0.4158$ as compared
to the exact value which is equal to~0.409. The exact value of the
fraction of degree-2 variable nodes is
\begin{equation*}
\Lambda_2 = \frac{\lambda_2 \, a_{\text{L}}}{2} =\frac{\lambda_2
\, (1-R) \, a_{\text{R}}}{2} = 0.6130
\end{equation*}
as compared to the upper bound in \eqref{eq: upper bound on
Lambda_2 in terms of a_R}, combined with the tight lower bound
$a_{\text{R}} \geq 6$, which gives $\Lambda_2 \leq 0.6232$. We
note that without the prior assumption about the fixed right
degree, the universal bounds give $a_{\text{R}} \geq 5.0189$ and
$\lambda_2 < 0.5173$ so these bounds are clearly loosened.
\end{example}

\begin{example}{\bf{[Comparison of the lower bound on the average right
degree from Theorem~\ref{Theorem: Lower bound on right degree} and
Discussion~\ref{discussion: extension of Theorem 1 to LDPC
ensembles} with the bound in \cite{Wiechman_Sason}]}} In the
following, we exemplify the practical use of the lower bound on
the average right degree of LDPC code ensembles, as given in
Theorem~\ref{Theorem: Lower bound on right degree} and its
adaptation to LDPC code ensembles in Discussion~\ref{discussion:
extension of Theorem 1 to LDPC ensembles}, and compare it with the
previously reported bound in \cite[Section~IV]{Wiechman_Sason}.
Consider the case where the communications takes place over a
BIAWGNC. The LDPC code ensembles in each sequence are specified by
the following pairs of degree distributions, followed by their
corresponding design rates and thresholds under BP decoding:\\[0.1cm]
Ensemble~1:
\begin{eqnarray*}
&& \lambda(x) = x, \quad \rho(x) = x^{19}, \quad R_{\text{d}} =
0.9000.\\
&& \sigma_{\text{BP}} = 0.4156590.
\end{eqnarray*}
Ensemble~2:
\begin{eqnarray*}
&& \lambda(x) = 0.4012x + 0.5981x^2 + 0.0007x^{29}, \quad \rho(x) = x^{24}\\
&& R_{\text{d}} = 0.9000, \quad \sigma_{\text{BP}} = 0.4741840.
\end{eqnarray*}
These code ensembles are taken from the data base in \cite{LTHC}.
From \cite[Example~4.38]{RiU_book} which expresses the capacity of
the BIAWGNC in terms of the standard deviation $\sigma$ of the
Gaussian noise, the minimum capacity of a BIAWGNC over which it is
possible to communicate with vanishing bit error probability under
BP decoding is~$C = 0.9685$ and 0.9323~bits per channel use for
Ensembles~1 and~2, respectively. The corresponding gap (in rate)
to capacity $\varepsilon = 1-\frac{R_{\text{d}}}{C}$ is equal to
$\varepsilon = 7.07 \cdot 10^{-2}$ and $3.46 \cdot 10^{-2}$,
respectively. Therefore, for the first ensemble which is a (2,20)
regular LDPC code ensemble, the new lower bound on the average
right degree which follows from Discussion~\ref{discussion:
extension of Theorem 1 to LDPC ensembles} is equal to~9.949
whereas the lower bound from \cite[Section~IV]{Wiechman_Sason}
(i.e., the un-numbered equation before
\cite[Eq.~(77)]{Wiechman_Sason}) is equal to 2.392. For the second
ensemble whose fixed right degree is equal to~25, the new lower
bound on the average right degree is~16.269 whereas the lower
bound from \cite{Wiechman_Sason} is~14.788. This shows that the
improvement obtained in Theorem~\ref{Theorem: Lower bound on right
degree} followed by Discussion~\ref{discussion: extension of
Theorem 1 to LDPC ensembles} is of practical use.

We note that the gap which still exists between the lower bounds
on the average right degrees and the actual values of
$a_{\text{R}}$ for the above two ensembles is partially attributed
to the fact that this information-theoretic lower bound holds even
under ML decoding, although we apply this bound here under the
sub-optimal BP decoding algorithm. The gaps to capacity under ML
decoding are smaller than those calculated under BP decoding, and
smaller values of $\varepsilon$ provide improved lower bounds on
$a_{\text{R}}$. \label{example: comparison of numerical results
for the new and old lower bounds}
\end{example}

\begin{example}{\bf{[BIAWGNC]}}
Table~\ref{Table: Ensembles over BIAWGNC} considers two sequences
of LDPC code ensembles of design rate $\frac{1}{2}$ which are
taken from \cite[Table~II]{ChungFRU_CommL01}. The transmission of
these ensembles is assumed to take place over the BIAWGNC. The
pair of degree distributions of the ensembles in each sequence is
fixed and the block length of these ensembles tends to infinity.
The LDPC code ensembles in each sequence are specified by the
following
pairs of degree distributions:\\[0.1cm]
Ensemble~1:
\begin{eqnarray*}
&& \lambda(x) = 0.170031 x + 0.160460 x^2 + 0.112837x^5\\
&& \hspace*{1.2cm} + 0.047489x^6+0.011481x^9
+\,0.091537x^{10}\\
&& \hspace*{1.2cm}
+0.152978x^{25}+0.036131x^{26}+0.217056x^{99}\\[0.1cm]
&& \rho(x) = \frac{1}{16} \, x^9 + \frac{15}{16} \, x^{10}.
\end{eqnarray*}
Ensemble~2:
\begin{eqnarray*}
&& \lambda(x) = 0.153425 x + 0.147526 x^2 + 0.041539x^5 \\
&& \hspace*{1.2cm} + 0.147551x^6 + 0.047938x^{17}
+\,0.119555x^{18} \\
&& \hspace*{1.2cm} +0.036379x^{54}+0.126714x^{55}+0.179373x^{199}\\
&& \rho(x) = x^{11}.
\end{eqnarray*}
The asymptotic thresholds of the considered LDPC code ensembles
under BP decoding are calculated with the DE technique, and these
calculations provide the thresholds $\sigma_{\text{BP}} = 0.97592$
and~0.97704, respectively. The minimum capacity of a BIAWGNC which
enables to communicate Ensembles~1 and~2 with vanishing bit error
probability under BP decoding is therefore~$C = 0.5019$ and
0.5011~bits per channel use, respectively (it is calculated via
the power series expansion of the capacity of a BIAWGNC as given
in \cite[page~194]{RiU_book}). This leads to the indicated gaps
(in rate) to capacity as given in Table~\ref{Table: Ensembles over
BIAWGNC}.
\begin{table}
\caption{Bounds vs. exact values of $\lambda_2$ and $a_{\text{R}}$
for two sequences of LDPC code ensembles of design rate
$\frac{1}{2}$ transmitted over the BIAWGNC. The sequences are
given in \cite[Table~II]{ChungFRU_CommL01} and achieve vanishing
bit error probability under the belief propagation (BP) decoding
algorithm with the indicated gaps to capacity.}
\begin{center}
\begin{tabular}{|c|c|c|c|c|c|} \hline
LDPC & Gap to & & Lower bound & & Upper bound\\
ense- & capacity &$a_{\text{R}}$ & on $a_{\text{R}}$ & $\lambda_2$
& on $\lambda_2$\\
mble & ($\varepsilon$) & & (Theorem~\ref{Theorem: Lower bound on
right degree}) & & (Theorem~\ref{Theorem: degree-2 variable nodes}) \\
\hline \hline 1 & $3.72 \cdot 10^{-3}$ & 10.938 & 9.249 & 0.170 &
0.205\\ \hline 2 & $2.22 \cdot 10^{-3}$ & 12.000 & 10.129 & 0.153
& 0.185\\ \hline
\end{tabular}
\end{center}
\label{Table: Ensembles over BIAWGNC}
\end{table}
The value of $\lambda_2$ for each sequence of LDPC code ensembles
(where we let the block length tend to infinity) is compared with
the upper bound in Theorem~\ref{Theorem: degree-2 variable nodes}
which corresponds to BP decoding. Note that for calculating the
bound in Theorem~\ref{Theorem: degree-2 variable nodes}, the
Bhattacharyya constant in \eqref{eq: definition of Bhattacharyya
constant} is given by $\mathcal{B}(a) = \exp\bigl(-\frac{R
E_{\text{b}}}{N_0}\bigr)$ for the BIAWGNC where
$\frac{E_{\text{b}}}{N_0}$ designates the energy per information
bit over the one-sided noise spectral density, and we substitute
here the threshold value of $\frac{E_{\text{b}}}{N_0}$ under BP
decoding. The average right degree of each sequence is also
compared with the lower bound in Theorem~\ref{Theorem: Lower bound
on right degree}. These comparisons exemplify that for the
examined LDPC code ensembles, both of the theoretical bounds are
informative.
\end{example}

\begin{example}{\bf{[BSC]}}
Table~\ref{Table: Ensembles over BSC} considers two sequences of
LDPC code ensembles, taken from \cite{LTHC}, where the pair of
degree distributions of the ensembles in each sequence is fixed
and the block length of these ensembles tends to infinity. The
transmission of these ensembles is assumed to take place over the
BSC. The LDPC code ensembles in each sequence are specified by the
following pairs of
degree distributions and design rates:\\[0.1cm]
Ensemble~1:
\begin{eqnarray*}
&& \hspace*{-0.3cm} \lambda(x) = 0.291157x +
0.189174x^2+0.0408389x^4\\ && \hspace*{0.8cm}
+0.0873393x^5+0.00742718x^6 +0.112581x^7\\ &&
\hspace*{0.8cm} +0.0925954x^{15}+0.0186572x^{20}+0.124064x^{32}\\
&&\hspace*{0.8cm}+0.016002x^{39} + 0.0201644x^{44} \\
&& \hspace*{-0.3cm} \rho(x) = 0.8x^4+0.2x^5 \\
&& \hspace*{-0.3cm} R = 0.250
\end{eqnarray*}
Ensemble~2:
\begin{eqnarray*}
&& \hspace*{-0.3cm} \lambda(x) = 0.160424x+0.160541x^2+0.0610339x^5\\
&& \hspace*{0.8cm}+0.153434x^6+0.0369041x^{12}+ 0.020068x^{15}\\
&& \hspace*{0.8cm} + 0.0054856x^{16} +
0.128127x^{19}+0.0233812x^{24}\\ &&
\hspace*{0.8cm}+0.05285542x^{34}+0.0574104x^{67} + 0.0898442x^{68}\\
&& \hspace*{0.8cm}
+ 0.0504923x^{85} \\
&& \hspace*{-0.3cm} \rho(x) = x^{10} \\
&& \hspace*{-0.3cm} R = 0.500.
\end{eqnarray*}
The thresholds of the above LDPC code ensembles under BP decoding
are equal to $p_{\text{BSC}} = 0.2120$ and~0.1090, respectively.
Hence, for Ensembles~1 and~2, the minimum capacity of a BSC which
enables to communicate with vanishing bit error probability under
BP decoding is~$C = 0.2547$ and 0.5031~bits per channel use. Since
of the design rates of these two ensembles are~$R_{\text{d}} =
0.250$ and~0.500, respectively, then the gaps to capacity are
given in Table~\ref{Table: Ensembles over BSC}.
\begin{table}
\caption{Comparison of theoretical bounds and actual values of
$\lambda_2$ and $a_{\text{R}}$ for two sequences of LDPC code
ensembles transmitted over the BSC. The sequences are taken from
\cite{LTHC} and achieve vanishing bit error probability under the
belief propagation (BP) decoding algorithm with the indicated gaps
to capacity.}
\begin{center}
\begin{tabular}{|c|c|c|c|c|c|} \hline
LDPC & Gap to & & Lower bound & & Upper bound\\
ense- & capacity &$a_{\text{R}}$ & on $a_{\text{R}}$ & $\lambda_2$
& on $\lambda_2$\\
mble & ($\varepsilon$) & & (Theorem~\ref{Theorem: Lower bound on
right degree}) & & (Theorem~\ref{Theorem: degree-2 variable nodes}) \\
\hline \hline 1 & $1.85 \cdot 10^{-2}$ & 5.172 & 4.301 & 0.291 & 0.371\\
\hline 2 & $6.18 \cdot 10^{-3}$ & 11.000 & 9.670 & 0.160 & 0.185
\\\hline
\end{tabular}
\end{center}
\label{Table: Ensembles over BSC}
\end{table}
The value of $\lambda_2$ for each sequence is compared with the
upper bound given in Theorem~\ref{Theorem: degree-2 variable
nodes}. Note that for calculating the bound in
Theorem~\ref{Theorem: degree-2 variable nodes}, the Bhattacharyya
constant $\mathcal{B}(a)$ introduced in \eqref{eq: definition of
Bhattacharyya constant} satisfies $\mathcal{B}(a) =
\sqrt{4p(1-p)}$ for a BSC whose crossover probability is equal to
$p$, and we substitute here the threshold value of $p$ under BP
decoding. Also, for the calculation of this bound for such a BSC,
Eq.~\eqref{eq: g_k for BSC} gives that $g_1 = (1-2p)^2$. The
average right degree of each sequence is also compared with the
lower bound in Theorem~\ref{Theorem: Lower bound on right degree}.
These comparisons show that for the considered sequences of LDPC
code ensembles, both of the theoretical bounds are fairly tight;
the upper bound on $\lambda_2$ is within a factor of 1.3 from the
actual value for the two sequences of LDPC code ensembles while
the lower bound on the average right degree is not lower than
$83\%$ of the corresponding actual values. The LDPC code ensembles
referred to in Table~\ref{Table: Ensembles over BSC} were obtained
in \cite{LTHC} by the DE technique with the goal of minimizing the
gap to capacity under a constraint on the maximal degree.
\end{example}

\begin{figure*}
\begin{center}
\epsfig{file=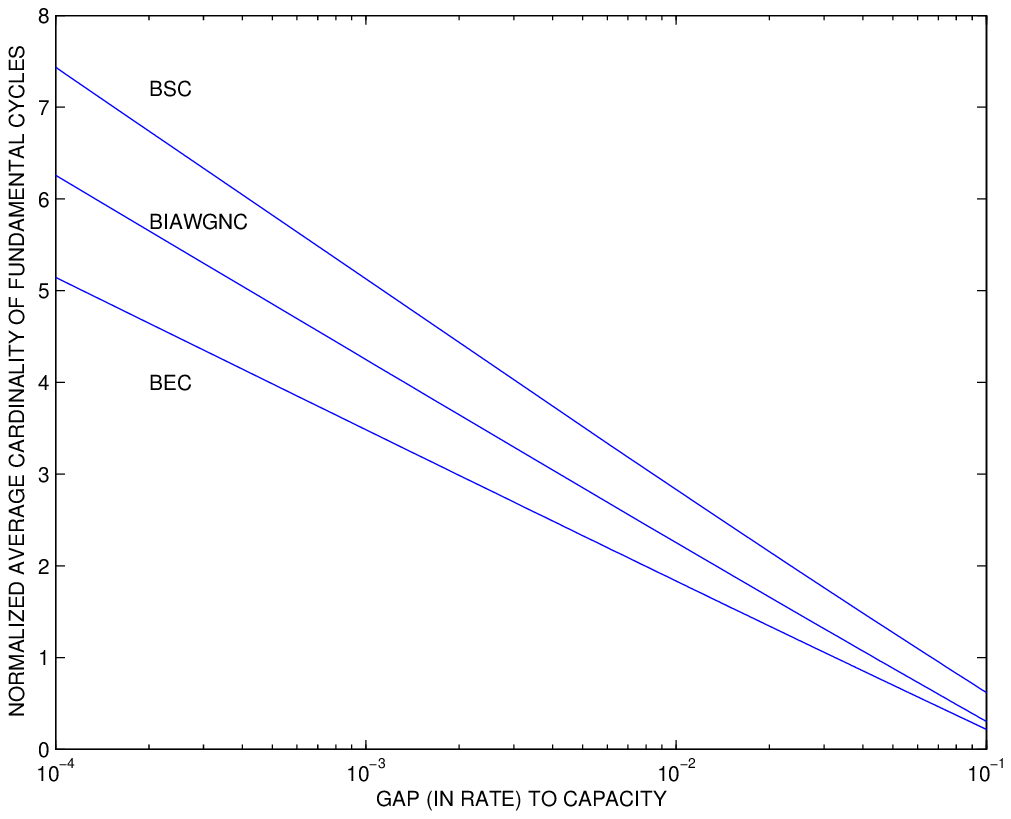,scale=0.60}
\end{center}
\caption{\label{Figure: fundamental system of cycles} Plot of the
asymptotic lower bounds in Corollary~\ref{corollary: lower bound
on the cardinality of the fundamental system of cycles of LDPC
ensembles} (see Eqs.~\eqref{eq: lower bound on the cardinality of
the fundamental system of cycles of LDPC ensembles for general
MBIOS channels} and \eqref{eq: tightened lower bound on the
cardinality of the fundamental system of cycles of LDPC ensembles
for the BEC}) for memoryless binary-input output-symmetric (MBIOS)
channels. These lower bounds correspond to the average cardinality
of the fundamental system of cycles for bipartite graphs
representing codes from an arbitrary LDPC code ensemble; the above
quantity is normalized with respect to the block length of the
ensemble, and the asymptotic result refers to the case where we
consider a sequence of LDPC code ensembles whose block lengths
tend to infinity. The bounds are plotted versus the achievable gap
(in rate) between the channel capacity and the design rate of the
LDPC code ensembles. This figure shows the bounds for the binary
symmetric channel (BSC), binary-input AWGN channel (BIAWGNC) and
the binary erasure channel (BEC) where it is assumed that the
design rate of the LDPC code ensembles is equal to one-half bit
per channel use.}
\end{figure*}

\subsection{On the Fundamental System of Cycles for Capacity-Approaching
Sequences of LDPC Code Ensembles} \label{subsection: fundamental
system of cycles for LDPC ensembles} Corollary~\ref{corollary:
lower bound on the cardinality of the fundamental system of cycles
of LDPC ensembles} considers an arbitrary sequence of LDPC code
ensembles, specified by a pair of degree distributions, whose
transmission takes place over an MBIOS channel. This corollary
refers to the asymptotic case where we let the block length of the
ensembles in this sequence tend to infinity and the bit error (or
erasure) probability vanishes; the design rate of these ensembles
is assumed to be a fraction $1-\varepsilon$ of the channel
capacity (for an arbitrary $\varepsilon \in (0,1)$). In
Corollary~\ref{corollary: lower bound on the cardinality of the
fundamental system of cycles of  LDPC ensembles}, Eq.~\eqref{eq:
lower bound on the cardinality of the fundamental system of cycles
of LDPC ensembles for general MBIOS channels} applies to a general
MBIOS channel and a tightened version of this bound is given in
\eqref{eq: tightened lower bound on the cardinality of the
fundamental system of cycles of LDPC ensembles for the BEC} for
the BEC. Based on these results, the asymptotic average
cardinality of the fundamental system of cycles for bipartite
graphs representing codes from LDPC code ensembles as above, where
this average cardinality is normalized with respect to the block
length, grows at least like $\ln \frac{1}{\varepsilon}$. We
consider here the BSC, BEC, and BIAWGNC as three representatives
of the class of MBIOS channels, and assume that the design rate of
the LDPC code ensembles is fixed to one-half bit per channel use.
It is shown in Fig.~\ref{Figure: fundamental system of cycles}
that for a given gap $(\varepsilon)$ to the channel capacity and
for a fixed design rate, the extreme values of this lower bounds
correspond to the BSC and BEC (which attain the maximal and
minimal values, respectively). This observation is consistent with
the last part of the statement in Corollary~\ref{corollary: lower
bound on the cardinality of the fundamental system of cycles of
LDPC ensembles}. \label{subsection: Fundamental System of Cycles
for Capacity-Approaching LDPC Code Ensembles}

\subsection{Linear Programming Bounds for the Degree Distributions of
LDPC Code Ensembles} \label{LP bounds on degree distributions of
LDPC code ensembles} This sub-section provides LP bounds on the
degree distributions of LDPC code ensembles. These bounds, which
are based on Sections~\ref{Section: main results}
and~\ref{Section: Proofs of Main Results}, are formulated in terms
of the target bit error probability and the gap (in rate) to
capacity required to achieve this target. The following LP bounds
refer to the node and the edge perspectives of the pair of degree
distributions, and they provide upper bounds on the fraction of
edges or nodes up to degree $k$ where $k$ is a parameter.
Similarly to Theorem~\ref{Theorem: Behavior of left and right
degrees}, the LP bounds which are introduced in this section hold
under ML decoding, and are therefore general in terms of the
decoding algorithm. These LP bounds apply to finite-length LDPC
code ensembles as well as to the asymptotic case of an infinite
block length. Analytical solutions for these LP bounds are
provided in Section~\ref{LP bounds on degree distributions of LDPC
code ensembles}, and these bounds are also compared with some
capacity-achieving sequences of LDPC code ensembles for the BEC
under BP decoding. The following LP bounds are separated into four
categories:
\begin{itemize}
\item {\bf{LP1}}: 'LP1' forms an LP upper bound on the degree
distribution of the parity-check nodes for LDPC code ensembles whose
transmission takes place over an MBIOS channel. Its first version
gives an upper bound on the fraction of parity-check nodes up to
degree $k$ (where $k \geq 1$ is an integer) as a function of the
achievable rate (and its gap to the channel capacity) with a given
bit error probability $P_{\text{b}}$. By combining \eqref{eq:
relationship between Gamma epsilon and P_b} with the trivial
constraints for an arbitrary degree distribution, the following
optimization problem follows:
\begin{equation}
\mbox{\fbox{$
\begin{array}{l}
\text{maximize} \; \; \sum\limits_{i=1}^k \Gamma_i, \quad k=1, 2, \ldots \nonumber \\
\text{subject to} \nonumber \\[0.1cm]
\begin{cases}
\hspace{0.2cm} \sum\limits_{i=1}^{\infty}
\left\{\left[1-h_2\biggl(\frac{1-g_1^{\frac{i}{2}}}{2}\biggr)\right]
\Gamma_i \right\} \leq \frac{\varepsilon\,C +
h_2(P_{\text{b}})}{1-(1-\varepsilon)C} \hspace{-0.2cm} \\ \\[-0.1cm]
\hspace{0.2cm} \sum\limits_{i=1}^{\infty} \Gamma_i = 1 \\ \\[-0.1cm]
\hspace{0.2cm} \Gamma_i \geq 0, \quad i=1, 2, \ldots
\end{cases}
\end{array}
$}}
\end{equation}
where the optimization variables are $\{\Gamma_i\}_{i \geq 1}$.
From~\eqref{switching between representations_1}, the following
equality holds:
\begin{equation}
\Gamma_i = \frac{\rho_i}{i}
\Bigl(\sum\limits_{j=1}^{\infty}\frac{\rho_j}{j}\Bigr)^{-1}.
\label{eq: switching between representations_2}
\end{equation}
The substitution of this equality in the first constraint of the
above LP bound gives the following optimization problem for the
degree distribution of the parity-check nodes from the edge
perspective (i.e., we get an upper bound on the fraction of edges
which are connected to parity-check nodes up to degree $k \geq
1$):
\begin{equation}
\mbox{\fbox{$
\begin{array}{l}
\text{maximize} \; \; \sum\limits_{i=1}^k \rho_i, \quad k=1, 2, \ldots \nonumber \\
\text{subject to} \nonumber \\[0.1cm]
\begin{cases}
\hspace{0.2cm} \sum\limits_{i=1}^{\infty}
\left\{\left[1-h_2\biggl(\frac{1-g_1^{\frac{i}{2}}}{2}\biggr)\right]
\frac{\rho_i}{i} \right\} \leq \frac{\varepsilon\,C +
h_2(P_{\text{b}})}{1-(1-\varepsilon)C} \hspace{-0.2cm} \; \sum\limits_{i=1}^{\infty} \frac{\rho_i}{i} \hspace{-0.3cm} \\ \\[-0.1cm]
\hspace{0.2cm} \sum\limits_{i=1}^{\infty} \rho_i = 1 \\ \\[-0.1cm] \hspace{0.2cm} \rho_i \geq 0,
\quad i=1, 2, \ldots
\end{cases}
\end{array}
$}}
\end{equation}
where the optimization variables are $\{\rho_i\}_{i \geq 1}$.
These two LP bounds on the parity-check degree distribution (from
the node and edge perspectives) rely both on
Theorems~\ref{Theorem: Lower bound on right degree} and
\ref{Theorem: Behavior of left and right degrees}, and are
therefore valid under ML decoding (hence, they also hold under any
other decoding algorithm). These bounds hold for finite-length
codes and also for the asymptotic case of an infinite block
length.

An analytical solution of the LP1 bound is given in
Appendix~\ref{Appendix: solution of the LP1 bound}. This bound is
tightened in Appendix~\ref{Appendix: solution of the LP1 bound}
for the BEC, followed by its analytical solution.

\item {\bf{LP2}}: 'LP2' provides a universal LP upper bound on the
degree distribution of the parity-check nodes for LDPC code
ensembles as a function of the required achievable rate (and its gap
to the channel capacity) with a required bit error probability
$P_{\text{b}}$. This bound follows from \eqref{eq: universal
relationship between Gamma epsilon and P_b} and \eqref{eq: switching
between representations_2}, and it gets the form:
\begin{equation}
\mbox{\fbox{$
\begin{array}{l}
\text{maximize} \; \; \sum\limits_{i=1}^k \rho_i, \quad k=1, 2, \ldots \nonumber \\
\text{subject to} \nonumber \\[0.1cm]
\begin{cases}
\hspace{0.2cm} \sum\limits_{i=1}^{\infty}
\left\{\left[1-h_2\biggl(\frac{1-C^{\frac{i}{2}}}{2}\biggr)\right]
\frac{\rho_i}{i} \right\} \leq \frac{\varepsilon\,C +
h_2(P_{\text{b}})}{1-(1-\varepsilon)C} \hspace{-0.2cm} \, \sum\limits_{i=1}^{\infty} \frac{\rho_i}{i} \hspace{-0.3cm} \\ \\[-0.1cm]
\hspace{0.2cm} \sum\limits_{i=1}^{\infty} \rho_i = 1 \\ \\[-0.1cm] \hspace{0.2cm} \rho_i \geq 0,
\quad i=1, 2, \ldots
\end{cases}
\end{array}
$}}
\end{equation}
where the optimization variables are $\{\rho_i\}_{i \geq 1}$, and
the bound holds under the same conditions as of the previous item.
However, as opposed to the LP1 bound, the LP2 bound is universal
since it holds for all MBIOS channels which exhibit a given capacity
$C$. Note that the LP2 bound is similar to the LP1 bound, except of
replacing the parameter $g_1$ in the LP1 bound with the channel
capacity $C$. This follows directly by comparing \eqref{eq:
relationship between Gamma epsilon and P_b} and \eqref{eq: universal
relationship between Gamma epsilon and P_b}. Note that the
transition from \eqref{eq: relationship between Gamma epsilon and
P_b} to \eqref{eq: universal relationship between Gamma epsilon and
P_b} follows from Lemma~\ref{lemma: Extreme values of g_1} which
implies that among all MBIOS channels with a given capacity $C$, the
channel which attains the minimal value of $g_1$ is the BEC, and the
minimal value of $g_1$ is equal to $C$.

The analytical solution of the LP2 bound follows directly from the
analysis in Appendix~\ref{Appendix: solution of the LP1 bound} for
the LP1 bound, by replacing $g_1$ in the LP1 bound with the channel
capacity $C$ in the LP2 bound.

\item {\bf{LP3}}: 'LP3' provides an LP upper bound on the degree distribution of the variable
nodes (from the edge perspective) for LDPC code ensembles whose
transmission takes place over an MBIOS channel. This bound
provides an upper bound on the fraction of edges which are
connected to variable nodes up to degree $k$ for a parameter $k
\geq 2$, and it is expressed in terms of the required achievable
rate (and its gap to capacity) with a given bit error probability
$P_{\text{b}}$. From \eqref{design rate of LDPC ensemble} and
\eqref{eq: lower bound on a_R with finite P_b}, this LP bound gets
the form
\begin{equation}
\mbox{\fbox{$
\begin{array}{l}
\text{maximize} \; \; \sum\limits_{i=2}^k \lambda_i, \quad k=2, 3, \ldots \nonumber \\
\text{subject to} \nonumber \\[0.1cm]
\begin{cases}
\sum\limits_{i=2}^{\infty} \frac{\lambda_i}{i} \leq
\frac{\ln\bigl(\frac{1}{g_1}\bigr)}{2(1-C)
\left(1+\frac{\varepsilon C}{1-C}\right) \ln\biggl(\frac{1}{1-2
h_2^{-1}\bigl(\frac{1-C-h_2(P_{\text{b}})}
{1-(1-\varepsilon)C}\bigr)}\biggr)}
\hspace{-0.2cm} \\ \\[-0.1cm]
\sum\limits_{i=2}^{\infty} \lambda_i = 1 \\ \\[-0.1cm] \lambda_i \geq 0,
\quad i=2, 3, \ldots
\end{cases}
\end{array}
$}}
\end{equation}
where the optimization variables are $\{\lambda_i\}_{i \geq 2}$.
Since the bound relies on Theorem~\ref{Theorem: Lower bound on
right degree}, then it is therefore valid under ML decoding (or
any other decoding algorithm). It holds for finite block-length as
well as in the asymptotic case where we let the block length tend
to infinity. We note that the focus on the degree distribution of
the variable nodes from the edge perspective is due to
Theorem~\ref{Theorem: Behavior of left and right degrees} and
Remark~\ref{remark: on the fraction of degree-2 variable nodes for
c.a. sequences over the BEC} (see p.~\pageref{remark: on the
fraction of degree-2 variable nodes for c.a. sequences over the
BEC}).
\item {\bf{LP4}}: 'LP4' provides a universal LP upper bound on the
degree distribution of the variable nodes for LDPC code ensembles
(from the edge perspective). It is based on \eqref{design rate of
LDPC ensemble} and \eqref{eq: universal lower bound on a_R with
finite P_b} which give the following problem:
\begin{equation}
\mbox{\fbox{$
\begin{array}{l}
\text{maximize} \; \; \sum\limits_{i=2}^k \lambda_i, \quad k=2, 3, \ldots \nonumber \\
\text{subject to} \nonumber \\[0.1cm]
\begin{cases}
\sum\limits_{i=2}^{\infty} \frac{\lambda_i}{i} \leq
\frac{\ln\bigl(\frac{1}{C}\bigr)}{2(1-C) \left(1+\frac{\varepsilon
C}{1-C}\right) \ln\biggl(\frac{1}{1-2
h_2^{-1}\bigl(\frac{1-C-h_2(P_{\text{b}})}
{1-(1-\varepsilon)C}\bigr)}\biggr)}
\hspace{-0.2cm} \\ \\[-0.1cm]
\sum\limits_{i=2}^{\infty} \lambda_i = 1 \\ \\[-0.1cm] \lambda_i \geq 0,
\quad i=2, 3, \ldots
\end{cases}
\end{array}
$}}
\end{equation}
where the optimization variables are $\{\lambda_i\}_{i \geq 2}$.
This bound holds for all MBIOS channels with a given capacity $C$.
\end{itemize}

\begin{figure*}
\begin{center}
\epsfig{file=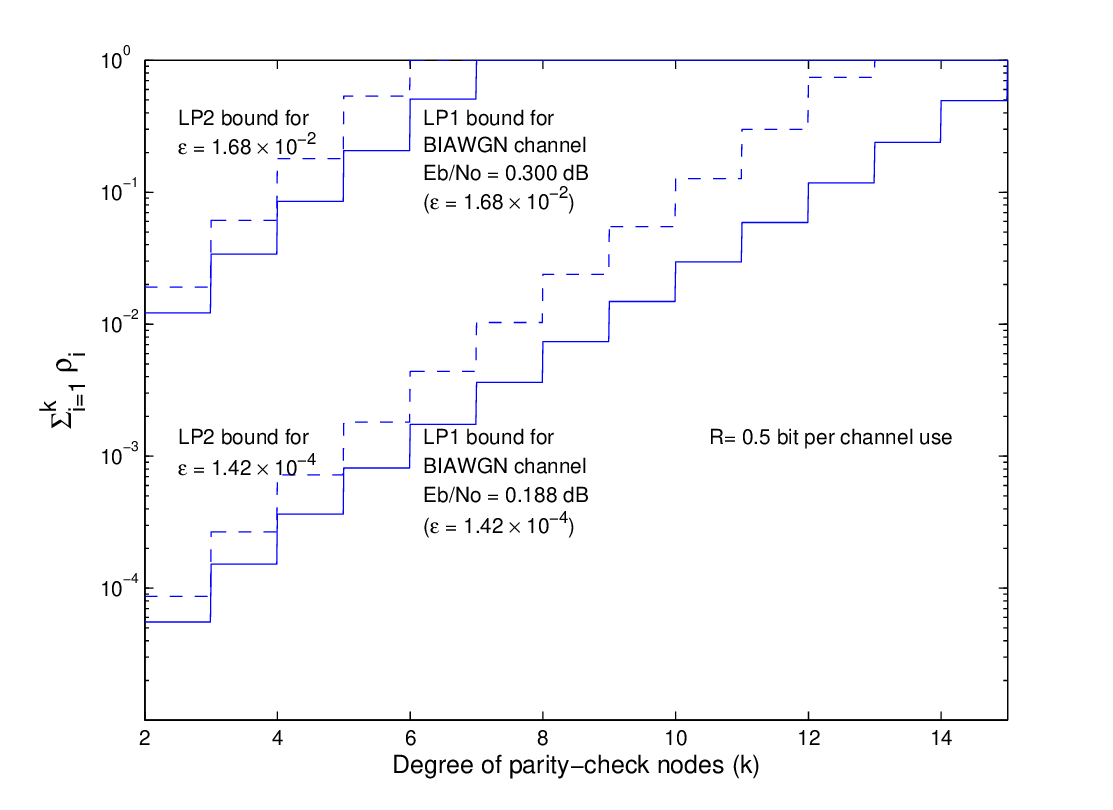,scale=0.60}
\end{center}
\caption{\label{Figure: LP1 vs LP2.eps} LP1 versus LP2 upper bounds
on the degree distributions of the parity-check nodes, from the edge
perspective, for LDPC code ensembles whose design rate is $R =
\frac{1}{2}$. The stair functions show upper bounds on the fraction
of the edges which are connected to parity-check nodes whose degrees
are at most $k$ for an integer $k \geq 2$. The bounds are valid
under ML decoding or any sub-optimal decoding algorithm. All these
curves refer to a target bit error probability of $P_{\text{b}} =
10^{-10}$. The two LP1 bounds (solid lines) refer to binary-input
AWGN (BIAWGN) channels for which $\frac{E_{\text{b}}}{N_0}$ = 0.300
and 0.188~dB, so the corresponding channel capacities are $C =
0.5086$ and 0.5001 bits per channel use, respectively; the
corresponding gaps (in rate) to capacity are therefore equal to
$\varepsilon = 1.68\cdot 10^{-2}$ and $1.42\cdot 10^{-4}$,
respectively. The two universal LP2 bounds (dashed lines) correspond
to all the MBIOS channels which exhibit a given capacity, whose
value coincides in each case with the capacity of the considered
BIAWGN channel.}
\end{figure*}

The universal (LP2) bound is compared in Figure~\ref{Figure: LP1
vs LP2.eps} to the LP1 bound for the BIAWGN channel with the same
capacity. It is shown in this figure that the difference between
these two bounds is not large. Note that the universal bound is
attained for the BEC with the same capacity as of the BIAWGN
channel.

\begin{remark}{\bf{[A discussion on the constraints given in the LP1
and LP2 bounds and the un-necessity of adding the constraint in
Theorem~\ref{Theorem: Lower bound on right degree}]}} We prove in
Appendix~\ref{Appendix: Proof for the un-necessity of adding the
additional constraint in the LP1 and LP2 bounds} that adding the
constraint which is imposed by the lower bound on the average
right degree (i.e., the lower bound on $a_{\text{R}} =
\sum_{i=1}^{\infty} i \Gamma_i$) does not affect the LP1 and LP2
bounds introduced here. This simplifies the formulation of the LP
bounds serves for the derivation of closed-form analytical
solutions of these bounds later in this section. \label{Remark: A
discussion on the constraints given in the LP1 and LP2 bounds}
\end{remark}

\begin{remark}{\bf{[The LP1 and LP2 bounds and their connection
with the asymptotic behavior as given in Theorem~\ref{Theorem:
Behavior of left and right degrees}]}} As shown via the upper
bounds in Fig.~\ref{Figure: LP1 vs LP2.eps}, the fraction of edges
which are connected to parity-check nodes of low degree is small,
especially when the achievable gap to capacity vanishes. This is
consistent with the theoretical result in Theorem~\ref{Theorem:
Behavior of left and right degrees} and Corollary~\ref{Corollary:
Behavior of the left and right degrees when P_b=0} which states
that the fraction of parity-check nodes of any finite degree
scales at most like $\varepsilon$ and the fraction of edges
connected to parity-check nodes of any finite degree scales at
most like $\frac{\varepsilon}{\ln \frac{1}{\varepsilon}}$ where
$\varepsilon$ designates the gap in rate to capacity, so both
quantities tend to zero as the gap to capacity vanishes.
\end{remark}

\vspace*{0.2cm} For solving the LP1 and LP2 bounds which are
introduced in this section we originally used \cite{CVX_software},
a package for specifying and solving convex optimization problems
\cite{CVX_book}. It enables to solve these problems on a standard
PC in a fraction of a second. However, it is still nice to get an
analytic solution of these LP bounds.

\vspace*{0.1cm} {\bf{Analytical solutions for the LP1 and LP2
bounds}}: The LP1 problem can be expressed in the following
equivalent form:
\begin{equation}
\mbox{\fbox{$
\begin{array}{l}
\text{maximize} \; \; \sum\limits_{i=1}^k \rho_i, \quad k=1, 2, \ldots \nonumber \\
\text{subject to} \nonumber \\[0.1cm]
\begin{cases}
\hspace{0.2cm} \sum\limits_{i=1}^{\infty} d_i \rho_i \leq 0 \\
\hspace{0.2cm} d_i \triangleq \frac{1}{i} \left[
1-h_2\biggl(\frac{1-g_1^{\frac{i}{2}}}{2}\biggr) -
\frac{\varepsilon\,C + h_2(P_{\text{b}})}{1-(1-\varepsilon)C}
\right] \\[0.15cm]
\hspace{0.2cm} \sum\limits_{i=1}^{\infty} \rho_i = 1 \\[0.15cm]
\hspace{0.2cm} \rho_i \geq 0, \quad i=1, 2, \ldots
\end{cases}
\end{array}
$}}
\end{equation}
An analytical solution for the LP1 bound is obtained in
Appendix~\ref{Appendix: solution of the LP1 bound} (via the use of
strong Lagrange duality).

In the following, the final solution of the LP1 bound is
presented. To this end, note that for indices $i$ large enough,
$d_i < 0$ and also $\lim_{i \rightarrow \infty} d_i = 0$.  Let
$d^* \triangleq \min_{i \geq 1} d_i$ be the minimal value of this
sequence, and let $i=l$ be the corresponding index of $d_i$ which
achieves this minimal value of the sequence $\{d_i\}$. Clearly,
$d^* < 0$. The resulting closed-form solution for the LP1 bound
gets the following form (see Appendix~\ref{Appendix: solution of
the LP1 bound}):
\begin{itemize}
\item For values of $k$ below the lower bound
on the average right degree in \eqref{eq: lower bound on a_R with
finite P_b}, it is equal to $-\frac{d^*}{d_k - d^*}$.
\item For values of $k$ larger or equal to the lower bound
on the average right degree in \eqref{eq: lower bound on a_R with
finite P_b}, it is equal to~1.
\end{itemize}
A similar solution is obtained for the LP2 bound where the only
difference is that $g_1$ in the definition of the sequence
$\{d_i\}$ is replaced by the channel capacity $C$. These
analytical solutions match the numerical solutions obtained via
\cite{CVX_software}.

\vspace*{0.2cm}
\begin{example}{\bf{[A comparison of the LP1 bound and capacity-achieving
LDPC code ensembles over the BEC]}} In the following, we compare
the LP1 bound for the BEC and the degree distributions of two
capacity-achieving sequences of LDPC code ensembles under
iterative message-passing decoding.

The first capacity-achieving sequence for the BEC refers to the
heavy-tail Poisson distribution, and it was introduced in
\cite[Section~IV]{LubyMSS_IT01}, \cite{Shokrollahi-IMA2000} (see
also \cite[Problem~3.20]{RiU_book}). The second capacity-achieving
sequence refers to the right-regular LDPC code ensembles
\cite{Shokrollahi-IMA2000}, based also on the analysis in the
proof of Proposition~\ref{Proposition: the tightness of the upper
bound on lambda_2} (see Section~\ref{Section: Proofs of Main
Results}).

This first capacity-achieving sequence is obtained via the pair of
degree distributions
\begin{eqnarray*}
&& \hat{\lambda}_{\alpha}(x) = -\frac{1}{\alpha} \cdot
\ln(1-x) = \frac{1}{\alpha} \sum_{i=1}^{\infty} \frac{x^i}{i} \\
&& \rho_{\alpha}(x) = e^{\alpha (x-1)} = e^{-\alpha}
\sum_{i=0}^{\infty} \frac{\alpha^i x^i}{i!}
\end{eqnarray*}
which satisfies the equality
$\hat{\lambda}_{\alpha}(1-\rho_{\alpha}(1-x))=x$ for all $\alpha >
0$. Starting with the heavy-tail Poisson distribution as above and
proceeding along the lines in \cite[Section~3.15]{RiU_book}, the
following two steps are performed for the construction of
capacity-approaching LDPC code ensembles for the BEC:
\begin{itemize}
\item The degree distribution $\hat{\lambda}_{\alpha}(x)$
is truncated so that it consists of the first $N$ terms of its
Taylor series expansion (up to and including the term $x^{N-1}$).
\item The truncated power series
$\hat{\lambda}_{\alpha}^{(N)}(x)$ is normalized so that it is
equal to~1 at $x=1$. The left degree distribution (from the edge
perspective) is then equal to $\lambda_{\alpha}^{(N)}(x) =
\frac{\hat{\lambda}_{\alpha}^{(N)}(x)}{\hat{\lambda}_{\alpha}^{(N)}(1)}$.
The right degree distribution, $\rho_{\alpha}(x)$, is not
modified.
\end{itemize}
This procedure provides the following degree distributions:
\begin{eqnarray}
&& \lambda_i = \frac{1}{H(N-1)\, (i-1)}, \quad i = 2, 3, \ldots N \nonumber\\
&& \rho_i = \frac{e^{-\alpha} \alpha^{i-1}}{(i-1)!}, \quad \quad
i=1, 2, \ldots \label{eq: heavy-tail Poisson distribution}
\end{eqnarray}
where $H(k) \triangleq \sum_{i=1}^k \frac{1}{i}$ for $k \geq 1$ is
a truncated harmonic sum. From \eqref{design rate of LDPC
ensemble}, straightforward calculus shows that the design rate of
the corresponding LDPC code ensemble is equal to
\begin{eqnarray}
R_{\text{d}}(\alpha, N) &=& 1 - \frac{\int_0^1 \rho_{\alpha}(x) \,
\text{d}x}{\int_0^1 \lambda_{\alpha}^{(N)}(x) \, \text{d}x} \nonumber \\
&=& 1 - \frac{N \, H(N-1) \, (1-e^{-\alpha})}{(N-1)\alpha}.
\label{eq: design rate of heavy-tail Poisson distribution}
\end{eqnarray}
We need to determine the parameters $\alpha$ and $N$ so that the
design rate in \eqref{eq: design rate of heavy-tail Poisson
distribution} forms (at least) a fraction $1-\varepsilon$ of the
capacity of the BEC. Let $p$ designate the erasure probability of
the channel, and let $r = (1-\varepsilon)(1-p)$ be the lower bound
on the required design rate. We need to choose $\alpha$ and $N$ to
satisfy the inequality $R_{\text{d}}(\alpha, N) \geq r$ with
vanishing bit erasure probability under BP decoding. Similarly to
the calculations in \cite[Example~3.88]{RiU_book}, the
satisfiability of the inequality
\begin{equation*}
\frac{\hat{\lambda}_{\alpha}^{(N)}(1)}{1-\hat{\lambda}_{\alpha}^{(N)}(1)}
\left(\frac{\int_0^1 \rho_{\alpha}(x) \, \text{d}x}{\int_0^1
\hat{\lambda}_{\alpha}^{(N)}(x) \, \text{d}x} - 1 \right) \leq
\varepsilon
\end{equation*}
implies this requirement, and straightforward algebra gives the
inequality
\begin{equation}
\frac{\frac{H(N-1)}{\alpha}}{1-\frac{H(N-1)}{\alpha}}
\left(\frac{N (1-e^{-\alpha})}{N-1} -1 \right) \leq \varepsilon.
\label{eq: inequality for choosing alpha and N for the heavy-tail
Poisson distribution}
\end{equation}
By choosing $\alpha$ to satisfy the equality
$\frac{H(N-1)}{\alpha} = 1-r$ and replacing $1-e^{-\alpha}$ by~1,
we get from \eqref{eq: inequality for choosing alpha and N for the
heavy-tail Poisson distribution} the following stronger
requirement:
\begin{equation}
\frac{1-r}{r} \, \frac{1}{N-1} \leq \varepsilon
\end{equation}
which then provides a proper choice for $N$. To conclude, the
parameters $\alpha$ and $N$ are chosen to be
\begin{equation}
\alpha = \frac{H(N-1)}{1-r}, \quad N = \biggl\lceil
\frac{1-r}{\varepsilon r} \biggr\rceil + 1. \label{eq: choice of
parameters for the heavy-tail Poisson distribution}
\end{equation}
In the following, we calculate the heavy-tail Poisson distribution
in \eqref{eq: heavy-tail Poisson distribution} with the choice of
parameters in \eqref{eq: choice of parameters for the heavy-tail
Poisson distribution}. The resulting degree distribution of the
parity-check nodes (from the edge perspective) is compared with
the LP1 bound for the BEC where the analytical solution of this
bound is given in Appendix~\ref{Appendix: solution of the LP1
bound}.

\begin{figure*}[here]
\begin{center}
\epsfig{file=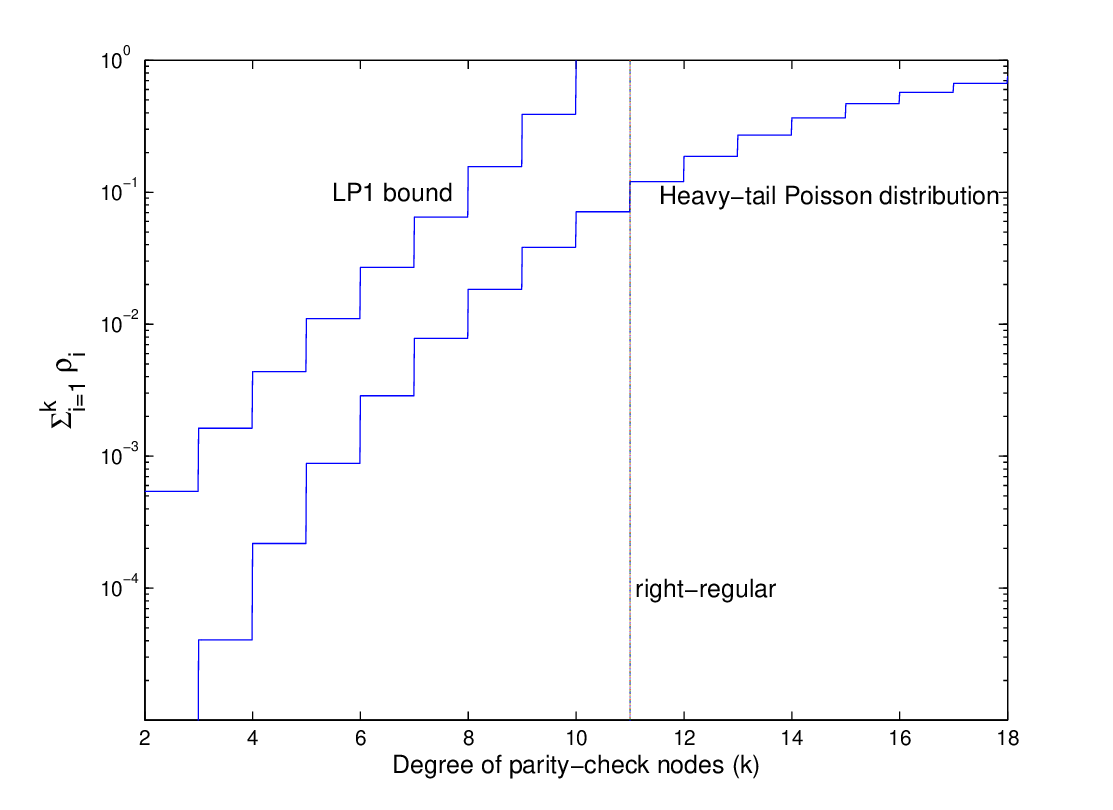,scale=0.60}\\
\epsfig{file=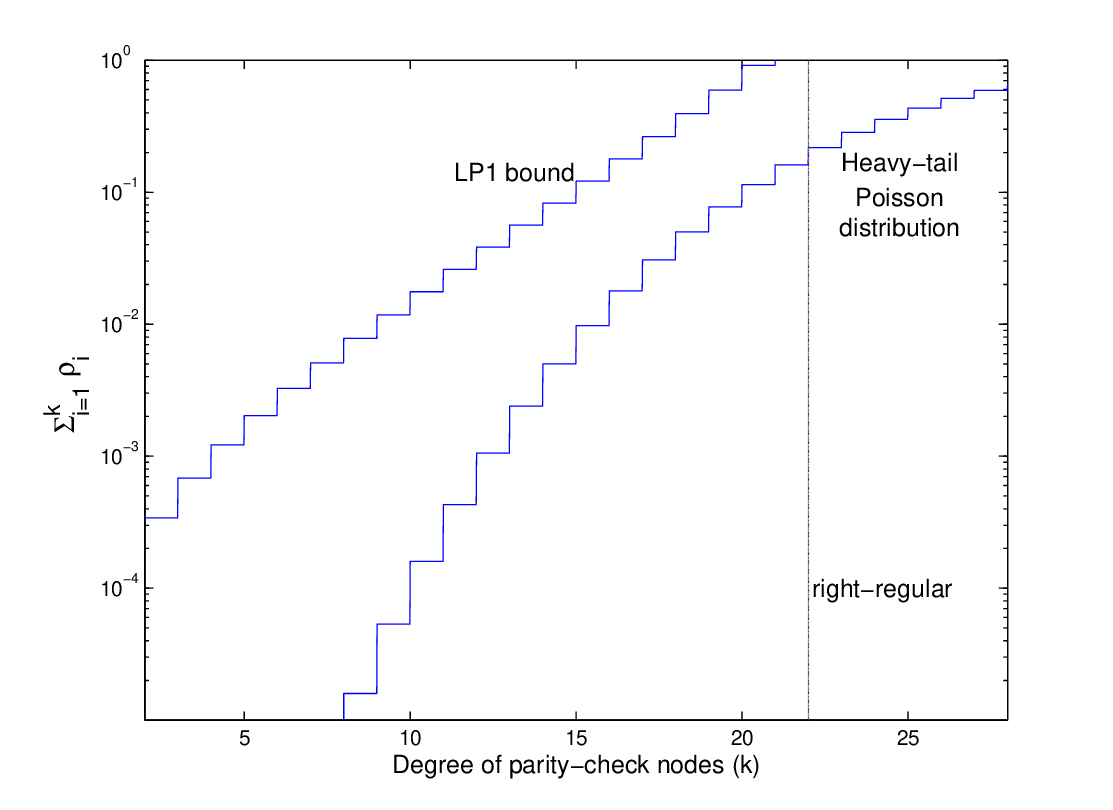,scale=0.60}
\end{center}
\caption{\label{Figure: comparison between the LP1 bound and the
heavy-tail Poisson distribution} A comparison between the LP1 bound,
the heavy-tail Poisson degree distribution in \eqref{eq: heavy-tail
Poisson distribution} and \eqref{eq: choice of parameters for the
heavy-tail Poisson distribution}, and the parity-check degree
distribution of the right-regular LDPC ensemble (it is calculated
via \eqref{eq: N for right-regular ensembles}, \eqref{eq: k2},
\eqref{eq: alpha as function of p and N} where the right degree is
equal to $a_{\text{R}} = \lceil \frac{1}{\alpha}\rceil+1$). This
comparison refers to a BEC whose capacity is one-half (upper plot)
and three-quarters (lower plot) bits per channel use, and the
setting where 99.9\% of the channel capacity is achieved under BP
decoding with vanishing bit erasure probability. The stair functions
correspond to the fraction of edges which are attached to
parity-check nodes whose degrees are at most $k$ for a positive
integer $k$.}
\end{figure*}
Comparisons between the heavy-tail Poisson distribution and the LP1
bound are shown in Figure~\ref{Figure: comparison between the LP1
bound and the heavy-tail Poisson distribution}. We note that the LP1
bound is an upper bound on the parity-check degree distribution
which is valid under ML decoding (and hence, it is general for any
decoding algorithm), whereas the heavy-tail Poisson distribution is
designed to achieve a certain gap to capacity under BP decoding. We
also show in this figure the fixed degree of the parity-check nodes
for the right-regular LDPC code ensemble; this calculation is done
via \eqref{eq: N for right-regular ensembles}, \eqref{eq: k2},
\eqref{eq: alpha as function of p and N} where the right degree is
equal to $a_{\text{R}} = \lceil \frac{1}{\alpha}\rceil+1$. Although
the latter case corresponds to a step function, the degree where
this function switches from zero to one provides an indication to
the reasonable tightness of the LP1 upper bound with respect to the
value of the parity-check degree $k$ where this upper bound is close
to~1.

The following analysis compares between the behavior of the upper
bound on $\rho_i$ as given in Corollary~\ref{Corollary: Behavior
of the left and right degrees when P_b=0} with the behavior of the
heavy-tail Poisson distribution in the limit where the gap to
capacity vanishes under BP decoding: Note that the truncated
harmonic sum $H(k)$ scales like the logarithm of $k$ (more
precisely, $H(k) \approx \ln(k) + \gamma$ for $k \gg 1$ where
$\gamma \approx 0.5772$ is Euler's constant), and the value of $N$
as given in \eqref{eq: choice of parameters for the heavy-tail
Poisson distribution} becomes un-bounded as the gap to capacity
vanishes (since it is inversely proportional to $\varepsilon$).
Hence, for small values of the gap to capacity (i.e., when
$\varepsilon \ll 1$), we get from \eqref{eq: choice of parameters
for the heavy-tail Poisson distribution}
\begin{equation*}
\alpha \approx \frac{\ln \frac{1-r}{\varepsilon r}}{1-r}, \quad N
\approx \frac{1-r}{\varepsilon r} + 1
\end{equation*}
and therefore \eqref{eq: heavy-tail Poisson distribution} yields
that the fraction of edges which are attached to parity-check
nodes of a given degree~$i$ scales like
$\varepsilon^{\frac{1}{1-r}} \left(\ln \frac{1}{\varepsilon}
\right)^{i-1}$ for $i \geq 1$. The upper bound on $\rho_i$ as
given in Corollary~\ref{Corollary: Behavior of the left and right
degrees when P_b=0} scales like $\frac{\varepsilon}{\ln
\frac{1}{\varepsilon}}$, where this bound is even valid under ML
decoding. For a comparison between this general upper bound and
the behavior of the Poisson distribution when the gap to capacity
vanishes, we note that for any rate $r<1$, a positive integer~$i$
and $\varepsilon \ll 1$, the inequality $
\varepsilon^{\frac{1}{1-r}} \left(\ln \frac{1}{\varepsilon}
\right)^{i-1} \ll \frac{\varepsilon}{\ln \frac{1}{\varepsilon}} $
holds, as expected from a comparison of a degree distribution with
a general upper bound. Moreover, it follows from the asymptotic
analysis that for small design rates (i.e., $r \ll 1$), the
Poisson distribution gets closer to the LP1 bound in the limit
where $\varepsilon \rightarrow 0$ (as exemplified in
Fig.~\ref{Figure: comparison between the LP1 bound and the
heavy-tail Poisson distribution} by comparing the upper and lower
plots which correspond to a capacity of $\frac{1}{2}$ and
$\frac{3}{4}$ bits per channel use, respectively).
\end{example}
{\bf{Analytical solutions for the LP3 and LP4 bounds}}: Consider
an LP problem of the form
\begin{equation}
\mbox{\fbox{$
\begin{array}{l}
\text{maximize} \; \; \sum\limits_{i=2}^k \lambda_i, \quad k=2, 3, \ldots \nonumber \\
\text{subject to} \nonumber \\
\begin{cases}
\sum\limits_{i=2}^{\infty} \frac{\lambda_i}{i} \leq \alpha
\hspace{-0.2cm} \\[0.3cm]
\sum\limits_{i=2}^{\infty} \lambda_i = 1 \\[0.3cm] \lambda_i
\geq 0, \quad i=2, 3, \ldots
\end{cases}
\end{array}
$}}
\end{equation}
If $k \alpha \leq 1$ then the optimal solution is obtained by
setting $\lambda_k = k\alpha$, $\lambda_j = 1-k\alpha$ for some $j
\rightarrow \infty$ where all the other $\lambda_i$'s are set to
zero. This gives a solution which is equal to $\sum_{i=1}^k
\lambda_i = \lambda_k = k\alpha$. If $k \alpha > 1$ then the
optimal solution is obtained by setting $\lambda_k = 1$ and all
the other $\lambda_i$'s to be zero. Hence, the solution of this LP
problem is given by $ \min\{k \alpha, 1\} $ which implies that the
closed-form solutions of the LP3 and LP4 bounds are given by
\begin{equation}
\min\left\{1, \frac{k \ln\bigl(\frac{1}{g_1}\bigr)}{2(1-C)
\left(1+\frac{\varepsilon C}{1-C}\right) \ln\biggl(\frac{1}{1-2
h_2^{-1}\bigl(\frac{1-C-h_2(P_{\text{b}})}
{1-(1-\varepsilon)C}\bigr)}\biggr)}  \right\} \label{eq: closed
form solution of LP3}
\end{equation}
and
\begin{equation}
\min\left\{1, \frac{k \ln\bigl(\frac{1}{C}\bigr)}{2(1-C)
\left(1+\frac{\varepsilon C}{1-C}\right) \ln\biggl(\frac{1}{1-2
h_2^{-1}\bigl(\frac{1-C-h_2(P_{\text{b}})}
{1-(1-\varepsilon)C}\bigr)}\biggr)}  \right\} \label{eq: closed
form solution of LP4}
\end{equation}
respectively.

Based on the observations in Theorems~\ref{Theorem: Behavior of
left and right degrees} and~\ref{Theorem: degree-2 variable
nodes}, the fraction of edges connected to variable nodes of small
degree is expected to be significantly larger than the fraction of
edges which are connected to parity-check nodes of the same
degree. This is shown in the following example:
\begin{example}[{\bf{LP3 bound}}]
Consider LDPC code ensembles whose design rate is one-half bit per
channel use, and whose transmission takes place over a BIAWGN
channel. Lets assume that we wish to find upper bounds on the
fraction of edges up to degree $k$ (for a parameter $k \geq 2$) for
the setting of a bit error probability of (at most) $P_{\text{b}} =
10^{-10}$ under ML decoding (or any sub-optimal decoding algorithm)
at $\frac{E_{\text{b}}}{N_0} = 0.188~\text{dB}$. This implies a gap
to capacity which is equal to $\varepsilon = 1.42 \cdot 10^{-4}$.
From \eqref{eq: closed form solution of LP3}, we obtain the
following inequalities (also verified numerically via
\cite{CVX_software}):
\begin{eqnarray*}
&& \lambda_2 \leq 0.2683\\
&& \lambda_2 + \lambda_3 \leq 0.4025\\
&& \lambda_2 + \lambda_3 + \lambda_4 \leq 0.5367\\
&& \lambda_2 + \lambda_3 + \lambda_4 + \lambda_5 \leq 0.6709\\
&& \lambda_2 + \lambda_3 + \lambda_4 + \lambda_5 + \lambda_6
\leq 0.8051 \\
&& \lambda_2 + \lambda_3 + \lambda_4 + \lambda_5 + \lambda_6 +
\lambda_7 \leq 0.9392 \\
&& \lambda_2 + \lambda_3 + \lambda_4 + \lambda_5 + \lambda_6 +
\lambda_7 + \lambda_8 \leq 1.0000.
\end{eqnarray*}
A comparison of these numerical results with those presented in
Fig.~\ref{Figure: LP1 vs LP2.eps} for the same value of
$\frac{E_{\text{b}}}{N_0}$ shows a big difference between the two
upper bounds on the sequences $\{\lambda_i\}$ and $\{\rho_i\}$.
This difference is well expected in light of the bounds in
Corollary~\ref{Corollary: Behavior of the left and right degrees
when P_b=0} where for every finite degree~$i$, the upper bounds on
$\lambda_i$ and $\rho_i$ scale like $\frac{1}{\log
\frac{1}{\varepsilon}}$ and $\frac{\varepsilon}{\log
\frac{1}{\varepsilon}}$, respectively. We note that this
difference is not an artifact of the bounding technique, as is
demonstrated in Proposition~\ref{Proposition: the tightness of the
upper bound on lambda_2} for the BEC.
\end{example}

\subsection{Bounds on the Graphical Complexity of Finite-Length Codes}
\label{Numerical Results for Finite-Length Analysis} In various
applications, there is a need to design a communication system which
fulfills several requirements on the available bandwidth with
acceptable delay for transmitting and processing the data while
maintaining a certain fidelity criterion in reconstructing the data.
In this setting, one wishes to design a code which satisfies the
delay constraint (i.e., the block length is limited) while adhering
to the required performance over the given channel. By fixing the
communication channel model and code rate (which is related to the
bandwidth expansion caused by the error-correcting code),
sphere-packing bounds are transformed into lower bounds on the
minimal block length required to achieve a target block error
probability at a certain gap to capacity using an arbitrary block
code and decoding algorithm. This issue is studied in
\cite[Section~V]{ISP08}.

\begin{figure*}[here!]
\begin{center}
\epsfig{file=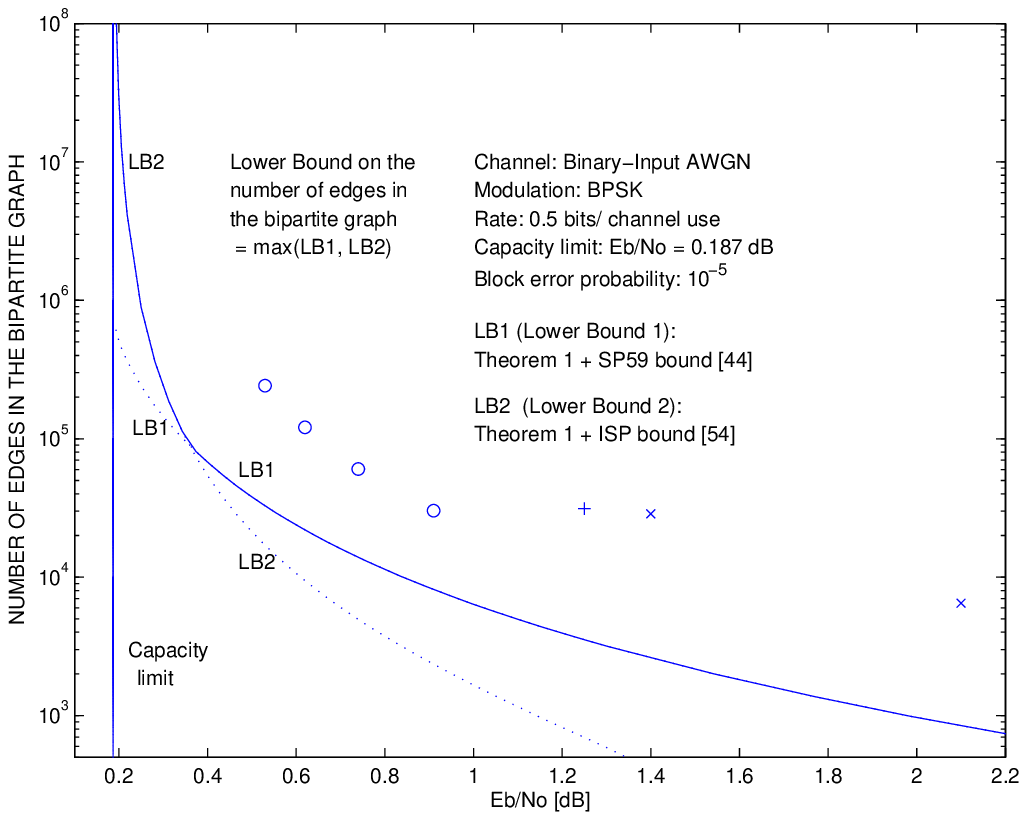,scale=0.68}
\end{center}
\caption{\label{Figure: graphical_complexity_Rate05.eps} A
comparison between the graphical complexity of various efficient
LDPC code ensembles and an information-theoretic lower bound. The
graphical complexity is measured by the number of edges which are
used to represent the codes (or code ensembles) by bipartite graphs
in order to achieve a fixed target block error probability over a
given communication channel. It is assumed that the code is BPSK
modulated and transmitted over a binary-input AWGN channel. This
figure refers to a target block error probability of $P_{\text{B}} =
10^{-5}$, and a design rate of one-half bit per channel use. The
information-theoretic lower bound is valid under maximum-likelihood
(ML) decoding (and, hence, it also holds under any sub-optimal
decoding algorithm). For the comparison of the lower bound with
various LDPC code ensembles, we refer to both ML and
belief-propagation (BP) decoding algorithms. The circled points
refer to ML decoding, and they are based on the tangential-sphere
upper bound which is applied to the (6,12) regular LDPC code
ensembles of Gallager for block lengths of 5040, 10080, 20160 and
40320 bits (these points rely on \cite[Table~II]{Tong_ICC08}). The
other three points in this figure refer to LDPC code ensembles which
are decoded by a BP decoder. The point marked by $`+'$ refers to a
non-punctured protograph LDPC code ensemble of block length~7360
bits and of rate one-half (see \cite[Fig.~9]{Divsalar_ISIT2006}).
The other two points which are marked by $`\times'$ refer to
irregular quasi-cyclic LDPC code ensembles (see \cite[Figs.~10 and
11]{QC LDPC_IT07}). The two information-theoretic lower bounds on
the graphical complexity ('LB1' and 'LB2') rely, respectively, on
the sphere-packing bound of Shannon \cite{Shannon_1959} and the
recently introduced sphere-packing bound in \cite{ISP08}. Both of
these bounds also rely on Theorem~\ref{Theorem: Lower bound on right
degree} which serves as a lower bound on the average right degree.
The information-theoretic lower bound that is shown in this figure
is obtained by taking the maximum of the LB1 and LB2 bounds.}
\end{figure*}

In the following, we refer to the graphical complexity of an
arbitrary bipartite graph which represents a binary linear block
code. The graphical complexity has an operational meaning for an
iterative message-passing decoder since the number of edges is
equal to the number of right-to-left and left-to-right messages
which are delivered in each iteration. As opposed to
\cite{Hsu_Achilleas1}, \cite{Pfister1} and \cite{Pfister2}, we
refer here to the graphical complexity of {\em finite-length
codes}. In order to evaluate an information-theoretic lower bound
on the graphical complexity which is expressed in terms of the
target block error probability and the corresponding achievable
gap to capacity, we rely here on the following algorithm:
\begin{itemize}
\item {\em Step~1}: Sphere-packing bounds are used to calculate a
lower bound on the minimal required block length in terms of the
achievable rate with a target block error probability and its gap
to capacity. For a memoryless symmetric channel, the lower bound
on the minimal block length is calculated via the ISP bound (for
finite-length codes, this recent sphere-packing bound suggests a
significant improvement over the bounds in \cite{SGB} and
\cite{Valembois_Fossorier}, see Section~\ref{sub-section:
sphere-packing bounds} and \cite[Section~III]{ISP08}). In
addition, this lower bound is also compared with the 1959
sphere-packing (SP59) bound of Shannon (see
Section~\ref{sub-section: sphere-packing bounds} and
\cite{Shannon_1959}) for a binary-input AWGN channel where the
transmitted signals are assumed to have equal energy.
\item {\em Step~2}: A lower bound on the average right degree is
calculated via Theorem~\ref{Theorem: Lower bound on right degree}
for an arbitrary bipartite graph which is used to represent a
binary linear block code. Note that for an LDPC code whose
parity-check matrix is not necessarily full-rank, one can apply
this lower bound by replacing the code rate with the design rate
(see Discussion~\ref{discussion: extension of Theorem 1 to LDPC
ensembles} in Section~\ref{Section: Proofs of Main Results}). The
calculation of this lower bound for a target block error
probability $P_{\text{B}}$ also stays valid if the block length
$n$ is replaced in \eqref{delta} with a lower bound $n'$ (as
calculated in the previous step).
\item {\em Step~3:} The total number of edges of a bipartite graph is a
measure of its graphical complexity. For a bipartite graph which
refers to a design rate of $R_{\text{d}}$, the total number of
edges is equal to $| \mathcal{E} | = (1-R_{\text{d}})n
a_{\text{R}}$. Replacing $n$ and $a_{\text{R}}$ by the lower
bounds calculated in Steps~1 and~2, respectively, gives a lower
bound on the number of edges.
\end{itemize}

The resulting lower bound on the total number of edges is general
for every representation of a binary linear block code by a
parity-check matrix and its respective bipartite graph. This bound
depends on the code rate (or design rate), the communication
channel, the achievable gap to capacity, and the target block
error probability. This lower bound holds for an arbitrary
representation of the code by a bipartite graph.

According to the above description of the three steps used to
calculate the information-theoretic lower bound on the graphical
complexity, we calculate here two lower bounds on the graphical
complexity:
\begin{itemize}
\item {\bf{LB1}}: A lower bound which combines a
lower bound on the block length calculated via the SP59 bound
\cite{Shannon_1959}, and a lower bound on the average right degree
which is calculated via Theorem~\ref{Theorem: Lower bound on right
degree} for a target block error probability $P_{\text{B}}$ and a
given code rate (or design rate).
\item {\bf{LB2}}: A lower bound which combines a
lower bound on the block length calculated via the ISP bound
\cite[Section~III]{ISP08}, and the same lower bound on the average
right degree.
\end{itemize}
We note that Steps~2 and~3 in the above algorithm are common for
the calculation of the LB1 and LB2~bounds, and the only difference
in the calculation of these two bounds is in Step~1 where the SP59
and ISP bounds are used for the LB1 and LB2 bounds, respectively.
The resulting lower bound (LB) on the graphical complexity is the
maximal value of the LB1 and LB2 bounds, i.e., $\text{LB} =
\max(\text{LB1}, \text{LB2})$. We note that the resulting lower
bound on the graphical complexity holds under ML decoding or any
sub-optimal decoding algorithm.

The above algorithm is applied in Figure~\ref{Figure:
graphical_complexity_Rate05.eps} to obtain a lower bound on the
graphical complexity of an arbitrary binary linear block code of
rate one-half and with a target block error probability of
$P_{\text{B}} = 10^{-5}$. It is assumed that the code is BPSK
modulated, and the transmission takes place over a binary-input
AWGN channel. The un-bounded complexity in the limit where the gap
to capacity vanishes is due to the infinite block length which is
required to obtain reliable communications at rates which are
arbitrarily close to capacity. We note that the {\em bounded}
graphical complexity for the BEC, as demonstrated in
\cite{Hsu_Achilleas1, Pfister1} and \cite{Pfister2}, is obtained
by addressing the graphical complexity per information bit, and by
also allowing more complicated Tanner graphs which include state
nodes (e.g., punctured bits) in addition to the variable and
parity-check nodes which are used for a representation of these
codes by bipartite graphs.

As shown in Figure~\ref{Figure: graphical_complexity_Rate05.eps},
the bound LB2 is advantageous over LB1 for low values of
$\frac{E_{\text{b}}}{N_0}$ which are close to the capacity limit;
this phenomenon is even more pronounced for higher code rates
(above one-half bit per channel use). This observation is
partially due to the fact that the ISP bound depends on the
particular type of modulation used, in contrast to the SP59 bound
which only assumes that the modulated signals have equal energy
but does not consider the particular modulation used.

The lower bound on the graphical complexity is compared here with
some efficient LDPC codes (or code ensembles) as reported in the
literature. To this end, we refer to computer simulations under BP
decoding, and also to upper bounds on the block error probability
under ML decoding. Although the number of edges is relevant for
the decoding complexity per iteration under BP decoding, some
comparisons with ML decoding provide a better assessment of the
tightness of this information-theoretic lower bound. The circled
points in Figure~\ref{Figure: graphical_complexity_Rate05.eps} are
based on the tangential-sphere upper bound\footnote{For a
presentation of the tangential-sphere bound, originally introduced
by Poltyrev \cite{Poltyrev}, we refer the reader to
\cite[pp.~23--32]{Tutorial}.} which is applied to the (6,12)
regular LDPC code ensembles of Gallager for block lengths of 5040,
10080, 20160 and 40320 bits whose block error probability is upper
bounded by $10^{-5}$ (see \cite[Table~II]{Tong_ICC08}). The other
three points which are shown in Figure~\ref{Figure:
graphical_complexity_Rate05.eps} refer to LDPC code ensembles
which are decoded by a BP decoder. The point marked by $`+'$
refers to a non-punctured protograph LDPC code ensemble of block
length of 7360 bits and a design rate of one-half (see
\cite[Fig.~9]{Divsalar_ISIT2006}). The other two points which are
marked by $`\times'$ refer to irregular quasi-cyclic LDPC code
ensembles (see \cite[Figs.~10 and~11]{QC LDPC_IT07}) where the
graphical complexity is obtained via the degree distributions
which are given in \cite[Examples~10 and 11]{QC LDPC_IT07}. To
conclude, the information-theoretic lower bound on the graphical
complexity becomes un-bounded as the gap to capacity vanishes
(even under ML decoding). It also behaves in a similar way to the
circled points in Figure~\ref{Figure:
graphical_complexity_Rate05.eps} (where these points refer to the
performance of a regular LDPC code ensembles under ML decoding).
Moreover, the comparison of this lower bound in
Figure~\ref{Figure: graphical_complexity_Rate05.eps} with some
efficient LDPC code ensembles under BP decoding (where the
corresponding points are marked by $`+'$ and $`\times'$) indicate
the gain that can be potentially obtained by improved designs of
efficient LDPC codes and iterative decoding algorithms defined on
graphs.

\section{Outlook}
\label{Section: Outlook}

This work considers some universal properties of
capacity-approaching low-density parity-check (LDPC) code ensembles
whose transmission takes place over memoryless binary-input
output-symmetric (MBIOS) channels. Properties of the degree
distributions, graphical complexity and the fundamental cycles of
the bipartite graphs are studied in this paper via the derivation of
information-theoretic bounds (see Sections~\ref{Section: main
results} and~\ref{Section: Proofs of Main Results}). The
applications of these bounds are exemplified in
Section~\ref{Section: Numerical Results}.

In the following, we gather some interesting open problems which are
related to this research work:
\begin{itemize}
\item The analysis in this paper relies (in part) on the lower bound
\eqref{eq: Lower bound on conditional entropy} on the conditional
entropy (see \cite{Wiechman_Sason}). Note that this bound depends on
the right degree distribution (i.e., the degree distribution of the
parity-check nodes), but the dependence on the left degree
distribution is rather weak (according to
Section~\ref{subsubsection: An adaptation of the analysis to LDPC
codes}, this dependence is made only through the design rate of the
LDPC code ensemble). It would be interesting to improve this bound
by also having an explicit dependence on the left degree
distribution. This goal can be obtained by improving the weak link
in the derivation of this bound, namely, by tightening the upper
bound \eqref{second step} on the conditional entropy of the syndrome
vector (which is expressed by the sum of the respective conditional
entropies of the components of the syndrome). Note that for the BSC,
the bound in \eqref{eq: Lower bound on conditional entropy}
coincides with the bound of Gallager in
\cite[Section~3.8]{Gallager_1962} (since the conditioning on the RHS
of \eqref{second step} becomes irrelevant for the BSC, due to the
fact that the absolute value of the LLR is a constant for this
channel). A step towards the improvement of Gallager's bound for the
BSC was done by Wadayama \cite{Wadayama_ISIT06} where the entropy of
the syndrome vector was calculated exactly in terms of the coset
weight distribution of the code (or the average coset weight
distribution of the ensemble). For a general MBIOS channel, the
improvement of the bound in \eqref{eq: Lower bound on conditional
entropy} is an open problem, and it may provide an explicit
dependence of the bound on the pair of degree distributions for a
code which is represented by a bipartite graph.
\item Unlike the information-theoretic bound in \eqref{eq: Lower
bound on conditional entropy}, the bounds presented in
\cite{Montanari_IT05} rely on statistical physics, and therefore do
not provide a bound on the conditional entropy which is valid for
every binary linear block code from the considered ensembles.
%(for a survey paper which introduces codes defined on graphs and
%highlights connections with statistical physics, the reader is
%referred to \cite{Montanari_Urbanke07}).
It would be interesting to get some theory that unifies the
information-theoretic and statistical physics approaches, and
provides bounds that are tight on the average and valid for each
code. We note that the bounds in \cite{Montanari_IT05} depend on
both the left and right degree distributions for LDPC code ensembles
(though their computation is more complicated than the bound given
in \eqref{eq: Lower bound on conditional entropy}).
\item The asymptotic bounds in Corollary~\ref{corollary: lower bound on the cardinality of the
fundamental system of cycles of  LDPC ensembles} address the average
cardinality of the fundamental system of cycles for bipartite graphs
representing LDPC code ensembles where the results are directly
linked to the average right degree of these ensembles. Further study
of the possible link between the statistical properties of the
degree distributions of capacity-approaching LDPC code ensembles and
some other graphical properties related to the bipartite graphs of
these ensembles is of interest.
\item The graphical complexity of capacity-approaching LDPC codes
is studied in this paper via an information-theoretic lower bound
which relies on both Theorem~\ref{Theorem: Lower bound on right
degree} and sphere-packing bounds (see Section~\ref{Numerical
Results for Finite-Length Analysis}). The graphical complexity is
defined to be the number of edges in the bipartite graphs used to
represent these codes. A recent sphere-packing bound which was
introduced in \cite{ISP08} is shown to be helpful for the
calculation of the lower bound on the graphical complexity,
especially when the gap to capacity becomes small (see the algorithm
for the calculation of this bound in Section~\ref{Numerical Results
for Finite-Length Analysis} and the results shown in
Figure~\ref{Figure: graphical_complexity_Rate05.eps}). Further
tightening of sphere-packing bounds for finite-length codes,
especially for codes of short to moderate block lengths, is of
interest and it has the potential of further improving the resulting
lower bound on the graphical complexity. An improvement of the
sphere-packing bounds introduced in \cite{Shannon_1959} and
\cite{ISP08} will also contribute to the study of the sub-optimality
of iteratively decoded codes for finite block lengths.
\item The derivation of universal bounds on the number of iterations
of code ensembles defined on graphs, measured in terms of the
achievable gap (in rate) to capacity, is of theoretical and
practical interest. In a recent work \cite{iterations}, this issue
is addressed for the BEC. It is demonstrated in \cite{iterations}
that the number of iterations which is required for successful
message-passing decoding scales at least like the inverse of the
achievable gap (in rate) to capacity, provided that the fraction of
degree-2 variable nodes of these turbo-like code ensembles does not
vanish (hence, the number of iterations becomes unbounded as the gap
to capacity vanishes). Note that Lemma~\ref{lemma: L_2 for capacity
approaching LDPC ensembles} (see p.~\pageref{lemma: L_2 for capacity
approaching LDPC ensembles}) provides a condition which ensures that
the fraction of degree-2 variable nodes stays strictly positive for
capacity-achieving LDPC code ensembles. A generalization of such a
lower bound on the number of iterations for an arbitrary MBIOS
channel is of interest. The matching condition for generalized
extrinsic information transfer (GEXIT) curves serves to conjecture
in \cite[Section~XI]{GEXIT} that, also for an arbitrary MBIOS
channel, this number of iterations scales like the inverse of the
achievable gap to capacity.
\item Extension of the results in this paper to channels with
memory (e.g., finite-state channels) is of interest. In this
respect, the reader is referred to \cite{Grover07} which considers
information-theoretic bounds on the achievable rates of LDPC code
ensembles for a class of finite-state channels.
\item Extension of the results in this work to general
ensembles of multi-edge type LDPC codes (see
\cite[Chapter~7]{RiU_book}) is of interest.
\end{itemize}

\appendices

\section{Proof of Lemma~\ref{lemma: extension of lower bound on the conditional
entropy}}
\renewcommand{\thelemma}{I.\arabic{lemma}}
\setcounter{lemma}{0} \label{Appendix: extension of lower bound on
the conditional entropy}

The following proof deviates from the analysis in
Section~\ref{subsubsection: The analysis for a full-rank
parity-check matrix}, starting from \eqref{second step}.
\begin{itemize}
\item In the transition to the last line in \eqref{second step},
the conditional entropy $H\bigl({\bf{S}} \, | \, \Omega_1, \ldots,
\Omega_n \bigr)$ is upper bounded by the sum of the conditional
entropies of the $n(1-R)$ independent components of the syndrome
${\bf{S}}$ under the assumption that the parity-check matrix is
full-rank. In the general case where this parity-check matrix is
not necessarily full-rank, the rate $R$ of the code may exceed the
design rate $R_{\text{d}}$ due to a possible linear dependence of
the rows in this matrix. Therefore, we obtain an upper bound on
the conditional entropy by summing over the $n(1-R_{\text{d}})$
components of the syndrome.
\item In parallel to \eqref{second step}, we get the inequality
\begin{eqnarray}
&& H\bigl( \Phi_1, \ldots, \Phi_n \, | \,
\Omega_1, \ldots, \Omega_n \bigr) \nonumber\\[0.1cm]
&& \leq H(M) + \sum_{j=1}^{n(1-R_{\text{d}})} H\bigl(S_j \, | \,
\Omega_1, \ldots, \Omega_n \bigr)\,. \label{modified second step}
\end{eqnarray}
\item The entropy of the transmitted codeword ${\bf{X}}$ is equal to the
entropy of the index $M$ of the received vector in the appropriate
coset, regardless of the rank of $H$. Hence, $H({\bf{X}})$ in the
second line of \eqref{eq: chain of equalities for the conditional
entropy} can be replaced by $H(M)$, and we get
\begin{equation*}
H(\mathbf{X}|\mathbf{Y}) = H(M) + n[H(\widetilde{Y}_1) -
I(X_1;\widetilde{Y}_1)] - H(\widetilde{\mathbf{Y}}).
\end{equation*}
\item Combining \eqref{eq:
trivial upper bound on the mutual information}--\eqref{first
step}, \eqref{modified second step} and the last equality, we get
the inequality (note that the entropy $H(M)$ cancels out)
\begin{equation*}
H(\mathbf{X}|\mathbf{Y}) \geq n(1-C) -
\sum_{j=1}^{n(1-R_{\text{d}})}H(S_j \big{|} \Omega_1, \ldots,
\Omega_n)
\end{equation*}
which is similar to \eqref{eq: new inequality for the conditional
entropy} except that the sum on the RHS is over the
$n(1-R_{\text{d}})$ (possibly linearly-dependent) components of
the syndrome.
\end{itemize}
From this point, the analysis is similar to
Section~\ref{subsubsection: The analysis for a full-rank
parity-check matrix} which then yields an extension of \eqref{eq:
Lower bound on conditional entropy} with $R$ replaced by
$R_{\text{d}}$ when the parity-check matrix is not necessarily
full-rank.

\section{Proof of Lemma~\ref{lemma: Extreme values of g_1}}
\label{Appendix: Proof of the lemma on the extreme values of g_1}

This lemma is proved by expressing the channel capacity as a
non-negative infinite series which depends on the sequence
$\{g_k\}_{k \geq 1}$, and solving an optimization problem for the
extreme values of $g_1$ subject to a constraint on the channel
capacity $C$. To this end, we rely on the equality in \eqref{eq:
channel capacity 2 of an MBIOS channel} for the capacity of an
MBIOS channel:
\begin{eqnarray}
C &=& \int_0^{\infty} a(l) (1+e^{-l})
\left[1-h_2\left(\frac{1}{1+e^{\,l}}\right)\right] \; \text{d}l
\nonumber\\
&\stackrel{\mathrm{(a)}}{=}& \int_0^{\infty} a(l) (1+e^{-l}) \;
\frac{1}{2\ln2}\sum_{k=1}^{\infty}
\frac{\tanh^{2k}\left(\frac{l}{2} \right)}{k(2k-1)} \; \text{d}l
\nonumber\\
&=& \frac{1}{2\ln2}\sum_{k=1}^{\infty} \Biggl\{
\frac{\int_0^{\infty}
a(l) (1+e^{-l}) \tanh^{2k}\left(\frac{l}{2} \right) \; \text{d}l}{k(2k-1)} \Biggr\} \nonumber\\
&\stackrel{\mathrm{(b)}}{=}& \frac{1}{2 \ln 2} \sum_{k=1}^{\infty}
\frac{g_k}{k(2k-1)} \label{eq: capacity of MBIOS channel}
\end{eqnarray}
where equality~(a) follows by substituting $x =
\frac{1}{1+e^{\,l}}$ in \eqref{eq: power series for h_2}, and
equality~(b) follows from \eqref{eq: definition of g_k}; this
provides an expression for the channel capacity in terms of the
non-negative sequence $\{g_k\}_{k=0}^{\infty}$ defined in
\eqref{eq: definition of g_k}. The representation of the capacity
as the infinite series in \eqref{eq: capacity of MBIOS channel}
follows in fact from the result which is obtained via
\cite[Propositions~3.1--3.3]{Sharon_COM06} by referring to an
equi-probable binary input, though the derivation here is more
direct.

We start with the proof of the upper bound on $g_1$, as given on
the RHS of \eqref{eq: extreme values of g_1}. Since we look for
the maximal value of $g_1$ among all MBIOS channels with a given
capacity $C$, then we need to solve the optimization problem
\begin{eqnarray}
&& \text{maximize} \quad g_1 \nonumber \\
&& \text{subject to} \quad \frac{1}{2 \ln 2} \;
\sum_{k=1}^{\infty} \frac{g_k}{k(2k-1)} = C. \label{eq:
maximization problem}
\end{eqnarray}
Based on Lemma~\ref{lemma: g_k and g_1}, for every MBIOS channel,
$g_k \geq (g_1)^k$ for all $k \in \naturals$. Therefore
\begin{eqnarray}
&&\frac{1}{2 \ln 2} \sum_{k=1}^{\infty} \frac{g_k}{k(2k-1)}
\nonumber\\
&& \geq \frac{1}{2 \ln 2} \sum_{k=1}^{\infty}
\frac{(g_1)^k}{k(2k-1)}
\nonumber\\
&&=1- h_2\left(\frac{1-\sqrt{g_1}}{2} \right)  \label{eq: equality
with epsilon_k}
\end{eqnarray}
where the last equality is based on \eqref{eq: power series for
1-h_2}. The equality constraint in \eqref{eq: maximization
problem} and the inequality \eqref{eq: equality with epsilon_k}
yield that
\begin{equation*}
1- h_2\left(\frac{1-\sqrt{g_1}}{2} \right) \leq C
\end{equation*}
from which the RHS of \eqref{eq: extreme values of g_1} follows.
Note that this upper bound on $g_1$ is attained when $g_k = (g_1)^k$
for all $k \in \naturals$. To show this equality, note that for a
BSC with crossover probability $p$, the LLR at the channel output
$(L)$ is bimodal and it gets the values $l_1 =
+\ln\bigl(\frac{1-p}{p}\bigr)$ and $l_2 = -l_1$ with probabilities
$1-p$ and $p$, respectively. Eq.~\eqref{eq: alternative definition
of g_k} then gives
\begin{eqnarray}
&& g_k \triangleq \expectation
\left[\tanh^{2k}\left(\frac{L}{2}\right) \right] \nonumber \\
&& \hspace*{0.4cm} = (1-p) \tanh^{2k}\left(\frac{{l_1}}{2} \right)
+ p \tanh^{2k}\left(-\frac{l_1}{2} \right) \nonumber \\
&& \hspace*{0.4cm} = \tanh^{2k}\left(\frac{l_1}{2} \right) \nonumber \\
&& \hspace*{0.4cm} = \left(\frac{e^{l_1}-1}{e^{l_1}+1} \right)^{2k} \nonumber\\
&& \hspace*{0.4cm} = (1-2p)^{2k} \; , \quad \forall \; k \in
\naturals. \label{eq: g_k for BSC}
\end{eqnarray}
Hence for the BSC, $g_k = (g_1)^k$ for all $k \in \naturals$. The
upper bound on $g_1$ on the RHS of \eqref{eq: extreme values of
g_1} is therefore achieved for a BSC whose crossover probability
is $p = h_2^{-1}(1-C)$.

The proof of the lower bound on $g_1$ relies on \eqref{eq:
capacity of MBIOS channel}. Since the sequence $\{g_k\}_{k \geq
1}$ is monotonically non-increasing and non-negative (this
property follows directly from \eqref{eq: alternative definition
of g_k}), then
\begin{eqnarray*}
C &=& \frac{1}{2\ln2} \sum_{k=1}^{\infty} \frac{g_k}{k(2k-1)} \nonumber\\
&\leq& \frac{g_1}{2\ln2} \sum_{k=1}^{\infty} \frac{1}{k(2k-1)} \nonumber\\
&=& g_1
\end{eqnarray*}
where the last equality follows from \eqref{eq: infinite series}.
This lower bound on $g_1$ is attained for a BEC (since for a BEC
whose erasure probability is $p$, \eqref{eq: alternative definition
of g_k} implies that the sequence $\{g_k\}$ is constant and $g_1 =
1-p = C$).

\section{Proof of Lemma~\ref{lemma: L_2 for capacity approaching
LDPC ensembles}} \label{Appendix: Proof of lemma on L_2 for c.a.
LDPC ensembles} From the assumption in Lemma~\ref{lemma: L_2 for
capacity approaching LDPC ensembles}, the satisfiability of the
flatness condition for this capacity-achieving sequence gives that
\begin{equation}
\lim_{m\rightarrow\infty} \mathcal{B}(a) \, \lambda_2^{(m)} \,
\rho'_{m}(1) = 1\,. \label{eq: flatness condition}
\end{equation}
From \eqref{eq: simple expression for Lambda_2}, the fraction of
degree-2 variable nodes is given by
\begin{equation}
\Lambda^{(m)}_2 =
\frac{\lambda^{(m)}_2}{2\,\int_0^1\lambda_m(x)\mathrm{d}x}\,,\quad\quad\forall
m\in\naturals  \label{eq: connection for degree-2 variable nodes}
\end{equation}
and therefore
\begin{eqnarray}
&& \hspace*{-0.5cm} \lim_{m\rightarrow\infty} \Lambda^{(m)}_2 \nonumber\\
&& \hspace*{-0.5cm} \stackrel{\text{(a)}}{=}
\lim_{m\rightarrow\infty} \frac{1}{2 \, \mathcal{B}(a) \;
\rho'_m(1)
\, \int_0^1\lambda_m(x)\mathrm{d}x}\nonumber\\[0.1cm]
&& \hspace*{-0.5cm} \stackrel{\text{(b)}}{=}
\lim_{m\rightarrow\infty} \frac{1-R_m}{2 \, \mathcal{B}(a) \;
\rho'_m(1)
\, \int_0^1\rho_m(x)\mathrm{d}x}\nonumber\\[0.1cm]
&& \hspace*{-0.5cm} \stackrel{\text{(c)}}{=} \frac{1-C}{2 \,
\mathcal{B}(a)} \lim_{m\rightarrow\infty} \frac{1}{\rho'_m(1) \,
\int_0^1\rho_m(x)\mathrm{d}x} \label{eq: first expression for
limit of L_2}
\end{eqnarray}
where (a) relies on \eqref{eq: flatness condition} and \eqref{eq:
connection for degree-2 variable nodes}, (b) follows from
\eqref{design rate of LDPC ensemble} where $R_m$ designates the
design rate of the $m$-th LDPC code ensemble in this sequence,
and~(c) follows by the assumption that the sequence is
capacity-achieving. Let $a_{\text{R}}^{(m)}$ designate the average
right degree of the LDPC code ensemble $(n_m, \lambda_n, \rho_m)$.
From \eqref{eq: average right degree}, this implies that $
a_{\text{R}}^{(m)} = \bigl(\int_0^1 \rho_m(x) \, \text{d}x
\bigr)^{-1} $ and, from Theorem~\ref{Theorem: Lower bound on right
degree} followed by Discussion~\ref{discussion: extension of Theorem
1 to LDPC ensembles}, the asymptotic average right degree of the
considered capacity-achieving sequence tends to infinity, i.e.,
\begin{equation}
\lim_{m \rightarrow \infty} a_{\text{R}}^{(m)} = \infty.
\label{eq: limit of average right-degree for c.a. sequence}
\end{equation}
We evaluate now the expression in \eqref{eq: first expression for
limit of L_2}. To this end, let $\rho_m(x)
\triangleq\sum_i\rho^{(m)}_i x^{i-1}$, and let $\Gamma_i^{(m)}$
designate the fraction of parity-check nodes of degree $i$ for
LDPC code ensemble $(n_m, \lambda_m, \rho_m)$, then
\begin{eqnarray}
&& \rho_m'(1) \, \int_0^1\rho_m(x)\mathrm{d}x \nonumber\\
&& = \frac{\rho_m'(1)}{a_{\text{R}}^{(m)}} \nonumber\\
&& = \frac{\sum_i (i-1) \rho_i^{(m)}}{a_{\text{R}}^{(m)}}
\nonumber\\
&& = \frac{\sum_i i \rho_i^{(m)} - 1}{a_{\text{R}}^{(m)}}
\nonumber\\
&& = \frac{1}{a_{\text{R}}^{(m)}} \left(\sum_i i \biggl(\frac{i
\Gamma_i^{(m)}}{\sum_j j \Gamma_j^{(m)}} \biggr) - 1 \right) \nonumber \\
&& = \frac{\sum_i i^2
\Gamma_i^{(m)}}{\bigl(a_{\text{R}}^{(m)}\bigr)^2} -
\frac{1}{a_{\text{R}}^{(m)}} \nonumber\\
&& = \left(\frac{\sqrt{\sum_i i^2 \Gamma_i^{(m)} -
\bigl(a_{\text{R}}^{(m)} \bigr)^2} }{a_{\text{R}}^{(m)}} \right)^2
+ 1 - \frac{1}{a_{\text{R}}^{(m)}}. \label{eq: chain of
equalities}
\end{eqnarray}
Consider any code from the LDPC code ensemble $(n_m, \lambda_m,
\rho_m)$. Note that the first term in \eqref{eq: chain of
equalities} is the square of the ratio of the standard deviation
and the average degree of the parity-check nodes for this code.
Since we denote the asymptotic limit of this ratio by $K$ (where
we assume that it exists and is finite) and also \eqref{eq: limit
of average right-degree for c.a. sequence} holds, then we get from
\eqref{eq: chain of equalities} that
\begin{equation}
\lim_{m \rightarrow \infty} \rho_m'(1) \,
\int_0^1\rho_m(x)\mathrm{d}x = K^2 + 1. \label{eq: asymptptic
limit}
\end{equation}
This completes the proof of the theorem by combining \eqref{eq:
first expression for limit of L_2} with \eqref{eq: asymptptic
limit}.

\section{Proof of Lemma~\ref{lemma: relation between the capacity and
Bhattacharyya constant of an MBIOS channel}} \label{Appendix:
inequality related to the Bhattacharyya constant and channel
capacity} Let $a$ denote the symmetric $L$-density {\em pdf} of the
transition probability of an MBIOS channel (see
\cite[Theorem~4.26]{RiU_book}). Let $C = C(a)$ and $B =
\mathcal{B}(a)$ be the corresponding capacity and Bhattacharyya
constant, respectively. From \eqref{eq: channel capacity of an MBIOS
channel}, \eqref{eq: definition of Bhattacharyya constant} and the
symmetry of $a$
\begin{eqnarray*}
&& C + B - 1 \\
&& = \int_{-\infty}^{\infty} a(l) e^{-\frac{l}{2}} \, \mathrm{d}l
- \int_{-\infty}^{\infty} a(l) \log_2(1+e^{-l}) \, \mathrm{d}l \\
&& = \int_{-\infty}^{\infty} a(l) e^{-\frac{l}{2}} \, \mathrm{d}l
\\ && \hspace*{0.3cm} - \frac{1}{2} \int_{-\infty}^{\infty} \Bigl[ a(l) \log_2(1+e^{-l})
+ a(-l) \log_2(1+e^{l}) \Bigr] \mathrm{d}l \\
&& = \int_{-\infty}^{\infty} a(l) e^{-\frac{l}{2}} \, \mathrm{d}l
\\ && \hspace*{0.3cm} - \frac{1}{2} \int_{-\infty}^{\infty} a(l) \Bigl[\log_2(1+e^{-l})
+ e^{-l} \log_2(1+e^{l}) \Bigr] \mathrm{d}l \\
&& = \int_{-\infty}^{\infty} e^{-\frac{l}{2}} \, a(l) g(l)
\mathrm{d}l
\end{eqnarray*}
where the function $g$ is given by
\begin{equation*}
g(l) = 1 - \frac{1}{2} \Bigl[ e^{\frac{l}{2}} \log_2(1+e^{-l}) +
e^{-\frac{l}{2}} \log_2(1+e^{l}) \Bigr], \quad l \in \reals.
% \label{eq: g}
\end{equation*}
In order to complete the proof, it suffices to show that the
function $g$ is non-negative. The substitution $x =\frac{1}{1+e^l}$
gives $ g(l) = 1-\frac{h_2(x)}{2\sqrt{x(1-x)}} $ where the interval
$(-\infty, +\infty)$ for $l$ is mapped into the interval $(0,1)$ for
$x$. The non-negativity of $g$ follows from the inequality $ h_2(x)
\leq 2 \sqrt{x(1-x)} $ which is satisfied for $0 \leq x \leq 1$. The
non-negativity of the function $g$ implies that $C+B \geq 1$.

Note that for a BEC with erasure probability $p$, the channel
capacity is $1-p$ bits per channel use, and the Bhattacharyya
constant is equal to $p$. Hence, the equality $C+B=1$ holds for
every BEC, irrespectively of the channel erasure probability.

\section{Proof of Corollary~\ref{Corollary: The fraction of edges connected to
degree-2 variable nodes}} \label{Appendix: Proof of the corollary
on lambda_2} A truncation of the power series on the LHS of
\eqref{eq: power series for 1-h_2} after its first term gives the
inequality
\begin{equation*}
1-h_2\left(\frac{1-\sqrt{u}}{2}\right) \geq \frac{u}{2\,\ln
2}\,,\quad 0\leq u\leq1.
\end{equation*}
Assigning $u = \bigl(1-2h_2^{-1}(x)\bigr)^2$ and rearranging terms
gives
\begin{equation}
h_2^{-1}(x) \geq \frac{1}{2} \left(1-\sqrt{2\ln2 \; \;
(1-x)}\right), \quad 0 \leq x \leq 1. \label{eq: lower bound on
h_2^-1(x)}
\end{equation}
Assigning $0\leq x \triangleq \frac{1-C}{1-(1-\varepsilon)C}\leq
1$ in \eqref{eq: lower bound on h_2^-1(x)} gives
\begin{eqnarray*}
&& h_2^{-1}\left(\frac{1-C}{1-C(1-\varepsilon)}\right) \nonumber
\\ && \geq \frac{1}{2} \left(1-\sqrt{2\ln2 \; \left(\frac{\varepsilon
C }{1-(1-\varepsilon)C}\right)}
\right)\nonumber\\
&& \geq \frac{1}{2} \left(1-\sqrt{2\ln2 \; \left(\frac{\varepsilon
C}{1-C}\right)}\right)
\end{eqnarray*}
and therefore
\begin{equation}
1-2h_2^{-1}\left(\frac{1-C}{1-(1-\varepsilon)C}\right) \leq
\sqrt{2\ln2 \; \left(\frac{\varepsilon C}{1-C}\right)}\,.
\label{eq: inequality involving 1-2h_2^-1}
\end{equation}
Substituting \eqref{eq: inequality involving 1-2h_2^-1} in
\eqref{eq: lower bound on a_R} provides the following lower bound
on the average right degree of the ensembles:
\begin{equation}
a_{\text{R}} \geq \frac{\ln\left(\frac{1}{2\,\ln
2}\frac{1-C}{\varepsilon
C}\right)}{\ln\left(\frac{1}{g_1}\right)}. \label{eq: looser lower
bound on a_R}
\end{equation}
As the average right degree of an LDPC code ensemble is not less
than~2 (as otherwise, some bits are forced to be zeros and can be
deleted from all codewords), then it follows from \eqref{eq:
looser lower bound on a_R} that
\begin{eqnarray}
a_{\text{R}} - 1 &\geq& \left[\frac{\ln\left(\frac{g_1}{2\,\ln2}
\frac{1-C}{\varepsilon C}\right)}{\ln
\left(\frac{1}{g_1}\right)}\right]^+\nonumber\\
&=& \left[\frac{\ln\left(\frac{g_1}{2\,\ln2}\frac{1-C}{C}\right) +
\ln \left(\frac{1}{\varepsilon}\right)}
{\ln\left(\frac{1}{g_1}\right)}\right]^+\,. \label{eq: looser
lower bound on a_R-1}
\end{eqnarray}
The proof is completed by combining \eqref{eq: upper bound on
lambda_2 in terms of a_R} with \eqref{eq: looser lower bound on
a_R-1}.

\section{Proof of Proposition~\ref{Proposition: the tightness of the upper bound on
lambda_2}} \label{Appendix: Proof of the proposition on the
tightness of the upper bound on lambda_2}

When the transmission takes place over a BEC whose erasure
probability is $p$, the constant $c_2$ given in \eqref{eq: c_2}
takes the form
\begin{equation}
c_2 = \frac{p}{\ln\left(\frac{1}{1-p}\right)} \; .
  \label{eq: c_2 for BEC}
\end{equation}
The starting point of this proof follows the concept in
\cite[Example~3.88]{RiU_book}, and its continuation relies on the
analysis used for the proof of \cite[Theorem~2.3]{Sason-it03}. For
$0<\alpha<1$, let
\begin{eqnarray}
\label{eq: general lambda distribution for Shokrollahi ensembles}
&& \hspace*{-.8cm} \hat{\lambda}_{\alpha}(x) = 1-(1-x)^{\alpha} =
\sum_{k=1}^{\infty}(-1)^{k+1} {\alpha \choose k} x^k\,,\quad 0\leq
x\leq 1 \nonumber \\
&& \hspace*{-.8cm} \rho_{\alpha}(x) = x^{\frac{1}{\alpha}}.
\end{eqnarray}
Note that all the coefficients in the power series expansion of
$\hat{\lambda}_{\alpha}$ are positive for all $0<\alpha<1$. Let us
now define the polynomials $\hat{\lambda}_{\alpha}^{(N)}$ and
$\lambda_{\alpha}^{(N)}$ where $\hat{\lambda}_{\alpha}^{(N)}(x)$
is the truncated power series of $\hat{\lambda}_{\alpha}(x)$
around $x=0$, consisting of all the terms up to (and including)
the term $x^{N-1}$, and the polynomial
\begin{equation}
\lambda_{\alpha}^{(N)}(x) \triangleq
\frac{\hat{\lambda}_{\alpha}^{(N)}(x)}{\hat{\lambda}_{\alpha}^{(N)}(1)}
  \label{eq: lambda_alpha,n}
\end{equation}
is normalized to satisfy the equality $\lambda_{\alpha}^{(N)}(1) =
1$. The sequence of right-regular LDPC code ensembles in
\cite{Shokrollahi-IMA2000} is of the form
$\big\{\big(n_m,\lambda_{\alpha}^{(N)}(x),\rho_{\alpha}(x)\big)\big\}_{m\geq
1}$ where $0<\alpha<1$ and $N\in\naturals$ are arbitrary
parameters which need to be selected properly. Assume that the
transmission takes place over a BEC whose erasure probability
is~$p$. Based on the proof of \cite[Theorem~2.3]{Sason-it03}, this
sequence achieves a fraction $1-\varepsilon$ of the capacity of
the BEC with vanishing bit erasure probability under BP decoding
when $\alpha$ and $N$ are chosen to satisfy
\begin{eqnarray}
&& \hspace*{-1.4cm} \frac{1}{N^{\alpha}} = 1-p \label{eq: N^-alpha = 1-p} \\
&& \hspace*{-1.4cm} N =
\max\biggl(\left\lceil\frac{1-(1-\varepsilon)(1-p)k_2(p)}{\varepsilon}
\right\rceil, \left\lceil(1-p)^{-\frac{1}{p}} \right\rceil \biggr)
\label{eq: N for right-regular ensembles}
\end{eqnarray}
where
\begin{equation}
k_2(p) \triangleq (1-p)^{\frac{\pi^2}{6}}\;
e^{\left(\frac{\pi^2}{6}-\gamma\right)\,p}  \label{eq: k2}
\end{equation}
and $\gamma$ is Euler's constant ($\gamma \approx 0.5772$).
Combining \eqref{eq: general lambda distribution for Shokrollahi
ensembles} and \eqref{eq: lambda_alpha,n}, and using the equality
\begin{equation*}
  \sum_{k=1}^{N-1}(-1)^{k+1} {\alpha \choose k} =
  1-\frac{N}{\alpha}{\alpha \choose N}(-1)^{N+1}
\end{equation*}
gives
\begin{equation*}
  \lambda_{\alpha}^{(N)}(x) = \frac{\sum_{k=1}^{N-1}(-1)^{k+1}{\alpha \choose k}
  x^k}{1-\frac{N}{\alpha}\,(-1)^{N+1}\, {\alpha \choose N}}\,.
\end{equation*}
Therefore, the fraction of edges adjacent to variable nodes of
degree two is given by
\begin{equation}
\lambda_2 = \frac{\alpha}{1-\frac{N}{\alpha}\,(-1)^{N+1}\, {\alpha
\choose N}}\,.
  \label{eq: lambda_2 as function of alpha and N}
\end{equation}
We now obtain upper and lower bounds on $\lambda_2$. From
\cite[Eq.~(67)]{Sason-it03} we have that
\begin{equation}
\frac{c(\alpha,N)}{N^{\alpha}} <
\frac{N}{\alpha}(-1)^{N+1}\,{\alpha \choose
N}\leq\frac{1}{N^{\alpha}}
  \label{eq: lower and upper bounds on -1^N+1 alpha choose N}
\end{equation}
where
\begin{equation}
c(\alpha,N) \triangleq (1-\alpha)^{\frac{\pi^2}{6}}\;
e^{\alpha\big(\frac{\pi^2}{6}-\gamma+\frac{1}{2N}\big)}\,.
\label{eq: c(alpha,N)}
\end{equation}
Substituting \eqref{eq: lower and upper bounds on -1^N+1 alpha
choose N} in \eqref{eq: lambda_2 as function of alpha and N} and
using \eqref{eq: N^-alpha = 1-p}, we get
\begin{equation}
  \frac{\alpha}{1-c(\alpha,N)\,(1-p)} < \lambda_2 \leq
  \frac{\alpha}{1-(1-p)} = \frac{\alpha}{p}\,.
  \label{eq: lower and upper bounds on lambda_2}
\end{equation}
Under the parameter assignments in \eqref{eq: N^-alpha = 1-p} and
\eqref{eq: N for right-regular ensembles}, the parameters $N$ and
$\alpha$ satisfy
\begin{eqnarray}
\label{eq: alpha as function of p and N}
  && \alpha = \frac{\ln\left(\frac{1}{1-p}\right)}{\ln N}\\
  \label{eq: lower bound on N}
  && N \geq \frac{1-(1-p) \, k_2(p)}{\varepsilon} \,.
\end{eqnarray}
Substituting \eqref{eq: alpha as function of p and N} and
\eqref{eq: lower bound on N} into the inequality on the RHS of
\eqref{eq: lower and upper bounds on lambda_2} gives an upper
bound on $\lambda_2$ which takes the form
\begin{eqnarray}
&& \lambda_2 \leq \frac{\alpha}{p}\nonumber\\
&& \hspace*{0.5cm} \leq \frac{\ln\left(\frac{1}{1-p}\right)}{p\,
\ln\left(\frac{1-(1-p) \, k_2(p)}{\varepsilon}\right)}\nonumber\\
%&& \hspace*{0.5cm} = \frac{\ln\left(\frac{1}{1-p}\right)}{p
%\Bigl[\ln \frac{1}{\varepsilon} + \ln\big(1-(1-p) \,
%k_2(p)\big)\Bigr]}
%\nonumber\\
&& \hspace*{0.5cm} = \frac{1}{c_3 + c_2\ln \frac{1}{\varepsilon}}
  \label{eq: upper bound on lambda_2}
\end{eqnarray}
where $c_2$ is the coefficient of the logarithmic growth rate in
$\frac{1}{\varepsilon}$, which coincides here with \eqref{eq: c_2
for BEC}, and
\begin{equation}
c_3 \triangleq \frac{p\,\ln\bigl(1-(1-p) \,
k_2(p)\bigr)}{\ln\left(\frac{1}{1-p}\right)} \label{eq: c_3}
\end{equation}
is a constant which only depends on the BEC. We turn now to derive
a lower bound on $\lambda_2$, and then examine it in the limit
where the gap to capacity vanishes. From \eqref{eq: N for
right-regular ensembles}, we have that for small enough values of
$\varepsilon$, the parameter $N$ satisfies
\begin{eqnarray}
&& \hspace*{-1.2cm} N =
\left\lceil\frac{1-k_2(p)\,(1-p)\,(1-\varepsilon)}{\varepsilon}\right\rceil\nonumber\\
&& \hspace*{-0.7cm} \leq
\frac{1-k_2(p)\,(1-p)\,(1-\varepsilon)}{\varepsilon} + 1.
\label{eq: upper bound on N for small epsilon}
\end{eqnarray}
Substituting \eqref{eq: alpha as function of p and N} and
\eqref{eq: upper bound on N for small epsilon} into the inequality
on the LHS of \eqref{eq: lower and upper bounds on lambda_2}, we
get
\begin{eqnarray}
&& \hspace*{-0.9cm} \lambda_2 >
\frac{\alpha}{p}\;\frac{p}{1-c(\alpha,N)\,(1-p)}\nonumber\\
&& \hspace*{-0.4cm} \geq
\frac{\ln\Bigl(\frac{1}{1-p}\Bigr)}{p\,\ln\left(\frac{1-k_2(p)
(1-p) (1-\varepsilon)+\varepsilon}{\varepsilon} \right)}
\cdot \frac{p}{1-c(\alpha,N)\,(1-p)}\nonumber\\
&& \hspace*{-0.4cm} =\frac{1}{c_3 +
c_2\ln\left(\frac{1}{\varepsilon}\right) +
\widetilde{\varepsilon}(\varepsilon,p)} \; \; \frac{p}{1-(1-p) \,
c(\alpha,N)} \label{eq: lower bound on lambda_2}
\end{eqnarray}
where $c_2$ is the coefficient of the logarithm in the denominator
of \eqref{eq: loosened version of the upper bound on lambda_2} and
it coincides with \eqref{eq: c_2 for BEC} for the BEC, $c_3$ is
given in \eqref{eq: c_3}, and
\begin{equation*}
\widetilde{\varepsilon}(\varepsilon,p) \triangleq \frac{p
\ln\biggl(1+\frac{\varepsilon \bigl(1+k_2(p) \,
(1-p)\bigr)}{1-k_2(p) \, (1-p)
}\biggr)}{\ln\left(\frac{1}{1-p}\right)}
\end{equation*}
which therefore implies that for $0 \leq p < 1$
\begin{equation}
\lim_{\varepsilon \rightarrow 0}
\widetilde{\varepsilon}(\varepsilon,p) = 0. \label{eq: limit of
epsilon tilde}
\end{equation}
Using the lower bound on the parameter $N$ in \eqref{eq: lower bound
on N}, in the limit where $\varepsilon$ tends to zero, the parameter
$N$ tends to infinity (since $1-(1-p) k_2(p)>0$ for all $0<p<1$
where $k_2$ in introduced in \eqref{eq: k2}). Also, from \eqref{eq:
N for right-regular ensembles} and \eqref{eq: alpha as function of p
and N}, we get
\begin{equation*}
\lim_{\varepsilon\rightarrow 0}\alpha = 0
\end{equation*}
which, from \eqref{eq: c(alpha,N)}, yields that
\begin{equation}
  \lim_{\varepsilon\rightarrow 0} c(\alpha,N) = 1\,.
  \label{eq: limit of c(alpha,N)}
\end{equation}
Substituting \eqref{eq: limit of epsilon tilde} and \eqref{eq:
limit of c(alpha,N)} in \eqref{eq: lower bound on lambda_2} yields
that in the limit where the gap to capacity vanishes (i.e.,
$\varepsilon \rightarrow 0$), the upper and lower bounds on
$\lambda_2$ in \eqref{eq: upper bound on lambda_2} and \eqref{eq:
lower bound on lambda_2} coincide. Specifically, we have shown
that
\begin{equation*}
\lim_{\varepsilon\rightarrow0}\lambda_2(\varepsilon)\cdot
c_2\,\ln\left(\frac{1}{\varepsilon}\right) = 1\,.
\end{equation*}
Therefore, as $\varepsilon\rightarrow0$, the upper bound on
$\lambda_2 = \lambda_2(\varepsilon)$ in Corollary~\ref{Corollary:
The fraction of edges connected to degree-2 variable nodes}
becomes tight for the sequence of right-regular LDPC code
ensembles in \cite{Shokrollahi-IMA2000} with the parameters chosen
in \eqref{eq: N^-alpha = 1-p} and \eqref{eq: N for right-regular
ensembles}. We note that the setting of the parameters $N$ and
$\alpha$ in \eqref{eq: N^-alpha = 1-p} and \eqref{eq: N for
right-regular ensembles} is identical to
\cite[p.~1615]{Sason-it03}.

\section{A proof of Remark~\ref{Remark: A discussion on the
constraints given in the LP1 and LP2 bounds}} \label{Appendix: Proof
for the un-necessity of adding the additional constraint in the LP1
and LP2 bounds} We prove in the following the claim in
Remark~\ref{Remark: A discussion on the constraints given in the LP1
and LP2 bounds} which states that adding the constraint that is
imposed by the lower bound on the average right degree (i.e., the
lower bound on $a_{\text{R}} = \sum_{i=1}^{\infty} i \Gamma_i$) does
not affect the LP1 and LP2 bounds introduced in Section~\ref{LP
bounds on degree distributions of LDPC code ensembles}. More
explicitly, for the LP1 bound, we prove that the constraint on
$\{\Gamma_i\}_{i \geq 1}$ which is imposed by \eqref{eq:
relationship between Gamma epsilon and P_b} implies the lower bound
on the average right degree as given in \eqref{eq: lower bound on
a_R with finite P_b} and \eqref{delta}.
\newline
\begin{proof}
Eq.~\eqref{eq: relationship between Gamma epsilon and P_b} gives
the first constraint in the LP1 bound. By substituting $x =
\frac{1-g_1^{\frac{i}{2}}}{2} $ in \eqref{eq: power series for
h_2}, we get that the following equality holds for $i \geq 1$
(note that since $0 \leq g_1 \leq 1$ then $0 \leq x \leq 1$ as
required in \eqref{eq: power series for h_2}):
\begin{equation*}
1 - h_2\biggl(\frac{1-g_1^{\frac{i}{2}}}{2}\biggr) = \frac{1}{2
\ln 2} \sum_{p=1}^{\infty} \frac{g_1^{pi}}{p(2p-1)}.
\end{equation*}
Plugging this equality into the LHS of \eqref{eq: relationship
between Gamma epsilon and P_b} gives
\begin{eqnarray}
&& \sum_{i=1}^{\infty}
\left\{\left[1-h_2\biggl(\frac{1-g_1^{\frac{i}{2}}}{2}\biggr)\right]
\Gamma_i \right\}
\nonumber\\
&& \stackrel{(a)}{=} \frac{1}{2 \ln 2} \sum_{p=1}^{\infty} \sum_{i=1}^{\infty} \frac{\Gamma_i g_1^{pi}}{p(2p-1)} \nonumber\\
&& \stackrel{(b)}{\geq} \frac{1}{2 \ln 2} \sum_{p=1}^{\infty} \frac{{g_1}^{p \sum_i i \Gamma_i}}{p(2p-1)} \nonumber\\
&& \stackrel{(c)}{=} \frac{1}{2 \ln 2} \sum_{p=1}^{\infty} \frac{{g_1}^{p a_{\text{R}}}}{p(2p-1)} \nonumber\\
&& \stackrel{(d)}{=}
1-h_2\biggl(\frac{1-g_1^{\frac{a_{\text{R}}}{2}}}{2} \biggr).
\label{eq: lower bound on the LHS}
\end{eqnarray}
where equality~(a) is obtained by interchanging the order of
summation, equality~(b) follows from Jensen's inequality,
equality~(c) follows from expressing the average right degree by
the equality $a_{\text{R}} = \sum_i i \Gamma_i$, and equality~(d)
follows from \eqref{eq: power series for h_2}. Combining
\eqref{eq: relationship between Gamma epsilon and P_b} with
\eqref{eq: lower bound on the LHS} gives that
\begin{equation*}
1-h_2\biggl(\frac{1-g_1^{\frac{a_{\text{R}}}{2}}}{2} \biggr) \leq
\frac{\varepsilon\,C + h_2(P_{\text{b}})}{1-(1-\varepsilon)C}
\end{equation*}
and then some straightforward algebra implies that
\begin{equation*}
a_{\text{R}} \geq \frac{2 \ln
\left(\frac{1}{1-2h_2^{-1}(\frac{1-C-h_2(P_{\text{b}})}{1-(1-\varepsilon)C})}
\right) }{\ln \left(\frac{1}{g_1}\right)}.
\end{equation*}
This lower bound on the average right degree coincides with the
bound in \eqref{eq: lower bound on a_R with finite P_b} and
\eqref{delta} which then completes our proof for the LP1 bound.
The same proof holds for the LP2 bound while referring to the
lower bound given in \eqref{delta} and \eqref{eq: universal lower
bound on a_R with finite P_b}.
\end{proof}

\section{Analytical Solution of the LP1 Bound}
\label{Appendix: solution of the LP1 bound} The LP1 bound in
Section~\ref{LP bounds on degree distributions of LDPC code
ensembles} can be equivalently expressed as the following
minimization problem:
\begin{equation}
\mbox{\fbox{$
\begin{array}{l}
\text{minimize} \; \; -\sum\limits_{i=1}^k \rho_i, \quad k=1, 2, \ldots \nonumber \\
\text{subject to} \nonumber \\[0.1cm]
\begin{cases}
\hspace{0.2cm} \sum\limits_{i=1}^{\infty} d_i \rho_i \leq 0 \\
\hspace{0.2cm} d_i \triangleq \frac{1}{i} \left[
1-h_2\biggl(\frac{1-g_1^{\frac{i}{2}}}{2}\biggr) -
\frac{\varepsilon\,C + h_2(P_{\text{b}})}{1-(1-\varepsilon)C}
\right], \; i \geq 1 \\[0.15cm]
\hspace{0.2cm} \sum\limits_{i=1}^{\infty} \rho_i \leq 1 \\[0.15cm]
\hspace{0.2cm} \rho_i \geq 0, \quad i=1, 2, \ldots
\end{cases}
\end{array}
$}}
\end{equation}
where we negated the objective function and turned the maximization
into a minimization, and also the equality constraint on $\sum_{i
\geq 1} \rho_i$ was turned into an inequality constraint. By
introducing the non-negative Lagrange multipliers $\mu_1$ and
$\mu_2$, respectively, to the first and second inequality
constraints, and also introducing the non-negative Lagrange
multiplies $\{\theta_i\}$ to the non-negativity constraint on
$\{\rho_i\}$, we get the Lagrangian
\begin{eqnarray}
&& \hspace*{-0.7cm} L(\{\rho_i\}, \mu_1, \mu_2, \{\theta_i\}) \nonumber \\
&& \hspace*{-0.7cm} = -\sum\limits_{i=1}^k \rho_i + \mu_1
\sum\limits_{i=1}^{\infty} d_i \rho_i + \mu_2
\biggl(\sum\limits_{i=1}^{\infty} \rho_i-1 \biggr) -
\sum\limits_{i=1}^{\infty} \theta_i \rho_i \nonumber \\
&& \hspace*{-0.7cm} = \sum\limits_{i=1}^{k} \bigl(-1+\mu_1 d_i +
\mu_2 - \theta_i \bigr) \rho_i + \sum\limits_{i=k+1}^{\infty}
\bigl(\mu_1 d_i + \mu_2 - \theta_i \bigr) \rho_i \nonumber
\\ && - \mu_2. \label{eq: Lagrangian for LP1}
\end{eqnarray}
By alternating again the sign of the objective function, we get the
following dual LP problem:
\begin{equation}
\mbox{\fbox{$
\begin{array}{l}
\text{minimize} \; \mu_2 \nonumber \\
\text{subject to} \nonumber \\[0.1cm]
\begin{cases}
\hspace{0.2cm} -1+\mu_1 d_i + \mu_2 - \theta_i = 0, \quad
i=1,2,\ldots, k \\
\hspace{0.2cm} \mu_1 d_i + \mu_2 - \theta_i = 0, \quad
i=k+1, k+2,\ldots \\
\hspace{0.2cm} \mu_1, \mu_2 \geq 0 \\
\hspace{0.2cm} \theta_i \geq 0, \quad i=1, 2, \ldots
\end{cases}
\end{array}
$}}
\end{equation}
Strong duality holds for linear programming provided that the primal
LP or its dual LP are feasible (see \cite[Problem~5.23]{CVX_book}).
Hence, strong duality holds for the LP1 problem.

Note that the sequence $\{d_i\}$ (see the above primal problem) is
positive if and only if $i < k_0$ where $k_0$ denotes the lower
bound on the average right degree as is given in \eqref{eq: lower
bound on a_R with finite P_b}. For $k < k_0$, the sequence
$\{d_i\}_{i=1}^{k}$ is positive and monotonic decreasing:
\begin{equation*}
d_1 > d_2 > \ldots, > d_k > 0, \quad \forall \, k < k_0.
\end{equation*}
Also $d_i \leq 0$ for $i \geq k_0$, and $\lim_{i \rightarrow \infty}
d_i = 0$. Let
\begin{equation}
d^* \triangleq \min_{i \geq 1} d_i \label{eq: d^*}
\end{equation}
where the minimum of the sequence $\{d_i\}$ is attained for some
index $i \geq k_0$, and $d^* \leq 0$ (note that except for the
degenerate case where $g_1=0$, for which the channel is completely
useless, the sequence $\{d_i\}$ is negative for $i>k_0$, and it
tends asymptotically to zero in the limit where $i \rightarrow
\infty$).

Let $k < k_0$. Due to the properties of the sequence $\{d_i\}$ and
the non-negativity constraint on $\{\theta_i\}$ in the dual LP
problem, the minimization of the objective function $(\mu_2)$ can be
simplified. To this end, one can remove all the equality constraints
from the dual LP problem except of the first equality constraint
with the index $i=k$, and the second equality constraint with the
index $i \geq k_0$ for which the sequence $\{d_i\}$ attains its
minimal value $(d^*)$. For these two indices of $i$, the Lagrange
multipliers $\theta_i$ in the two equality constraints of the dual
LP problem are set to zero; this setting attains the minimal value
of $\mu_2$ (for the other equality constraints that were removed
from the dual LP problem, the corresponding $\theta_i$'s are
strictly positive; however, these equality constraints are redundant
for the minimization of $\mu_2$ in the dual LP). Hence, for $k <
k_0$, the dual LP problem is simplified to
\begin{equation}
\mbox{\fbox{$
\begin{array}{l}
\text{minimize} \; \mu_2 \nonumber \\
\text{subject to} \nonumber \\[0.1cm]
\begin{cases}
\hspace{0.2cm} -1+\mu_1 d_k + \mu_2 = 0 \\
\hspace{0.2cm} \mu_1 d^* + \mu_2 = 0 \\
\hspace{0.2cm} \mu_1, \mu_2 \geq 0
\end{cases}
\end{array}
$}}
\end{equation}
whose solution is
\begin{equation*}
\mu_1 = \frac{1}{d_k-d^*}, \quad \mu_2 = -\frac{d^*}{d_k-d^*}
\end{equation*}
and the optimal value of the dual LP is equal to
$-\frac{d^*}{d_k-d^*}$ which is indeed bounded between~0 and~1
(since $d^* \leq 0$ and $d_k > 0$ for $k < k_0$).

For $k \geq k_0$, we get the following system of inequalities from
the dual LP problem:
\begin{eqnarray*}
\left\{
\begin{array}{ll}
-1+\mu_1 d_i + \mu_2 \geq 0, \quad & \mbox{for $i=1,2, \ldots, k$}\\
\mu_1 d_i + \mu_2 \geq 0, \quad & \mbox{for $i=k+1, k+2, \ldots$}
\end{array}
\right.
\end{eqnarray*}
Since $d_k \leq 0$, then the optimal solution of the dual LP is
obtained at $\mu_1=0$ and $\mu_2=1$, which then gives an optimal
value of~1 for the minimization of $\mu_2$.

\begin{remark}: Consider again the solution of the LP1 problem in the
case where $k \leq k_0$. From the solution of the dual problem, it
follows that it is obtained by setting $\rho_i$ to be zero, except
for two indices. To this end, let $i=l$ be the index for which the
sequence $\{d_i\}$ achieves its negative minimal value $(d^*)$, and
let us choose the values of $\rho_k$ and $\rho_l$ to satisfy the two
equalities:
\begin{eqnarray*}
&& d_k \rho_k + d_l \rho_l = 0 \\[0.1cm]
&& \rho_k + \rho_l = 1.
\end{eqnarray*}
Since $d^* = d_l$ for some $l>k_0$, then for $k \leq k_0$ and the
above selection of $\{\rho_i\}$
\begin{equation*}
\sum_{i=1}^k \rho_i = \rho_k = -\frac{d^*}{d_k - d^*}
\end{equation*}
which indeed coincides with the solution of the dual problem.
\end{remark}

\subsection*{A Tightened Version of the LP1 bound for the BEC and its Analytical Solution}
A tightened version of the LP1 bound for the BEC is obtained via
\eqref{eq: relationship between Gamma epsilon and P_b for the BEC}.
By substituting the equality \eqref{eq: switching between
representations_2} into the LHS of \eqref{eq: relationship between
Gamma epsilon and P_b for the BEC} and using the equality $C=1-p$
for a BEC gives
\begin{equation*}
\sum_{i=1}^{\infty} \frac{\rho_i C^i}{i} \leq \frac{\varepsilon C +
P_{\text{b}}}{1-(1-\varepsilon)C} \sum_{i=1}^{\infty}
\frac{\rho_i}{i} \, .
\end{equation*}
This inequality constraint forms a tightened constraint for the BEC,
as compared to the first inequality constraint which was formulated
in the LP1 problem for a general MBIOS channel. In order to use the
analytical result derived earlier in this appendix and adapt it to
this case, we formulate the tightened version of the LP1 bound for
the BEC as follows:
\begin{equation}
\mbox{\fbox{$
\begin{array}{l}
\text{minimize} \; \; -\sum\limits_{i=1}^k \rho_i, \quad k=1, 2, \ldots \nonumber \\
\text{subject to} \nonumber \\[0.1cm]
\begin{cases}
\hspace{0.2cm} \sum\limits_{i=1}^{\infty} d_i \rho_i \leq 0 \\
\hspace{0.2cm} d_i \triangleq \frac{1}{i} \left(C^i -
\frac{\varepsilon C + P_{\text{b}}}{1-(1-\varepsilon)C}
\right), \; \; i=1,2,\ldots \\[0.15cm]
\hspace{0.2cm} \sum\limits_{i=1}^{\infty} \rho_i \leq 1 \\[0.15cm]
\hspace{0.2cm} \rho_i \geq 0, \quad i=1, 2, \ldots
\end{cases}
\end{array}
$}}
\end{equation}
Similarly to the above analysis in this appendix, the new sequence
$\{d_i\}$ is non-negative if and only if $i \leq k_0$ where $k_0$
denotes the lower bound on the average right degree as is given in
\eqref{eq: lower bound on a_R for the BEC with finite P_b}. For $k
\leq k_0$, the sequence $\{d_i\}_{i=1}^{k}$ is non-negative and
monotonic decreasing; moreover, $d_i < 0$ for $i
> k_0$, and $\lim_{i \rightarrow \infty} d_i = 0$. By using the
same notation of $d^*$ in \eqref{eq: d^*}, we obtain that the
tightened version of the LP1 bound for the BEC has the same
analytical solution as of the general LP1 bound, except for the
change of the sequence $\{d_i\}$ (and its corresponding minima
$d^*$).

\subsection*{Acknowledgment} Discussions with
Henry Pfister, Tom Richardson and Ruediger Urbanke during the 2008
Turbo Coding Symposium in Lausanne are acknowledged. A discussion
with Gil Wiechman in an early stage of this work is also
acknowledged. The author wishes to thank Boaz Shuval, Moshe Twitto,
and Oren Zeitlin for pointing out typos in a previous draft. Thanks
are due to the anonymous reviewers for their feedback which
contributed to the lucidity of the presentation.

\end{document}